\documentclass[11pt]{article}
\usepackage[utf8]{inputenc}
\usepackage{amsmath,amsfonts,amsthm,fullpage,xcolor,tikz}
\usepackage{hyperref}
\usepackage{url}
\usepackage{soul}
\newtheorem{theorem}{Theorem}
\newtheorem{claim}{Claim}
\newtheorem{proposition}{Proposition}
\newtheorem{lemma}{Lemma}
\newtheorem{fact}{Fact}
\newtheorem{corollary}{Corollary}

\newtheorem{definition}{Definition}
\newtheorem{example}{Example}
\newtheorem{remark}{Remark}
\newcommand{\mkw}[1]{\textcolor{blue}{\textbf{[#1 --mary]}}}
\newcommand{\TODO}[1]{\textcolor{orange}{\textbf{[TODO: #1 ]}}}
\newcommand{\yi}[1]{\textcolor{red}{\textbf{[#1 --yuval]}}}
\newcommand{\victor}[1]{\textcolor{magenta}{\textbf{[#1 --victor]}}}

\newcommand{\newvictor}[1]{{#1}}

\newcommand{\F}{\mathbb{F}}

\newcommand{\Eval}{\mathsf{Eval}}
\newcommand{\Share}{\mathsf{Share}}
\newcommand{\Rec}{\mathsf{Rec}}
\newcommand{\tr}{\mathrm{tr}}

\renewcommand{\vec}[1]{\mathbf{#1}}
\newcommand{\nvec}{\mathrm{Vec}}
\newcommand{\nmat}{\mathrm{Mat}}

\newcommand{\POLY}{\textnormal{\textsf{POLY}}}
\newcommand{\ALL}{\textnormal{\textsf{ALL}}}
\newcommand{\CONCAT}{\textnormal{\textsf{CONCAT}}}

\newcommand{\bx}{\mathbf{x}}
\newcommand{\bX}{\mathbf{X}}
\newcommand{\bz}{\mathbf{z}}
\newcommand{\by}{\mathbf{y}}
\newcommand{\bw}{\mathbf{w}}
\newcommand{\bv}{\mathbf{v}}
\newcommand{\bY}{\mathbf{Y}}
\newcommand{\bZ}{\mathbf{Z}}
\newcommand{\cX}{\mathcal{X}}
\newcommand{\cY}{\mathcal{Y}}
\newcommand{\cF}{\mathcal{F}}
\newcommand{\br}{\mathbf{r}}
\newcommand{\FF}{\mathbb{F}}
\newcommand{\EE}{\mathbb{E}}
\newcommand{\B}{\mathbb{B}}
\newcommand{\eps}{\varepsilon}

\title{On the Download Rate of Homomorphic Secret Sharing\thanks{This is a full version of~\cite{FIKWconf}.}}

\author{
    Ingerid Fosli\thanks{Google.  \texttt{ifosli@gmail.com}. Work done while at Stanford University.}
    \and
    Yuval Ishai\thanks{Technion. {\tt \{yuvali,tkolobov\}@cs.technion.ac.il}. Research partially supported by ERC Project NTSC (742754), ISF grant 2774/20, and BSF grant 2018393.}
    \and 
    Victor I. Kolobov\footnotemark[3]
    \and 
  Mary Wootters\thanks{Stanford University.  \texttt{marykw@stanford.edu}.  Research partially supported by NSF grants CCF-1844628 and CCF-1814629, and by a Sloan Research Fellowship.}
}

\begin{document}
\maketitle

\begin{abstract}
A homomorphic secret sharing (HSS) scheme is a secret sharing scheme that supports evaluating functions on shared secrets by means of a local mapping from input shares to output shares. We initiate the study of the {\em download rate} of HSS, namely, the achievable ratio between the length of the output shares and the output length when amortized over $\ell$ function evaluations. We obtain the following results.
\begin{itemize}
\item In the case of {\em linear} information-theoretic HSS schemes for degree-$d$ multivariate polynomials, we characterize the optimal download rate in terms of the optimal minimal distance of a linear code with related parameters. We further show that for sufficiently large $\ell$ (polynomial in all problem parameters), the optimal rate can be realized using Shamir's scheme, even with secrets over $\F_2$.
\item We present a general rate-amplification technique for HSS that improves the download rate at the cost of requiring more shares. As a corollary, we get high-rate variants of computationally secure HSS schemes and efficient private information retrieval protocols from the literature. 
\item We show that, in some cases, one can beat the best download rate of linear HSS by allowing {\em nonlinear} output reconstruction and $2^{-\Omega(\ell)}$ error probability. 
\end{itemize} 
\end{abstract}
\thispagestyle{empty}	
\pagebreak

\tableofcontents
\thispagestyle{empty}	
\pagebreak

\setcounter{page}{1}

\section{Introduction}
\emph{Homomorphic Secret Sharing} (HSS)~\cite{C:Benaloh86a,BGI16a,BGILT18} is a form of secret sharing that supports computation on the shared data by means of locally computing on the shares. HSS can be viewed as a distributed analogue of homomorphic encryption~\cite{RivestAD78,STOC:Gentry09} that allows for better efficiency and weaker cryptographic assumptions, or even unconditional security. 

More formally, a standard $t$-private (threshold) secret-sharing scheme randomly splits an input $x$ into $k$ shares $x^{(1)}, \ldots, x^{(k)}$, distributed among $k$ servers, so that no $t$ of the servers learn any information about the input. (Here we assume {\em information-theoretic} security by default, but we will later also consider computational security.) 
A secret-sharing scheme as above is an HSS for a function class $\mathcal{F}$ if it additionally allows computation of functions $f\in\mathcal{F}$ on top of the shares.
More concretely,
an HSS scheme $\Pi$ consists of three algorithms, $\Share$, $\Eval$ and $\Rec$.  Given $m$ inputs $x_1, \ldots, x_m$, which we think of as originating from $m$ distinct \emph{input clients}, the randomized $\Share$ function independently splits each input $x_i$ among $k$ servers.  Each server $j$ computes $\Eval$ on its $m$ input shares and a target function $f\in\mathcal{F}$, to obtain an \emph{output share} $y^{(j)}$.  These output shares are then sent to an \emph{output client}, who runs $\Rec(y^{(1)},\ldots,y^{(k)})$ to reconstruct $f(x_1, \ldots, x_m)$.\footnote{One may also consider {\em robust} HSS in which reconstruction can tolerate errors or erasures. While some of our results can be extended to this setting, in this work (as in most of the HSS literature) we only consider the simpler case of non-robust reconstruction. 
}

As described up to this point, the HSS problem admits a trivial solution: simply let $\Eval$ be the identity function (which outputs all input shares along with the description of $f$) and let $\Rec$ first reconstruct the $m$ inputs and then compute $f$. To be useful for applications, HSS schemes are required to be {\em compact}, in the sense that the output shares $y^{(j)}$ are substantially shorter than what is sent in this trivial solution. A strong compactness requirement, which is often used in HSS definitions from the literature, is additive reconstruction. Concretely, in an {\em additive HSS} scheme the output of $f$ is assumed to come from an Abelian group $\mathbb G$, each output share $y^{(j)}$ is in $\mathbb G$, and $\Rec$ simply computes group addition. 
Simple additive HSS schemes for linear functions~\cite{C:Benaloh86a}, finite field multiplication~\cite{STOC:BenGolWig88,STOC:ChaCreDam88,EC:CraDamMau00}, and low-degree multivariate polynomials~\cite{BeaverF90,BeaverFKR90,Chor:1998:PIR:293347.293350} are implicit in classical protocols for information-theoretic secure multiparty computation and private information retrieval. More recently, {\em computationally secure} additive HSS schemes were constructed for a variety of function classes under a variety of cryptographic assumptions~\cite{BoyleGI15,DHRW16,BGI16a,BGI16,FazioGJS17,BGILT18,BoyleKS19,BCGIKS19,CouteauM21,OrlandiSY21,RoyS21}. 

While additive HSS may seem to achieve the best level of compactness one could hope for, allowing for 1-bit output shares when evaluating a Boolean function $f$, it still leaves a factor-$k$ gap between the output length and the total length of the output shares communicated to the output client. This is undesirable when $f$ has a long output, especially when $k$ is big. 
We refer to the total output share length of $\Pi$ as its {\em download cost} and to the ratio between the output length and the download cost as its {\em download rate} or simply {\em rate}.  
We note that even when allowing a bigger number of servers and using a non-additive output encoding, it is not clear how to optimize the rate of existing HSS schemes (see Section~\ref{sec:optimal-linear-hss-intro} for further discussion). 

In the related context of homomorphic encryption, it was recently shown that the download rate, amortized over a long output, can approach 1 at the limit~\cite{DottlingGIMMO19,TCC:GenHal19,TCC:BDGM19}. However, here the concrete download cost must inherently be bigger than a cryptographic security parameter, and the good amortized rate only kicks in for big output lengths that depend polynomially on the security parameter. The relaxed HSS setting has the qualitative advantage of allowing the rate to be independent of any security parameters, in addition to allowing for information-theoretic security and better concrete efficiency. 


In this work, we initiate the systematic study of the download cost of homomorphic secret sharing. We ask the following question:
\begin{quote}
{\em
    How compact can HSS be? Can existing HSS schemes be modified to achieve amortized download rate arbitrary close to $1$, possibly by employing more servers? }
\end{quote}

More concretely, our primary goal is to understand the best download rate attainable given the number of servers $k$, the security threshold $t$ and the class of functions $\mathcal{F}$ that the HSS is guaranteed to work for.  As a secondary goal, we would also like to minimize the overhead to the {\em upload cost}, namely the total length of the input shares.

To help establish tight bounds, we study the download rate when amortized over multiple instances.  That is, given inputs $x_{i,j}$ for $i \in [m], j \in [\ell]$, all shared separately, and functions $f_1, \ldots, f_\ell \in \mathcal{F}$, we consider the problem of computing $f_j(x_{1,j}, x_{2,j}, \ldots, x_{m,j})$ for all $j \in [\ell]$. (Note that positive results in this setting also apply in the easier settings of computing multiple functions on the same inputs or the same function on multiple inputs.)
HSS with a big number of instances $\ell$ can arise in many application scenarios that involve large-scale computations on secret-shared inputs. This includes private information retrieval, private set intersection, private statistical queries, and more. Such applications will motivate specific classes $\mathcal{F}$ we  consider in this work.



\subsection{Contributions}


We develop a framework in which to study the download rate of HSS and obtain both positive and negative results for special cases of interest. In more detail, we make the following contributions.

\subsubsection{Optimal-download linear HSS for low-degree polynomials, and applications}\label{sec:optimal-linear-hss-intro}

We consider information-theoretic HSS when the function class $\mathcal{F}$ is the set of degree-$d$ $m$-variate polynomials over a finite field $\F$. A standard HSS for this class~\cite{BeaverF90,BeaverFKR90} uses $k=dt+1$ servers and has download rate of $1/k$. By using $k\gg dt$ servers, a multi-secret extension of Shamir's secret sharing scheme~\cite{FY92} can be used to get the rate arbitrarily close to $1/d$, for sufficiently large $\ell$.\footnote{Intuitively, this is because one can use polynomials of degree $\approx k/d$ to share the secrets (yielding rate $\approx 1/d$ when $k\gg dt$). 
A higher degree is not possible because the product of $d$ polynomials should have degree $<k$ to allow interpolation.
See Remark~\ref{rem:triv_concat_wont_work} for more details.} In Section~\ref{sec:linear} we give two constructions that obtain a better rate, arbitrarily close to 1. 

Our first construction (Theorem~\ref{thm:cnf-hss}) is based on the highly redundant \emph{CNF sharing}~\cite{ISN89}, where each input is shared by replicating ${k\choose t}$ additive shares. This construction is defined for all choices of $\F,m,d,t,\ell$ and $k>dt$, and its rate is determined by the best minimal distance of a linear code with related parameters. For sufficiently large $\ell$ this code is an MDS code, in which case the rate is $1 - dt/k$.
The main downside of this construction is a ${k\choose t}$ overhead to the upload cost. Settling for computational security, this overhead can be converted into a {\em computational} overhead (which is reasonable in practice for small values of $k,t$) by using a pseudorandom secret sharing technique~\cite{GI99,CDI05}.

Our second construction  (Theorem~\ref{thm:shamir-hss}) uses \emph{Shamir sharing}~\cite{Shamir79}, where each input is shared by evaluating a random degree-$t$ polynomial over an extension field of $\F$ at $k$ distinct points. This construction also achieves a download rate of $1 - dt/k$ for sufficiently large $\ell$, but here this rate is achieved with upload cost that scales polynomially with $t$ and the other parameters.

Both constructions are {\em linear} in the sense that $\Share$ and $\Rec$ are linear functions. In Section~\ref{sec:linearNegative} we show that for such linear HSS schemes, $1 - dt/k$ is the best rate possible, implying optimality of our schemes.  

We compare the above two HSS schemes in Section~\ref{sec:comparison}.  Briefly, the Shamir-based scheme has better upload cost (which scales polynomially with all parameters) but is more restrictive in its paramater regime: that is, it only yields an optimal scheme in a strict subset of the parameter settings where the CNF-based scheme is optimal. One may wonder if this is a limitation of our Shamir-based scheme in particular or a limitation of Shamir sharing in general.  We show in Proposition~\ref{prop:impossibleShamir} that it is the latter.  That is, there are some parameter regimes where \emph{no} HSS based on Shamir sharing can perform as well as an HSS based on CNF sharing.

\paragraph{Applications: High-rate PIR and more.} 
We apply our HSS for low-degree polynomials to obtain the first information-theoretic private information retrieval (PIR) protocols that simultaneously achieve low (sublinear) upload cost and near-optimal download rate that gets arbitrarily close to~1 when the number of servers grows. A $t$-private $k$-server PIR protocol~\cite{Chor:1998:PIR:293347.293350} allows a client to retrieve a single symbol from a database in $\mathcal{Y}^N$, which is replicated among the servers, such that no $t$ servers learn the identity of the retrieved symbol. The typical goal in the PIR literature is to minimize the communication complexity when $\mathcal{Y}=\{0,1\}$. In particular, the communication complexity should be sublinear in $N$. Here we consider the case where the database has (long) $\ell$-bit records, namely $\mathcal{Y}=\{0,1\}^\ell$. Our goal is to maximize the download rate while keeping the upload cost sublinear in $N$.
Chor et al.~\cite{Chor:1998:PIR:293347.293350} obtain, for any integers $d,t\ge 1$ and $k=dt+1$,  a $t$-private $k$-server PIR protocol with upload cost $O(N^{1/d})$ and download rate $1/k$ (for sufficiently large $\ell$). This protocol implicitly relies on a simple HSS for degree-$d$ polynomials. Using our high-rate HSS for degree-$d$ polynomials, by increasing the number of servers $k$ the download rate can be improved to $1-dt/k$ (in particular, approach 1 when $k\gg dt+1$) while maintaining the same asymptotic upload cost. See Theorem~\ref{thm:PIR} for a formal statement. 

It is instructive to compare this application to a recent line of work on the download rate of PIR. \newvictor{Sun and Jafar~\cite{SJ16,SJ17,SJ18robust}, following~\cite{SRR14}, have shown that the optimal download rate of $t$-private $k$-server PIR is $(1-t/k)/(1-(t/k)^N)$.}
However, their positive result has $\Omega(N)$ upload cost. We get a slightly worse\footnote{
Note that our positive result applies also to a stronger variant of amortized PIR, which amortizes over $\ell$ {\em independent} instances of PIR with databases in $\{0,1\}^N$. In this setting, our construction with $d=1$ achieves an optimal rate of \newvictor{$1-t/k$} (where optimality follows from~\cite{SRR14} or from Lemma~\ref{lem:negNotLinear}). In Section~\ref{sec:bbintro} below we discuss a construction of {\em computationally} secure PIR that achieves the same rate of $1-1/k$ \newvictor{(for $t=1$)} with {\em logarithmic} upload cost.
} download rate of \newvictor{$1-dt/k$}, where the upload cost is sublinear for $d\ge 2$.

Finally, beyond PIR, HSS for low-degree polynomials can be directly motivated by a variety of other applications.  For instance, an inner product between two integer-valued vectors (a degree-2 function) is a measure of correlation. To amortize the download rate of computing $\ell$ such correlations, our HSS scheme for degree-2 polynomials over a big field $\F$ can be applied. As another example, the intersection of $d$ sets $S_i\subseteq [\ell]$, each represented by a characteristic vector in $\F_2^\ell$, can be computed by $\ell$ instances of a degree-$d$ monomial over $\F_2$. See~\cite{AC:LaiMalSch18,IshaiLM21} for more examples.



\subsubsection{Black-box rate amplification for additive HSS}
\label{sec:bbintro}
The results discussed so far are focused on information-theoretic HSS for a specific function class. Towards obtaining other kinds of high-rate HSS schemes, in Section~\ref{sec:blackbox} 
we develop a general black-box transformation technique (Lemma~\ref{black-box}) that can improve the download rate of any  {\em additive} HSS (where $\Rec$ adds up the output shares) by using additional servers. 
More concretely, the transformation can obtain a high-rate $t$-private $k$-server HSS scheme $\Pi$ by making a {\em black-box} use of any additive $t_0$-private $k_0$-server $\Pi_0$, for suitable choices of $k_0$ and $t_0$.  The transformation typically has a small impact on the upload cost and applies to both information-theoretically secure and computationally secure HSS. While we cannot match the parameters of the HSS for low-degree polynomials (described above) by using this approach, we can apply it to other function classes and obtain rate that approaches 1 as $k$ grows.  

We present three useful instances of this technique. In the first (Theorem~\ref{th:bb1}), we use $\Pi_0$ with $k_0={k \choose t}$ and $t_0=k_0-1$ to obtain a $t$-private $k$-server HSS $\Pi$ with rate $1-t/k$. Combined with a computational HSS for circuits from~\cite{DHRW16} (which is based on a variant of the Learning With Errors assumption) this gives a general-purpose computationally $t$-private HSS with rate $1-t/k$, approaching $1$ when $k\gg t$, at the price of upload cost and computational complexity that scale with ${k \choose t}$.  
This should be compared to the $1/(t+1)$ rate obtained via a direct use of~\cite{DHRW16}. Note that, unlike recent constructions of ``rate-1'' fully homomorphic encryption schemes~\cite{TCC:GenHal19,TCC:BDGM19}, here the concrete download rate is independent of the security parameter.

The above transformation is limited in that it requires $\Pi_0$ to have a high threshold $t_0$, whereas most computationally secure HSS schemes from the literature are only 1-private. Our second instance of a black-box transformation (Theorem~\ref{hss2}) uses any 1-private 2-server $\Pi_0$ to obtain a 1-private $k$-server $\Pi$ with rate $1-1/k$. Applying this to HSS schemes from~\cite{BoyleGI15,BGI16}, we get (concretely efficient) 1-private $k$-server computational PIR schemes with download rate $1-1/k$, based on any pseudorandom generator, with upload cost \newvictor{$O(\lambda k \log (N))$} (where $\lambda$ is a security parameter). \newvictor{By using a different approach,  the authors in \cite{hafiz2019bit} show that when $k$ is a power of two an upload cost of $O(\lambda \log(k) \log (N))$ is sufficient.} We can also apply this transformation to 1-private 2-server HSS schemes from~\cite{BGI16a,OrlandiSY21,RoyS21}, obtaining 1-private $k$-server HSS schemes for branching programs based on number-theoretic cryptographic assumptions (concretely, DDH or DCR), with rate $1-1/k$.

Our third and final instance of the black-box transformation is motivated by the goal of {\em information-theoretic} PIR with sub-polynomial ($N^{o(1)}$) upload cost and download rate approaching 1. Here the starting point is a 1-private 3-server PIR scheme with sub-polynomial upload cost based on matching vectors~\cite{STOC:Yekhanin07,STOC:Efremenko09}. While this scheme is not additive, it can be made additive by doubling the number of servers. We then (Theorem~\ref{th:mv-pir}) apply the third variant of the transformation to the resulting 1-private 6-server PIR scheme,  obtaining a 1-private $k$-server PIR scheme with sub-polynomial upload cost and rate $1-1/\Theta(\sqrt{k})$. Note that here we cannot apply the previous transformation since $k_0=6>2$.

We leave open the question of fully characterizing the parameters for which such black-box transformations exist.


\subsubsection{Nonlinear download rate amplification}

All of the high-rate HSS schemes considered up to this point have a {\em linear} reconstruction function $\Rec$. Moreover, they all improve the rate of existing baseline schemes by increasing the number of servers. 
In Section~\ref{sec:nonlinear} we study the possibility of circumventing this barrier by relaxing the linearity requirement, without increasing the number of servers.  

The starting point is the following simple example. Consider the class $\cal F$ of degree-$d$ monomials (products of $d$ variables) over a field $\F$ of size $|\F|\approx d$. Letting $k=d+1$ and $t=1$, we have the following standard ``baseline'' HSS scheme: $\Share$ applies Shamir's scheme (with $t=1$) to each input; to evaluate a monomial $f$, $\Eval$ (computed locally by each server) multiplies the shares of the $d$ variables of $f$; finally $\Rec$ can recover the output by interpolating a degree-$d$ polynomial, applying a linear function to the output shares. The key observation is that since each input share is uniformly random over $\F$, the output of $\Eval$ is biased towards 0. Concretely, by the choice of parameters, each output share is 0 except with $\approx 1/e$ probability. It follows that when amortizing over $\ell$ instances, and settling for $2^{-\Omega(\ell)}$ failure probability, the output of $\Eval$ can be compressed by roughly a factor of $e$ by simply listing all $\approx \ell/e$ nonzero entries and their locations.

While this example already circumvents our negative result for linear HSS, it only applies to evaluating products over a big finite field, which is not useful for any applications we are aware of. Moreover, this naive compression method does not take advantage of correlations between output shares. In Theorem~\ref{thm:slepianWolf} we generalize and improve this method 
by using a variant of Slepian-Wolf coding tailored to the HSS setting. Note that we cannot use the Slepian-Wolf theorem directly, because the underlying joint distribution depends on the output of $f$ and is thus not known to each server. We apply our general methodology to the simple but useful case where $f$ computes the AND of two input bits. As discussed above, $\ell$ instances of such $f$ can be motivated by a variant of the private set intersection problem in which the output client should learn the intersection of two subsets of $[\ell]$ whose characteristic vectors are secret-shared between the servers. By applying Theorem~\ref{thm:slepianWolf} to a 1-private 3-server HSS for AND based on CNF sharing, we show (Corollary~\ref{cor:beatlinear}) that the download rate can be improved from $1/3$ to $\approx 0.376$ (with $2^{-\Omega(\ell)}$ failure probability), which is the best possible using our general compression method. 

Perhaps even more surprisingly, the improved rate can be achieved while ensuring that the output shares reveal no additional information except the output (Proposition~\ref{weakshss}). We refer to the latter feature as {\em symmetric privacy}. This should be contrasted with the above example of computing a monomial over a large field, where the output shares {\em do} reveal more than just the output (as they reveal the product of $d$ degree-1 polynomials that encode the inputs). While symmetric privacy can be achieved by rerandomizing the output shares---a common technique in protocols for secure multiparty computation~\cite{STOC:BenGolWig88,STOC:ChaCreDam88,EC:CraDamMau00}---this eliminates the possibility for compression.

Our understanding of the rate of nonlinear HSS is far from being complete. Even for simple cases such as the AND function, some basic questions remain. Does the compression method of Theorem~\ref{thm:slepianWolf} yield an optimal rate? Can the failure probability be eliminated? Can symmetric privacy be achieved with nontrivial rate even when the output client may collude with an input client? We leave a more systematic study of these questions to future work.



\subsection{Technical Overview}

\newvictor{In this section} we give a high level overview of the main ideas used by our results.
\subsubsection{Linear HSS for low-degree polynomials}
We begin by describing our results in Section~\ref{sec:linear} for linear HSS.  We give positive and negative results; we start with the positive results.

\paragraph{HSS for Concatenation.}
Both of our constructions of linear HSS (the first based on CNF sharing, the second on Shamir sharing) begin with a solution for the special problem of \emph{HSS for concatenation} (Definition~\ref{def:CONCAT}).  Given $\ell$ inputs $x_1, \ldots, x_\ell$ that are shared separately, the goal is for the servers to produce output shares (the outputs of $\Eval$) that allow for the joint recovery of $(x_1, \ldots, x_\ell)$, while still using small communication.  Once we have an HSS for concatenation based on either CNF sharing (Definition~\ref{cnf}) or Shamir sharing (Definition~\ref{shamir}), an HSS for low-degree polynomials readily follows by exploiting the specific structure of CNF or Shamir.\footnote{In some parameter regimes, HSS for concatenation with optimal download rate is quite easy to achieve using other secret sharing schemes, such as the multi-secret extension of Shamir's scheme due to Franklin and Yung~\cite{FY92}. However, this does not suffice for obtaining rate-optimal solutions for polynomials of degree $d>1$ (see Remark~\ref{rem:triv_concat_wont_work}). 
}

This problem can be viewed as an instance of \emph{share conversion}. Concretely, the problem is to locally convert from a linear secret sharing scheme (LSSS) that shares $x_1, \ldots, x_\ell$ \emph{separately} (via either CNF or Shamir) to a linear multi-secret sharing scheme (LMSSS) that shares $x_1, \ldots, x_\ell$ \emph{jointly}.  Thus, our constructions follow by understanding such share conversions.

\paragraph{Construction from CNF sharing.}  If we begin with CNF sharing, we can completely characterize the best possible share conversions as described above.  Because $t$-private CNF shares can be locally converted to \emph{any} $\geq t$-private LMSSS (Corollary~\ref{cdiplusplus}, extending~\cite{CDI05} from LSSS to LMSSS), the above problem of share conversion collapses to the problem of understanding the best rate attainable by an LMSSS with given parameters.  It is well-known that LMSSS's can be constructed from linear error correcting codes with good dual distance (see, e.g., \cite{Mas95,CCGHV07}). 
However, in order to construct HSS we are interested only in $t$-private LMSSS with the property that \emph{all} $k$ parties can reconstruct the secret (as opposed to any $t+1$ parties or some more complicated access structure), which results in a particularly simple correspondence (Lemma~\ref{folded-codes}, generalizing~\cite{GM10}).
This in turn leads to Theorem~\ref{thm:cnf-hss-d1}, which \emph{characterizes} the best possible download rate for any linear HSS-for-concatenation in terms of the best trade-off between the rate and distance of a linear code.  This theorem gives a characterization (a negative as well as a positive result).  The positive result (when we plug in good codes) gives our CNF-based HSS-for-concatenation, which leads to Theorem~\ref{thm:cnf-hss}, our CNF-based HSS for general low-degree polynomials. 

\paragraph{Construction from Shamir sharing.}  If we begin with Shamir sharing, we can no longer locally convert to any LMSSS we wish.  Instead, we develop a local conversion to a specific LMSSS with good rate.  In order to develop this construction, we leverage ideas from the \em regenerating codes \em literature (see the discussion in Section~\ref{sec:related} below).  Unfortunately, we are not able to use an off-the-shelf regenerating code for our purposes, but instead we take advantage of some differences between the HSS setting and the regenerating code setting in order to construct a scheme that suits our needs.  This results in Theorem~\ref{reg-share-conversion} for HSS-for-concatenation, and then Theorem~\ref{thm:shamir-hss} for HSS for general low-degree polynomials.

As an application of our Shamir-based construction, 
we extend information-theoretic PIR protocols from \cite{Chor:1998:PIR:293347.293350}
to allow better download rate by employing more servers, while maintaining the same (sublinear) upload cost.  This leads to Theorem~\ref{thm:PIR}.

\paragraph{Negative results.}  As mentioned above, Theorem~\ref{thm:cnf-hss-d1} contains both positive and negative results, with the negative results stemming from negative results about the best possible trade-offs between the rate and distance of linear codes.  This shows that our CNF-based construction is optimal for HSS-for-concatenation, but unfortunately does not extend to give a characterization of the best download rate for HSS for low-degree polynomials.  Instead, in Theorem~\ref{thm:barrier}, we use a linear-algebraic argument to show that no linear HSS for degree $d$ polynomials can have download rate better than $1 - dt/k$. This means that for sufficiently large $\ell$, both of our HSS schemes for low-degree polynomials have an optimal rate.

\subsubsection{Black-box rate amplification for additive HSS}

Next, we give a brief technical overview of our results in Section~\ref{sec:blackbox} on  black-box rate amplification for additive HSS.


\paragraph{The general approach.} We show that starting from any $t_0$-private $k_0$-server HSS scheme $\Pi_0$ with additive reconstruction (over some finite field), it is possible to construct other $t$-private $k$-server HSS schemes with higher rate. The main observation is that due to the additive reconstruction property, after the servers perform their evaluation, the output shares form an additive sharing of the output $y=y^{(1)}+\ldots+y^{({k_0})}$ (which is $t_0$-private). By controlling how the shares are replicated among the servers, each output $y_i$, $i=1,\ldots,\ell$, is shared among the servers according to some LSSS. Hence, at this stage, this becomes a share conversion problem, where we want to convert separately shared outputs into a high-rate joint LMSSS, which yields our high-rate HSS scheme. 

\paragraph{Black-box transformations with large $k_0$.} We observe that if $k_0=\binom{k}{t}$ and $t_0=k_0-1$, then we can replicate the shares of $\Pi_0$ in such a manner that each output $y$ is $t$-CNF shared among the servers. Concretely, if $y=y^{(1)}+\ldots+y^{({k_0})}$, then we can identify each index $i=1,\ldots,k_0$ with a subset $T_i\in\binom{[k]}{t}$, and provide each server $j=1,\ldots,k$ with $y^{(i)}$ if and only if $j\notin T_i$, after which the servers hold a $t$-CNF sharing of $y$ (Definition~\ref{cnf}). Therefore, as in Section \ref{sec:optimal-linear-hss-intro}, this now reduces to finding a $t$-private LMSSS with the best possible rate. 

\paragraph{Black-box transformations with $k_0=2$.} Most computationally secure HSS schemes from the literature are 1-private 2-server schemes, to which the previous transformation does not apply. Our second transformation converts any (additive) 1-private 2-server $\Pi_0$ to a $1$-private $k$-server $\Pi$ with rate $1-1/k$. This is obtained by replicating $k-1$ pairs of  (output) shares among the $k$ servers in a way that: (1) each server gets only one share from each pair; (2) the servers can locally convert their shares to a 1-private $(k-1)$-LMSSS of the outputs of rate $1-1/k$. To illustrate this approach for $k=3$, suppose we are given a $2$-additive secret sharing for every output $y_i=y_i^{({1})}+y_i^{({2})}$, $i=1,2$. We obtain a $1$-private $3$-server $2$-LMSSS sharing of the outputs with information rate $2/3$ in the following way:

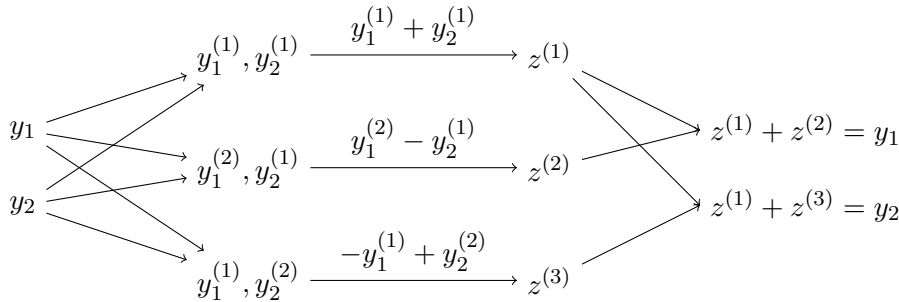
\begin{center}
\begin{tikzpicture}
\node(x) at (0,1) {$y_1$};
\node(y) at (0,0) {$y_2$};
\node(a) at (3, 2) {$y_1^{({1})},y_2^{({1})}$};
\node(b) at (3, 0.5) {$y_1^{({2})},y_2^{({1})}$};
\node(c) at (3, -1) {$y_1^{({1})},y_2^{({2})}$};
\node(ay) at (7, 2) {$z^{({1})}$};
\node(by) at (7, 0.5) {$z^{({2})}$};
\node(cy) at (7, -1) {$z^{({3})}$};
\node[anchor=west](f) at (9, 1) {$z^{({1})}+z^{({2})}=y_1$};
\node[anchor=west](f2) at (9, 0) {$z^{({1})}+z^{({3})}=y_2$};
\draw[->] (x) to (a);
\draw[->] (x) to (b);
\draw[->] (x) to (c);
\draw[->] (y) to (a);
\draw[->] (y) to (b);
\draw[->] (y) to (c);
\draw[->] (a) to node[above]{$y_1^{({1})}+y_2^{({1})}$} (ay);
\draw[->] (b) to node[above]{$y_1^{({2})}-y_2^{({1})}$}(by);
\draw[->] (c) to node[above]{$-y_1^{({1})}+y_2^{({2})}$}(cy);
\draw[->] (ay) to (f.west);
\draw[->] (by) to (f.west);

\draw[->] (ay) to (f2.west);

\draw[->] (cy) to (f2.west);
\end{tikzpicture}
\end{center}
Since we need only $3$ shares to reconstruct $2$ secrets, the rate is $2/3$. This can be generalized in a natural way to $k$ servers and $k-1$ outputs.


\paragraph{Sub-polynomial upload cost PIR with high rate.} Our third variant of the black-box transformation is motivated by the goal of high-rate information-theoretic PIR with sub-polynomial upload cost. Unlike the previous parts, here we start with a $1$-private $6$-server HSS, which has a lower privacy-to-servers ratio. While we don't obtain a tight characterization for this parameter setting, we can reduce the problem to a combinatorial packing problem, which suffices to get rate approaching 1. Concretely, for a universe $\{1,\ldots,q\}$ we need the largest possible family of subsets $\mathcal{S}\subseteq\binom{[q]}{5}$, such that distinct sets in $\mathcal{S}$ have at most a single element in common. Next, we show that it is possible to associate every set in $\mathcal{S}$ with a secret from $\FF$, and also every set in $\mathcal{S}$ \emph{and} an element of the universe $\{1,\ldots,q\}$ with a server, in such a way that the output shares, each an element of $\FF$, constitute a high-rate LMSSS sharing of the outputs. This gives us a rate of $1-q/(q+|\mathcal{S}|)$. Using known constructions of subset families as above of size $\Theta(q^2)$ \cite{EH63}, we get a download rate of $1-1/\Theta(q)=1-1/\Theta(\sqrt{k})$. 

\subsubsection{Nonlinear download rate amplification}\label{sec:SWtech}
Finally, we describe our results in Section~\ref{sec:nonlinear} for nonlinear rate amplification with a small error probability.

\paragraph{Slepian-Wolf-style Compression.}  We begin with any HSS scheme $\Pi$ for a function class $\mathcal{F}$.  Suppose that, sharing a secret $\bx$ under $\Pi$, each server $j$ has an output share (that is, the output of $\Eval$), $z_j$.  The vector $\bz$ of these output shares is a random variable, over the randomness of $\Share$ and $\Eval$.  Thus, if we repeat this $\ell$ times with $\ell$ secrets $\bx_1, \ldots, \bx_\ell$ to get $\ell$ draws $\bz_1, \bz_2, \ldots, \bz_\ell$, we may hope to compress the sequence of $\bz_i$'s if, say, $H(\bz)$ is small.  

There are two immediate obstacles to this hope.  The first obstacle is that each vector $\bz_i$ is split between the $k$ servers, with each server holding only one coordinate.  The second obstacle is that the underlying distribution of each $\bz_i$ depends on the secret $\bx_i$, which is not known to the reconstruction algorithm.  
Both of these obstacles can be overcome directly by having each server compress its shares individually.  This trivially gets around the first obstacle, and it gets around the second because, by $t$-privacy, the distribution of any one output share does not depend on the secret.  However, we can do better.

The first obstacle has a well-known solution, known as \emph{Slepian-Wolf coding}.  In Slepian-Wolf coding, a random source $\bz$ split between $k$ servers as above can be compressed separately by each server, with download cost for a sequence of length $\ell$ approaching $\ell \cdot \max_{S \subseteq [k]} H(\bz_S | \bz_{S^c})$.  (Here, $\bz_S$ denotes the restriction of $\bz$ to the coordinates in $S$.)  Unfortunately, classical Slepian-Wolf coding does not work in the face of the second obstacle, that is if the underlying distribution is unknown.  

The most immediate attempt to adapting the classical Slepian-Wolf argument to deal with unknown underlying distributions is to take a large union bound over all $|\mathcal{X}^m|^\ell$ possible sequences of secrets.  Unfortunately, this does not work, as the union bound is too big.  However, by using the \emph{method of types} (see, e.g., \cite[Section 11.1]{coverThomas}), we are able to reduce the union bound to a manageable size.  This results in our main technical theorem of this section, Theorem~\ref{thm:slepianWolf}.  We instantiate Theorem~\ref{thm:slepianWolf} with 3-server HSS for the AND function, based on $3$-party CNF sharing, demonstrating how to beat the impossibility result in Theorem~\ref{thm:barrier} even for a simple and well-motivated instance of HSS. 

\paragraph{Symmetric privacy for free.} A useful added feature for HSS is having output shares that hide all information other than the output. We refer to this as {\em symmetric privacy}.
The traditional method of achieving this is by ``rerandomizing'' the output shares. However, this approach conflicts with the compression methodology discussed above. Somewhat unexpectedly, we show (Proposition~\ref{weakshss}) that our rate-optimized HSS for AND already satisfies the symmetric privacy property. 

To give a rough idea why, we start by describing the HSS scheme for AND that we use to  instantiate Theorem~\ref{thm:slepianWolf} in an optimal way. In fact, we describe and analyze a generalization to multiplying two inputs $a,b$ in a finite field $\F$ (the AND scheme is obtained by using $\F=\F_2$). The $\Share$ function shares each secret using 1-private CNF. Concretely, $a$ is first randomly split into $a=a_1+a_2+a_3$ and similarly $b$, and then server $j$ gets the 4 shares $a_i,b_i$ with $i\neq j$. For defining $\Eval$, we can assign each of the 9 monomials $a_ib_j$ to one of the servers that can evaluate it, and let each server compute the sum of its assigned monomials. It turns out that the monomial assignment for which Theorem~\ref{thm:slepianWolf} yields the best rate is the {\em greedy} assignment, where each monomial is assigned to the first server who can evaluate it. Using this assignment, the first and last output shares are $y^{(1)}=(a_2+a_3)(b_2+b_3)=(a-a_1)(b-b_1)$ and $y^{(3)}=a_1b_2+a_2b_1$. Since $y^{(1)}+y^{(2)}+y^{(3)}=ab$, it suffices to show that the joint distribution of $(y^{(1)},y^{(3)})$ reveals no information about $a,b$ other than $ab$. 

This can be informally argued as follows. First, viewing $y^{(3)}$ as a randomized function of $a_1,b_1$ with randomness $a_2,b_2$, the only information revealed by $y^{(3)}$ about $a_1,b_1$ is whether $a_1=b_1=0$. Since $a_2,b_2$ are independent of $(a-a_1)(b-b_1)$, the information about $(a,b)$ revealed by $(y^{(1)},y^{(3)})$ is equivalent to $(a-a_1)(b-b_1)$ together with the predicate $a_1=b_1=0$. Since $y^{(1)}$ is independent of $a,b$ and is equal to $ab$ conditioned on $a_1=b_1=0$, the latter reveals nothing about $a,b$ other than $ab$, as required.  In the formal proof of symmetric privacy (Proposition~\ref{weakshss}) we show an explicit bijection between the randomness leading to the same output shares given two pairs of inputs that have the same product.

To complement the above, we observe (Proposition~\ref{shamir-not-shss}) that if we use the natural HSS for multiplication based on Shamir's scheme (namely, locally multiply Shamir shares without rerandomizing), then symmetric privacy no longer holds. Indeed, in this scheme the output shares determine the product of two random degree-1 polynomials with free coefficients $a$ and $b$ respectively. Thus, one can distinguish between the case $a=b=0$, in which the product polynomial is of the form $\alpha X^2$, and the case where $a=0,b=1$, in which the product polynomial typically has a linear term. Note that the two cases should be indistinguishable, since in both we have $ab=0$. The insecurity of homomorphic multiplication without share randomization has already been observed in the literature on secure multiparty computation~\cite{STOC:BenGolWig88}.



\subsection{Related Work}\label{sec:related}

We already mentioned related work on homomorphic secret sharing, fully homomorphic encryption, private information retrieval, and secure multiparty computation. 
In the following we briefly survey related work on regenerating codes and communication-efficient secret sharing.

\subsubsection{Regenerating codes}
Our Shamir-based HSS scheme is inspired by \emph{regenerating codes}~\cite{dimakis}, and in particular the work on using Reed-Solomon codes as regenerating codes~\cite{SPDC14, GW16,TYB17}.  
A \emph{Reed-Solomon code} of block length $k$ and degree $d$ is the set $C = \{ (p(\alpha_1), \ldots, p(\alpha_k)) : p \in \F[X], \deg(p) \leq d \}$.  
A \emph{regenerating code}, introduced by \cite{dimakis} in the context of distributed storage, is a code that allows the recovery of a single erased codeword symbol by downloading not too much information from the remaining symbols.  
The goal is to minimize the number of bits downloaded from the remaining symbols.  Thus, a repair scheme for degree $dt$ Reed-Solomon codes immediately yields an HSS for degree-$d$ polynomials with $t$-private Shamir-sharing with the same download cost.
It turns out that one can indeed obtain download-optimal HSS schemes for low degree polynomials this way from the regenerating codes in~\cite{TYB17} (see Corollary~\ref{cor:TYB}).
However, while this result obtains the optimal download rate of $1 - dt/k$, even for $\ell=1$, the field size $\FF$ must be extremely large: doubly exponential in the number of servers $k$.  Alternatively, if we would like to share secrets over $\FF_2$, for example, the upload cost must be huge (see Remark~\ref{rem:regenCompare}), even worse than CNF.
Moreover, \cite{TYB17} shows that this is unavoidable if we begin with a regenerating code: any linear repair scheme for Reed-Solomon codes that corresponds to an optimal-rate HSS must have (nearly) such a large field size.  In contrast, our results in Section~\ref{sec:linear} yield Shamir-based HSS with optimal download rate and with reasonable field size and upload cost.

The reason that we are able to do better (circumventing the aforementioned negative result of \cite{TYB17} for Reed-Solomon regenerating codes) is that (a) in HSS we are only required to recover the secret, while in renegerating codes one must be able to recover any erased codeword symbol (corresponding to any given share); (b) we allow the shares to be over a larger field than the secret comes from;\footnote{In the regenerating codes setting, this corresponds to moving away from the MSR (Minimum Storage Regenerating codes) point and towards the MBR (Minimum Bandwidth Regenerating codes) point; see~\cite{dimakisSurvey}.  To the best of our knowledge, repair schemes for Reed-Solomon codes have not been studied in this setting.} and (c) we amortize over $\ell > 1$ instances.

However, even though we cannot use a regenerating code directly, we use ideas from the regenerating codes literature.  In particular, our scheme can be viewed as one instantiation of the framework of \cite{GW16} and has ideas similar to those in~\cite{TYB17}; again, our situation is simpler particularly due to (a) above.

We mention a few related works that have also used techniques from regenerating codes.  First, the work~\cite{ACEDX21} uses regenerating codes, including a version of the scheme from \cite{GW16}, in order to reduce the communication cost per multiplication in secure multiparty computation. Their main result is a logarithmic-factor improvement in the communication complexity for a natural class of MPC tasks compared to previous protocols with the same round complexity.\footnote{One may \newvictor{ask} why \cite{ACEDX21} can use a regenerating code while we cannot.  The reason is that we are after \emph{optimal} download rate.  Indeed, one can obtain nontrivial download rate in our setting using a variant of the scheme in \cite{GW16}, which does have a small field size.  However, as is necessary for regenerating codes over small fields, the bandwidth of the regenerating code does not meet the so-called \emph{cut-set bound}, and correspondingly the download rate obtained this way is not as good as the optimal $1 - dt/k$ download rate that is achieved with our approach.}
Second, the recent work~\cite{SW21} studies an extension of regenerating codes (for the special case of Reed-Solomon codes) where the goal is not to compute a single missing symbol but rather any linear function of the symbols.  While primarily motivated by distributed storage, that result can be viewed as studying the download cost of HSS for Shamir sharing, in the single-client case where $m=1$, and restricted to linear functions.  One main difference of \cite{SW21} from our work is that in \cite{SW21} the secrets are shared jointly, while in our setting (with several clients) the secrets must be shared independently.   Thus \cite{SW21} does not immediately imply any results in our setting. Finally, the work \cite{EGKY21} studies the connection between regenerating codes and proactive secret sharing.

\subsubsection{Communication-efficient secret sharing}
As noted above, the HSS problem is easier than the general problem of regenerating codes, as we only need to recover the secret(s), rather than any missing codeword symbol (which corresponds to recovering any missing share in the HSS setting).  As such, one might hope to get away with smaller field size.  In fact, this has been noticed before, and previous work has capitalized on this in the literature on \emph{Communication-Efficient Secret Sharing} (CESS)~\cite{HLKB16,HB17,BR17,RKOV18}.  The simplest goal in this literature is to obtain optimal-download-rate HSS for the special case that $\mathcal{F}$ consists only of the identity function; more complicated goals involve (simultaneously) obtaining the best download rate for any given authorized set of servers (not just $[k]$); and also being able to recover missing shares (as in regenerating codes).  Most relevant for us, the simplest goal (and more besides) have been attained, and optimal schemes are known (e.g.,~\cite{HLKB16}).

However, while related, CESS---even those based on Shamir-like schemes as in \cite{HLKB16}---do not immediately yield results for HSS, or even for HSS-for-concatenation.  The main difference is that in CESS, the inputs need not be shared separately.  For example, when restricted to the setting of HSS for the identity function, the scheme in \cite{HLKB16} is simply the $\ell$-LMSSS described in Remark~\ref{rem:ell-shamir}(b), where the $\ell$ inputs are interpreted as coefficients of the same polynomial and are shared \emph{jointly.}  

One exception is the scheme from \cite{HB17}, which is directly based on Shamir's scheme (with only one input) over a field $\FF$.  The scheme is linear, and so it immediately yields an HSS for 
degree-$d$ polynomials.  However, while the download rate approaches optimality as the size of the field $\FF$ grows, it is not optimal.\footnote{In more detail, the download rate of the $t$-private, $k$-server Shamir-based scheme for degree-$d$ polynomials in \cite{HB17} is $\left(\frac{k}{k-dt} + \frac{ k^2(k-dt)^2 }{4 \log_{|\B|}|\FF|} \right)^{-1}$, where $\B$ is an appropriate subfield of $\FF$.  In particular, $\FF$ should be exponentially large in $k$ before this rate is near-optimal.}

\paragraph{Organization of the paper.} 
In Section~\ref{sec:prelim}, we set notation and formally define the notions of HSS and LMSSS that we will use.  In Section~\ref{sec:linear}, we present our results for linear HSS.  In Section~\ref{sec:blackbox}, we present our results for black-box rate amplification of additive HSS.  Finally, in Section~\ref{sec:nonlinear}, we present our results for rate amplification using nonlinear reconstruction.

\section{Preliminaries}\label{sec:prelim}

\paragraph{Notation.} For an integer $n$, we use $[n]$ to denote the set $\{1, 2, \ldots, n\}$.  For an object $w$ in some domain $\mathcal{W}$, we use $|w|$ to denote the number of \emph{bits} required to write down $w$.  That is, $|w| = \log_2|\mathcal{W}|$.  We will only use this notation when the domain is clear. 
We generally use bold symbols (like $\bx$) to denote vectors.

\subsection{Homomorphic Secret Sharing} We consider HSS schemes with $m$ inputs and
$k$ servers; we assume that each input is shared independently.  
We would like to compute functions from a function class
$\mathcal{F}$ consisting of functions $f:\mathcal{X}^m\rightarrow\mathcal{Y}$, where 
$\mathcal{X}$ and $\mathcal{Y}$ are input and output domains, respectively.  Formally, we have the following definition.


\begin{definition}[HSS, modified from \cite{BGILT18}]\label{def:HSS}
A $k$-server HSS for a function family $\mathcal{F}=\{f:\mathcal{X}^m\rightarrow\mathcal{Y}\}$ is a tuple of algorithms $\Pi=(\Share,\Eval,\Rec)$ with the following syntax:
\begin{itemize}
    \item $\Share(x)$:
    On input 
    $x\in\mathcal{X}$
    the (randomized) sharing algorithm $\Share$ outputs $k$ shares 
    $(x^{(1)},\ldots,x^{(k)})$. We will sometimes write $\Share(x,r)$ to explicitly refer to the randomness $r$ used by $\Share$.
We refer to the
$x^{(j)}$ as \emph{input shares}.

    \item $\Eval(f,j,(x_1^{(j)},\ldots,x_m^{(j)}))$: On input $f\in\mathcal{F}$ (evaluated function), $j\in[k]$ (server index) and $x_1^{(j)},\ldots,x_m^{(j)}$ ($j$th share of each input), the evaluation algorithm $\Eval$ outputs $y^{(j)}$, corresponding to server $j$'s share of $f(x_1,\ldots,x_m)$. We refer to the $y^{(j)}$ as the \emph{output shares}.
    \item $\Rec(y^{(1)},\ldots,y^{(k)})$: Given $y^{(1)},\ldots,y^{(k)}$ (list of output shares), the reconstruction algorithm $\Rec$ computes a final output $y\in\mathcal{Y}$.
\end{itemize}
The algorithms $\Pi = (\Share, \Eval, \Rec)$ should satisfy the following requirements:

\begin{itemize}
    \item \textbf{Correctness:} For any $m$ inputs $x_1,\ldots,x_m\in\mathcal{X}$ and $f\in\mathcal{F}$,
    \[
\Pr\left[
\Rec\left(y^{(1)},\ldots,y^{(k)} \right)=f(x_1,\ldots,x_m):
\begin{array}{c}
\forall i\in[m]\left(x_i^{(1)},\ldots,x_i^{(k)} \right)\leftarrow
\Share(x_i)\\
\forall j\in [k]\ y^{(j)} \leftarrow \Eval(f,j,(x_1^{(j)},\ldots,x_m^{(j)}))
\end{array}\right]= 1.
\]
If instead the above probability of correctness is at least $\alpha$ for some $\alpha \in (0,1)$ (rather than being exactly $1$), we say that $\Pi$ is \emph{$\alpha$-correct}.
\item \textbf{Security:} We say that $\Pi$ is \emph{$t$-private}, if for every $T\subseteq[k]$ of size $|T|\leq t$ and $x,x'\in\mathcal{X}$, the distributions $(x^{(j)})_{j\in T}$ and $((x')^{(j)})_{j\in T}$  are identical, where $\bx$ is sampled from $\Share(x)$ and $\bx'$ from $\Share(x')$.
\end{itemize}
\end{definition}


While the above definition does not refer to computational complexity, in positive results we require by default that all of algorithms are polynomial in their input and output length.

A major theme of this work is amortizing the download cost of HSS over $\ell$ function evaluations.  Informally, there are now $\ell$ points in $\mathcal{X}^m$, $\bx_j = (x_{1,j}, x_{2,j}, \ldots, x_{m,j})$ for each $j \in [\ell]$, and each input $x_{i,j}$ is shared \emph{separately} using $\Share$.  
The goal is to compute $f_j(\bx_j)$ for each $j \in [\ell]$ for some $f_j \in \mathcal{F}$.
Formally, we can view this as a special case of Definition~\ref{def:HSS} applied to the following class $\mathcal{F}^\ell$.

\begin{definition}[The class $\mathcal{F}^\ell$] Given a function class $\mathcal{F}$ that maps $\mathcal{X}^m$ to $\mathcal{Y}$, we define
$\mathcal{F}^\ell$ to be the function class that maps $\mathcal{X}^{\ell m}$ to $\mathcal{Y}^\ell$ given by
\[\mathcal{F}^{\ell}:=\{(x_{i,j})_{i\in[m],j\in [\ell]} \mapsto
(f_1(\vec{x}_1),\ldots,f_\ell(\vec{x}_\ell))\ :\ f_1,\ldots,f_\ell\in\mathcal{F}\}.
\]
\end{definition}

\paragraph{Computational HSS.} In this work we will be primarily interested in information-theoretic HSS as in Definition~\ref{def:HSS}.  However, in Section~\ref{sec:blackbox} we will also be interested in {\em computationally secure} HSS schemes, where the security requirement is relaxed to hold against computationally bounded distinguishers. \newvictor{A formal definition appears in} Appendix~\ref{app:computational}. 
\medskip

We will be particularly interested in HSS schemes whose sharing and/or reconstruction functions are linear functions over a finite field, defined as follows.

\begin{definition}[Linear HSS]\label{def:linear}
Let $\mathbb{F}$ be a finite field.
We say that an HSS scheme $\Pi=(\Share,\Eval,\Rec)$ has \emph{linear reconstruction} over $\mathbb{F}$ if $\mathcal{Y}=\mathbb{F}^b$ for some integer $b
\ge 1$; $\Eval(f, j, \bx^{(j)})$ outputs $y^{(j)} \in \mathbb{F}^{b_j}$ for integer $b_j\ge 0$; and $\Rec: \mathbb{F}^{\sum_j b_j} \to \mathbb{F}^b$ is an $\FF$-linear map. We say that $\Pi$ has {\em additive reconstruction} over $\F$, or simply that $\Pi$ is {\em additive}, if $b=b_j=1$ for all $j$ and $\Rec(y^{(1)},\ldots,y^{(k)})=y^{(1)}+\ldots+y^{(k)}$. 

Finally, we say that $\Pi$ is \emph{linear} if it has linear reconstruction and in addition, $\cX = \mathbb{F}$ and 
$\Share(x,\vec{r})$ is an $\mathbb{F}$-linear function of $x$ and a random vector $\vec{r}$ with i.i.d. uniform entries in $\mathbb{F}$. 
Notice that we never require $\Eval$ to be linear.
\end{definition}

The main focus of this work is on the communication complexity of an HSS scheme.  We formalize this with the following definitions.

\begin{definition}[Upload and download costs and rate]
\label{def-rate}
Let $k,t$ be integers and let $\mathcal{F}=\{f:\mathcal{X}^m\rightarrow\mathcal{Y}\}$ be a function class. 
Let $\Pi$ be a $k$-server $t$-private HSS for $\mathcal{F}$.
Suppose that the input shares for $\Pi$ are $x_i^{(j)}$ for $i \in [m], j \in [k]$, and that the output shares are $y^{(1)}, \ldots, y^{(k)}$. 
We define
\begin{itemize}
    \item The \emph{upload cost} of $\Pi$, $\mathrm{UploadCost}(\Pi)=\sum_{i=1}^m\sum_{j=1}^k|x_i^{(j)}|$.

    \item The \emph{download cost} of $\Pi$, $\mathrm{DownloadCost}(\Pi)=\sum_{j=1}^k|y^{(j)}|$.
    \item The \emph{download rate} (or just {\em rate}) of $\Pi$, 
    \[
    \mathrm{Rate}(\Pi)=\frac{ \log_2|\mathcal{Y}|}{\mathrm{DownloadCost}(\Pi)}
    \]
\end{itemize}
\end{definition}

\paragraph{Symmetrically private HSS.} Several applications of HSS motivate a {\em symmetrically private} variant in which the output shares $(y^{(1)},\ldots,y^{(k)})$ reveal no additional information about the inputs beyond the output of $f$. Any HSS with {\em linear reconstruction} (Definition~\ref{def:linear}) can be modified to meet this stronger requirement without hurting the download rate (and with only a small increase to the upload cost) via a simple randomization of the output shares. We further discuss this variant in Section~\ref{sec:nonlinear}.


\paragraph{Private information retrieval.} 
Some of our HSS results have applications to  \emph{private information retrieval} (PIR)~\cite{Chor:1998:PIR:293347.293350}.
A $t$-private $k$-server PIR protocol allows a client to retrieve a single symbol from a database in $\mathcal{Y}^N$, which is replicated among the servers, such that no $t$ servers learn the identity of the retrieved symbol. 
Note that such a PIR protocol reduces to 
a $t$-private $k$-server HSS scheme for the family $\ALL_\mathcal{Y}$ of \emph{all} functions $f:[N]\rightarrow\mathcal{Y}$, where the number of inputs is $m=1$. 
Indeed, in order to retrieve the $i$'th symbol $f(i)$ from a database represented by the function $f$, the client may use an HSS to share the input $x=i$ among the $k$ servers; each server computes $\Eval$ on their input share and the database and sends the output share back to the client; the client then runs $\Rec$ in order to obtain $f(i)$. The download rate and upload cost of a PIR protocol are defined as in Definition~\ref{def-rate}.



%

\subsection{Linear Secret Sharing Schemes}

In this section we define and give common examples of (information theoretic) linear secret sharing schemes (LSSS), with secrets from some finite field $\mathbb{F}$. We consider a generalized \emph{linear multi-secret sharing scheme} (LMSSS) notion, which allows one to share multiple secrets. 

\begin{definition}[LMSSS]\label{def:lmsss}
Let $\Gamma, \mathcal{T} \subseteq 2^{[k]}$ be monotone (increasing and decreasing, respectively)\footnote{We say that $\Gamma$ and $\mathcal{T}$ are monotone (increasing and decreasing, respectively) if $Q \subseteq Q'$ and $Q \in \Gamma$ then $Q' \in \Gamma$; and if $T' \subseteq T$ and $T \in \mathcal{T}$ then $T' \in \mathcal{T}$.} collections of subsets of $[k]$, so that $\mathcal{T} \cap \Gamma = \emptyset$.
A $k$-party $\ell$-LMSSS $\mathcal{L}$ over a field $\F$ with access structure $\Gamma$ and adversary structure $\mathcal{T}$ is specified by numbers $e,b_1,\ldots,b_k$ and a linear mapping $\Share:\F^{\ell}\times\F^{e}\rightarrow\F^{b_1}\times\ldots\times\F^{b_k}$ so that the following holds.
\begin{itemize}
    \item Correctness: For any qualified set $Q=\{j_1,\ldots,j_m\}\in\Gamma$ there exists a linear reconstruction function $\Rec_Q:\F^{b_{j_1}}\times\ldots\times\F^{b_{j_m}}\rightarrow\F^{\ell}$ such that for every $\vec{x}\in\F^{\ell}$ we have that $\Pr_{\br\in\F^e}[\Rec_Q(\Share(\vec{x},\br)_Q)=\vec{x}]=1$, where $\Share(\vec{x},\vec{r})_Q$ denotes the restriction of $\Share(\vec{x},\vec{r})$ to its entries indexed by $Q$. 
    \item Privacy: For any unqualified set $U\in\mathcal{T}$ and secrets $\vec{x},\vec{x}'\in\F^{\ell}$, the random variables $\Share(\vec{x},\vec{r})_U$ and $\Share(\vec{x}',\vec{r})_U$, for uniformly random $\br\in\F^e$, are identically distributed.
\end{itemize}
If $\mathcal{T}$ 
contains all sets of size 
at most $t$ (and possibly more), we say that $\mathcal{L}$ is \emph{$t$-private}.
If $\ell=1$ we simply call $\mathcal{L}$ an \emph{LSSS}, and we refer to the $\ell$-LMSSS obtained via $\ell$ independent repetitions of $\mathcal{L}$ as \emph{$\ell$ instances of $\mathcal{L}$}. 
Finally, we define the \emph{information rate} of $\mathcal{L}$ to be $\ell/(b_1+\ldots+b_k)$. 

\end{definition}

\emph{Additive sharing} is an important example of an LSSS.
\begin{example}[Additive sharing]\label{additive}
The {\em additive sharing} of a secret $x \in \FF$ is a $(k-1)$-private LSSS with $\Gamma=[k]$, $e=k-1$ and $b_j = 1$ for all $j \in [k]$. It is defined as follows. 
\begin{itemize}
    \item \textbf{Sharing.} Let $\Share(x,\br) = (r_1,r_2,\ldots,r_{k-1},x-r_1-\ldots-r_{k-1})$. Note that the shares are uniformly distributed over $\FF^k$ subject to the restriction that they add up to $x$.
    \item \textbf{Reconstruction.}  Let $\Rec_{[k]}(x^{(1)},\ldots,x^{(k)})=x^{(1)}+\ldots+x^{(k)}$.
\end{itemize}
\end{example}

We now define two standard LSSS's and associated $\ell$-LMSSS's we will use in this work: the so-called ``CNF scheme''~\cite{ISN89}  (also referred to as replicated secret sharing) and Shamir's scheme~\cite{Shamir79}.

\begin{definition}[$t$-private CNF sharing]\label{cnf}
The $t$-private $k$-party {\em CNF 
sharing} of a secret $x \in \FF$ 
is an LSSS with parameters $e = {k \choose t}-1$ and $b_j = {k-1 \choose t}$ for all $j \in [k]$. (We use {\em $t$-CNF} when $k$ is clear from the context.) It is defined as follows. 
\begin{itemize}
    \item \textbf{Sharing.} Using a random vector $\br \in \F^e$, we first {\em additively} share $x$ by choosing ${k \choose t}$ random elements of $\FF$, $x_T$, so that $x = \sum_{T\subseteq[k] : |T| =t} x_T$.  Then we define 
    $\Share(x,\br)_j = (x_T)_{j \not\in T}$ for $j \in [k]$. 
    \item \textbf{Reconstruction.} 
    Any $t+1$ parties together hold all of  the additive shares $x_T$, and hence can recover $x$.  This defines $\Rec_Q$ for $|Q|>t$.
\end{itemize}


\end{definition}
We note that there is a trivial $\ell$-LMSSS variant of $t$-private CNF sharing, as per 
Definition~\ref{def:lmsss}, which shares $\ell$ secrets with $\ell$ independent instances of CNF sharing.


\begin{definition}[$t$-private Shamir sharing]\label{shamir}
Let $\FF$ be a finite field and let $\EE \supseteq \FF$ be an extension field (typically, the smallest extension field so that $|\EE| > k$), and suppose that $s = [\EE:\FF]$ is the degree of $\EE$ over $\FF$.  Fix distinct evaluation points $\alpha_0, \alpha_1, \ldots, \alpha_k \in \EE$.  The $t$-private, $k$-party {\em Shamir sharing} of a secret $x \in \FF$
(with respect to $\EE$ and the $\alpha_i$'s) is an LSSS with parameters $e = t\cdot s$ and $b_i = s$ for all $i \in [k]$, defined as follows.
\begin{itemize}
    \item \textbf{Sharing}. Let $x \in \FF$ and let $\br \in \FF^{ts}$.  We may view $\br$ as specifying $t$ random elements of $\EE$, and we use $\br$ to choose a random polynomial $p \in \EE[X]$ so that $\deg(p) \leq t$ and so that $p(\alpha_0) = x$.  Then $\Share(x, \br)_j = p(\alpha_j)$.
    \item \textbf{Reconstruction}.  Any $t+1$ parties together can obtain $t+1$ evaluation points of the random polynomial $p$, and hence can recover $x = p(\alpha_0)$ by polynomial interpolation.  This constitutes the $\Rec$ function.
\end{itemize}


\end{definition}

\begin{remark}[$\ell$-LMSSS variants of Shamir sharing]\label{rem:ell-shamir}
The definition above is for an LSSS ($1$-LMSSS).
There are several $\ell$-LMSSS varints of $t$-private Shamir sharing.  In particular: 
\begin{itemize}
    \item[(a)] The first variant is the trivial $\ell$-LMSSS variant of $t$-private Shamir sharing where each of $\ell$ secrets are shared independently.  As per Definition~\ref{def:lmsss},
    we refer to this as ``$\ell$ instances of $t$-private Shamir sharing''.
    \item[(b)] The second (and third) variants are where $\ell=k-t$ secrets are encoded as different evaluation points of a polynomial with degree $\ell+t-1$ (requiring $|\mathbb{E}|>2k-\ell$), or, alternatively, different coefficients (requiring $|\mathbb{E}|>k$). These two $\ell$-LMSSS variants of Shamir sharing (the first of which is sometimes referred to as the Franklin-Yung scheme~\cite{FY92}) have an information rate of $\frac{ \ell \log |\FF| }{k \log|\EE|} = \frac{1-t/k}{s}$. 
\end{itemize}
\end{remark}

\paragraph{Local share conversion.} 
Informally, local share conversion allows the parties to convert from one LMSSS to another without communication. That is, the conversion maps {\em any} valid sharing of $\bx$ using a source scheme $\cal L$ to {\em some} (not necessarily random) valid sharing of  $\bx$ (more generally, some function $\psi(\bx)$), according to the target scheme $\cal L'$. 
Formally, we have the following definition, which extends the definitions of \cite{CDI05,BeimelIKO12} to multi-secret sharing.
\begin{definition}[Local share conversion]
Suppose that $\mathcal{L} = (\Share, \Rec)$  is a $k$-party $\ell$-LMSSS with parameters $(e,b_1, \ldots, b_k)$, and suppose that $\mathcal{L}' = (\Share', \Rec')$ is a $k$-party $\ell'$-LMSSS with parameters $(e',b'_1, \ldots, b'_k)$.   Let $\psi: \FF^\ell \to \FF^{\ell'}$.  A \emph{local share conversion} from $\mathcal{L}$ to $\mathcal{L}'$ with respect to $\psi$ is given by functions $\varphi_i: \FF^{b_i} \to \FF^{b_i'}$ for $i \in [k]$, so that for any secret $\bx \in \FF^\ell$, for any $\br \in \FF^e$, there is some $\br' \in \FF^{e'}$ so that
\[(\varphi_1(\Share(\bx, \br)_1), \ldots , \varphi_k( \Share(\bx, \br)_k) ) = \Share'( \psi(\bx), \br'). \]
If there is a local share conversion from $\mathcal{L}$ to $\mathcal{L}'$ with respect to $\psi$, we say that $\mathcal{L}$ is \emph{locally convertible with respect to $\psi$} to $\mathcal{L}'$.  When $\psi$ is the identity map, we just say that $\mathcal{L}$ is locally convertible to $\mathcal{L}'$.
\end{definition}

It was shown in \cite{CDI05} that $t$-private CNF sharing 
can be locally converted to any LSSS $\mathcal{L}'$ which is (at least) $t$-private. Formally:

\begin{theorem}[\cite{CDI05}]\label{cdi}
Let $\mathcal{L}$ be the $t$-private $k$-party CNF LSSS over a finite field $\F$ (Definition~\ref{cnf}).  Then $\mathcal{L}$ is locally convertible (with respect to the identity map $\psi$) to any $t$-private LSSS $\mathcal{L}'$ over $\F$.
\end{theorem}

We will use a natural extension of this idea: that $\ell$ instances of $k$-server CNF can be jointly locally converted to any $k$-server $\ell$-LMSSS with appropriate adversary structure.

\begin{corollary}\label{cdiplusplus}
Let $\mathcal{L}$ be the $k$-party $\ell$-LMSSS given by $\ell$ instances of $t$-CNF secret sharing over $\FF$ (Definition~\ref{cnf}).  Then $\mathcal{L}$ is locally convertible to any $t$-private $k$-party $\ell$-LMSSS $\mathcal{L}'$ over $\FF$.
\end{corollary}

\begin{proof}
Observe that we may obtain an LSSS $\mathcal{L}'_i$ for each $i \in [\ell]$ from $\mathcal{L}'$ by considering the LSSS that uses $\mathcal{L}'$ to share $(0, \ldots, 0, x_i, 0, \ldots, 0)$, where $x_i$ is in the $i'th$ position. Note that each $\mathcal{L}'_i$ 
is also $t$-private, by definition of an LMSSS. 
Now, consider the secret-sharing scheme $\mathcal{L}_i$ that shares $(0, 0, \ldots, 0, x_i, 0, \ldots, 0)$ using $\mathcal{L}$; this is just the standard $t$-CNF LSSS.  Thus, we may apply Theorem \ref{cdi} to locally convert $\mathcal{L}_i$ to $\mathcal{L}_i'$ for each $i \in [\ell]$.  Finally, each party adds up element-wise its shares of all schemes $\mathcal{L}_i'$ to obtain, by linearity, a sharing of $(x_1,\ldots,x_\ell)$ according to $\mathcal{L}'$.
\end{proof}

\section{Linear HSS for Low-Degree Polynomials}\label{sec:linear}

In this section, we give both positive and negative results for linear HSS for low-degree multivariate polynomials. 
Before we get into our results, we state a few useful definitions.
\begin{definition}
\label{def:POLY}
Let $m > 0$ be an integer and let $\FF$ be a finite field.  We define
\[
\POLY_{d,m}(\FF)=\{f\in\mathbb{F}[X_1,\ldots,X_m]\,:\, \deg(f) \leq d\}\]
to be the class of all $m$-variate degree-at-most-$d$ polynomials over $\FF$. When $m$ and $\FF$ are clear from context, we will just write $\POLY_d$ to refer to $\POLY_{d,m}(\FF)$.
\end{definition}

The class $\POLY_d^\ell$ may be interesting even when $d=1$.   In this case, the problem can be reduced to ``HSS for concatenation.'' That is, we are given $\ell$ secrets $x_1, \ldots, x_\ell \in \FF$, shared \emph{separately}, and we must locally convert these shares to small joint shares of $\bx = (x_1, \ldots, x_\ell)$. (To apply this towards HSS for $\POLY_1^\ell$, first locally compute shares of the $\ell$ outputs from shares of the inputs, and then apply HSS for concatenation to reconstruct the outputs.)  Formally, we have the following definition.

\begin{definition}[HSS for concatenation]
\label{def:CONCAT}
Let $\mathcal{X}$ be any alphabet and let $\mathcal{Y} = \mathcal{X}^\ell$. 
 We define $f:\mathcal{X}^\ell \to \mathcal{Y}$ to be the identity map, and
$\CONCAT_\ell(\mathcal{X}) = \{f\}$.
We refer to an HSS for $\CONCAT_\ell(\mathcal{X})$ as \emph{HSS for concatenation}.  
\end{definition}

Note that we view $m=\ell$ as the number of inputs, and so an HSS for concatenation must share each input $x_i \in \mathcal{X}$ independently. Also note that a linear HSS for $\CONCAT_\ell(\FF)$ is equivalent to a linear HSS for $\POLY_{1,1}(\FF)^\ell$.
While HSS for concatenation seems like a basic primitive, to the best of our knowledge it has not been studied before. In this work we will need HSS-for-concatenation with specific kinds of $\Share$ functions. Without this restriction, the problem has a trivial optimal solution discussed in the following remark.

\begin{remark}[Trivial constructions of HSS-for-concatenation]\label{rem:triv_concat}
Given  any $\ell$-LMSSS $\mathcal{L}$, there is a simple construction of an HSS-for-concatenation $\Pi = (\Share, \Eval, \Rec)$ with the same information rate as $\mathcal{L}$.  In more detail, suppose that $\mathcal{L}$ has share function $\Share'$ and reconstruction function $\Rec'$.  Given input $x \in \mathcal{X}$, we define $\Share$ by sharing $(x,0,\ldots,0),\ldots,(0,\ldots,0,x)$ separately using $\Share'$. Now, suppose the servers are given shares $\mathbf{y}_i^{(j)}$ of inputs $x_i$ for $i \in [\ell]$ and $j \in [k]$. Then we define $\Eval$ so that server $j$'s output share is $z^{(j)}=\sum_i y_i^{(j)}$, where $y_i^{(j)}$ is the part of $\mathbf{y}_i^{(j)}$ corresponding to the $i$'th instance of $\Share'$.   Finally, we define
\[ \Rec(z^{(1)}, \ldots, z^{(k)}) = \Rec'_{[k]}(z^{(1)}, \ldots, z^{(k)}). \]
The correctness of this scheme follows from the linearity of $\Rec'_{[k]}$. 

While this simple construction results in optimal linear HSS-for-concatenation, our applications depend on HSS-for-concatenation with particular $\Share$ functions (like the LSSS versions of $t$-CNF or Shamir).  Unfortunately, the share functions arising from natural $\ell$-LMSSS schemes (like the $\ell$-LMSSS version of Shamir in Remark~\ref{rem:ell-shamir}(b)) will not work in our applications (c.f. Remark~\ref{rem:triv_concat_wont_work}).
\end{remark}



\subsection{Constructions of Linear HSS for Low-Degree Polynomials}\label{sec:linear-constructions}

In this section, we give two constructions of HSS for $\POLY_d^\ell$, one based on CNF sharing in Section~\ref{sec:HSS-cnf} and one on Shamir sharing in Section~\ref{sec:HSS-shamir}; we will show later in Section~\ref{sec:linearNegative} that both of these have optimal download rate among linear schemes for sufficiently large $\ell$.  We further compare the two schemes in Section~\ref{sec:comparison}, and we work out an application to PIR in Section~\ref{sec:linearPIR}.

\subsubsection{Linear HSS from CNF sharing}\label{sec:HSS-cnf}

Our main result in this section is a tight connection between HSS-for-concatenation (that is, for $\CONCAT_\ell$) 
and linear codes.  More precisely, we show in Theorem~\ref{thm:cnf-hss-d1} that the best trade-off between $t$ and the download rate $R$ for HSS-for-concatenation (assuming that the output of $\Eval$ has the same size for each party) is captured by the best trade-off between distance and rate for certain linear codes.  Then we extend this result to obtain constructions of HSS for $\POLY_d^\ell$.

The basic idea is to begin with CNF sharing, and then use Corollary~\ref{cdiplusplus} to convert the CNF shares to 
shares of some $\ell$-LMSSS $\mathcal{L}$.  The final outputs of $\Eval$ will be the shares of $\mathcal{L}$, so the total download rate of the resulting HSS scheme will be the information rate of $\mathcal{L}$.  
Thus, we need to study the best $t$-private $\ell$-LMSSS with access structure $[k]$.  To do this, we utilize a connection between such LMSSS and linear error correcting codes.

Let $C \subseteq (\FF^b)^k$ be an $\F$-linear subspace of $(\FF^b)^k$.  We call $C$ an \emph{$\F$-linear code with alphabet $\FF^b$ and block-length $k$.}\footnote{Codes that are linear over a field $\FF$ and have alphabet $\FF^b$ show up often in coding theory; for example, folded codes and multiplicity codes both have this feature.  Moreover, any linear code over $\FF_{q^b}$ can be viewed as an $\FF_q$-linear code over $\FF_q^b$, by replacing each symbol $\alpha \in \FF_{q^b}$ with an appropriate vector in $\FF_q^b$.}
We define the \emph{rate} of such a code $C$ to be
\[ \mathrm{Rate}(C) = \frac{\mathrm{dim}_\FF(C)}{bk}\]
and the distance to be
\[ \mathrm{Dist}(C) = \min_{c \neq c' \in C} |\{ i \in [k] \,:\, c_i \neq c'_i \}|. \]

It is well-known that LMSSSs can be obtained from error correcting codes with good dual distance (see, e.g., \cite{Mas95,CCGHV07}).  However, in order to construct HSS we are interested only in LMSSS with access structure $\Gamma = [k]$, which results in a particularly simple correspondence.  We record this correspondence in the lemma below. (The special case where $b=1$ follows from~\cite{GM10}.)

\begin{lemma}[Generalizing Theorem 1 of~\cite{GM10}]\label{folded-codes}
Let $\ell < bk$.
There is a $t$-private $k$-party $\ell$-LMSSS $\mathcal{L}$ over $\FF$ with shares in $\FF^b$ and access structure $\Gamma = [k]$ (in particular, with rate $\frac{\ell}{kb}$)
 if and only if there is an $\FF$-linear code $C \subseteq (\FF^b)^k$ of information rate $R \geq \frac{\ell}{kb}$ and distance at least $t+1$.


\end{lemma}

\begin{proof}

\newcommand{\chop}{\mathrm{chop}}
We begin with some notation.
Given a vector $\bv \in \FF^{bk}$, we will write $\chop(\bv) = (\bv^{(1)}, \bv^{(2)}, \ldots, \bv^{(k)})$, where each $\bv^{(i)} \in \FF^b$, to denote the natural way of viewing $\bv$ as an element of $(\FF^b)^k$.
Similarly, given $H \in \FF^{e \times bk}$, we write $H = [H^{(1)} | H^{(2)} | \cdots | H^{(k)} ]$ where $H^{(i)} \in \FF^{e \times b}$.

Given a matrix $H \in \FF^{e \times kb}$, we may define a corresponding $\FF$-linear code $C_H \subset (\FF^b)^k$ by
\begin{equation}\label{eq:CH}
C_H := \{ \chop(\vec{c}) \,:\, \vec{c} \in \FF^{bk} \,:\, H\vec{c} = 0 \}.
\end{equation}
 Conversely, any $\FF$-linear code $C \subseteq (\FF^b)^k$ can be written as $C_H$ for some (not necessarily unique) $H \in \FF^{e \times bk}$ for some $e$.  

Now we prove the ``only if'' direction of the lemma. Let $C \subseteq (\FF^b)^k$ be a code as in the lemma statement.  By assumption, the dimension of the code is $Rkb \geq \ell$.  Since $C$ is linear, we may write $C = C_H$ as in \eqref{eq:CH} for some full-rank matrix $H \in \FF^{e \times bk}$, where $e = bk(1 - R)$.
Let $G \in \FF^{bk \times \ell}$ be any matrix so that $[G | H^T] \in \FF^{bk \times (\ell + e)}$ has rank $\ell + e$.  Notice that this exists because $H$ is full-rank and \[e + \ell = bk(1 - R) + \ell \leq bk(1 - R) + bkR = bk.\]  Then consider the $\ell$-LMSSS $\mathcal{L}$ given by
\[ \Share(\bx, \mathbf{r}) = \chop(G\bx + H^T \mathbf{r}), \]
where $\bx \in \FF^\ell$ contains the secrets and $\br \in \FF^e$ is a uniformly random vector.  We claim that $\mathcal{L}$ is a $t$-private $\ell$-LMSSS with access structure $\Gamma = [k]$.  To see that it is $t$-private, note that the distance of $C_H$ implies that there is no vector $\bv \in \FF^{bk}$ so that $\chop(\bv)$ has  weight at most $t$ (that is, at most $t$ of the $k$ ``chunks'' of $\bv$ are nonzero) and so that $H\bv = 0$.  This implies that for any $i_1, \ldots, i_t \in [k]$, the matrix $[H^{(i_1)} | H^{(i_2)} | \ldots | H^{(i_t)} ]$ is full rank.  This in turn implies that any $t$ elements of $\chop(H^T \br)$ are uniformly random, which implies that any $t$ elements of $\Share(\bx, \br)$ are uniformly random.  Thus $\mathcal{L}$ is $t$-private.  The fact that $\mathcal{L}$ admits access structure $\Gamma = [k]$ follows from the fact that $[G | H^T]$ is full rank.  

The ``if'' direction follows similarly.  Suppose we have an $\ell$-LMSSS $\mathcal{L}$ as in the theorem statement.  From the definition of LMSSS, we may write
$ \Share(\bx, \br) = \chop(G\bx + H^T \br)$
for some matrix $G \in \FF^{\ell \times bk}$ and some matrix $H \in \FF^{bk \times e}$ for some $e$.  Without loss of generality, we may assume that $H$ has full rank, since otherwise we may replace it with a full rank matrix with fewer rows.  Since $\mathcal{L}$ has access structure $\Gamma = [k]$, the matrix $[G|H^T]$ must have rank $\ell +e$. 
This implies that $e \leq bk - \ell$.  Consider the code $C_H$ defined by $H$.  The rate of $C_H$ is $(bk - e)/(bk) \geq \ell/bk,$ which is the information rate of $\mathcal{L}$.  Finally, by the same logic as above, the fact that $\mathcal{L}$ is $t$-private implies that $C_H$ has distance at least $t$.
\end{proof}

%
%
%

With the connection to error correcting codes established, we can now use well-known constructions of error correcting codes in order to obtain good HSS schemes. 
We begin with a statement for HSS for concatenation (Definition~\ref{def:CONCAT}).  Then we will use share conversion from $t$-private CNF to $dt$-private CNF (Lemma~\ref{cnf-product} below) to extend this to $\POLY_d^\ell$ for $d > 1$.

\begin{theorem}[Download rate of HSS-for-concatenation is captured by the distance of linear codes]\label{thm:cnf-hss-d1}
Let $\ell, k$ be integers, and let $\FF$ be a finite field.  Suppose that $b > \ell/k$.

There is a $t$-private, $k$-server linear HSS for $\CONCAT_\ell(\FF)$ 
with output shares in $\FF^b$ (and hence with download rate $\ell/(kb)$) if and only if
there is an $\FF$-linear code $C \in (\FF^b)^k$  with rate at least $\ell/(kb)$ and distance at least $t+1$.

Further, for the ``if'' direction, the HSS scheme guaranteed by the existence of $C$ uses $t$-CNF sharing and has upload cost $k\ell {k - 1\choose t} \log_2|\FF|$.
\end{theorem}
\begin{proof}
For the ``if'' direction (a code $C$ implies an HSS-for-concatenation $\Pi$), we follow the outline sketched above.  
Our goal is to share each input $x_i$ independently with CNF sharing and to reconstruct $\bx = (x_1, \ldots, x_\ell)$ from the outputs of $\Eval$.
Let $C$ be a code as in the theorem statement, and define an HSS $\Pi = (\Share, \Eval, \Rec)$ as follows.
\begin{itemize}
    \item \textbf{Sharing.} 
Let $\Share(x) = (\bx^{(1)}, \ldots, \bx^{(k)})$ be $t$-CNF sharing of $x$, as in Definition~\ref{cnf}.

    \item \textbf{Evaluation.}  Given $t$-CNF shares $\bx^{(j)} = (\bx_1^{(j)}, \ldots, \bx_\ell^{(j)})$ for each $j \in [k]$, we use Corollary~\ref{cdiplusplus} to locally convert these to shares $\by^{(j)} \in \FF^b$ of an $\ell$-LMSSS $\mathcal{L}$ for the secret $\bx = (x_1, \ldots, x_\ell)$, where $\mathcal{L}$ is the $\ell$-LMSSS that is guaranteed by Lemma~\ref{folded-codes} and the existence of the code $C$.  We define  $\Eval(f, j, \bx^{(j)}) = \by^{(j)}$.
    \item \textbf{Reconstruction.} By construction, we have \[ \Rec_{[k]}(\by^{(1)}, \ldots, \by^{(k)}) = \bx, \]
    where $\Rec_{[k]}$ is the reconstruction algorithm from $\mathcal{L}$.  This defines the reconstruction algorithm $\Rec$ for $\Pi$.
\end{itemize}
The parameters for the ``if'' and ``further'' parts of the theorem follow by tracing the parameters through Corollary~\ref{cdiplusplus} and Lemma~\ref{folded-codes}.

Now we prove the ``only if'' direction.  Suppose that $\Pi = (\Share, \Eval, \Rec)$ is an HSS as in the theorem statement.  Consider the $\ell$-LMSSS $\mathcal{L}$ that shares $\bx = (x_1, \ldots, x_\ell) \in \FF^\ell$ as $\Share'(\bx, \br) = (\by^{(1)}, \ldots, \by^{(k)})$, where $\by^{(j)}$ is the $j$'th output of $\Eval$ when $\bx$ is shared using $\Share$.  Since $\Rec(\by^{(1)}, \ldots, \by^{(k)}) = \bx$ is linear, this secret sharing scheme has a linear reconstruction algorithm.  By \cite[Claims 4.3, 4.7, 4.9]{beimelThesis}, any secret sharing scheme with a linear reconstruction algorithm is equivalent to a scheme with a linear share function; this argument can be extended to $\ell$-LMSSS. Thus there is a linear share function $\Share''$ so that $\mathcal{L} = (\Share'', \Rec)$ forms a $t$-private, $k$-party $\ell$-LMSSS with information rate $\ell/(bk)$.  By Lemma~\ref{folded-codes}, there is a code $C \subseteq (\FF^b)^k$ with distance at least $t+1$ and rate at least $\ell/(bk)$.  This establishes the ``only if'' direction.
\end{proof}

We can easily extend the ``if'' direction (the construction of HSS) to $\POLY_{d,m}^\ell(\FF)$ by using the following lemma, which shows that $t$-CNF shares of $d$ secrets can be locally converted to valid $dt$-CNF shares of the product of the secrets via a suitable assignment of monomials to servers. A similar monomial assignment technique was used in the contexts of communication complexity~\cite{BGKL}, private information retrieval~\cite{BIK}, and secure multiparty computation~\cite{Maurer}.

\begin{lemma}\label{cnf-product} 
Let $\mathcal{L}$ be the $d$-LMSSS that is given by $d$ instances of $t$-private, $k$-party CNF sharing.
Let $\mathcal{L}'$ be the LSSS given by $dt$-private, $k$-party CNF sharing.  Let  $\psi(x_1, \ldots, x_d) = \prod_{i=1}^d x_i$.
Then $\mathcal{L}$ is locally convertible to $\mathcal{L}'$ with respect to $\psi$.
\end{lemma}

\begin{proof}
We describe the maps $\varphi_i:\FF^{d{k-1 \choose t}} \to \FF^{{k-1 \choose dt}}$ for each $i \in [d]$ that will define the local share conversion.  
Suppose that the original secrets to be shared with $\mathcal{L}$ are $x_1, \ldots, x_d \in \FF$.  Under $\mathcal{L}$, party $j$'s share consists of $\by^{(j)} = (x_{i,S} \,:\, i \in [d], j \not\in S),$  where for $S \subseteq [k]$ of size $t$, the random variables $x_{i,S} \in \F$  are uniformly random so that $x_i = \sum_S x_{i,S}$ for all $i \in [d]$.

For each set $W \subseteq [k]$ of size $|W| \leq dt$, define
\[ y_W = \sum_{S_1, \ldots, S_d : \bigcup_i S_i = W} \prod_{i=1}^d x_{i,S_i}, \]
where the sum is over all sets $S_1, \ldots, S_d \subseteq [k]$ of size $t$ whose union is exactly $W$.  Notice that
\begin{align*}
\sum_{W \subseteq [k], |W| \leq dt} y_W &=
\sum_{S_1, \ldots, S_d \subset [k]: |S_i| = t} \prod_{i=1}^d x_{i,S_i} \\
&= \prod_{i=1}^d\left( \sum_{S \subset [k] : |S| = t} x_{i,S} \right)
= \prod_{i=1}^d x_i 
= \psi(\bx).
\end{align*}

Now, for all $W \subset [k]$ with $|W| \leq dt$, arbitrarily assign to $W$ some set $T = T(W) \subset [k]$ of size exactly $dt$, so that $W \subseteq T$.  For $T \subseteq [k]$ of size exactly $dt$, define
\[ z_T = \sum_{W : T=T(W)} y_W, \]
where the sum is over all sets $W \subseteq [k]$ of size at most $dt$ so that $T(W)$ is equal to $T$.
By the above, we have
\[ \sum_{T \subseteq [k] : |T| = dt} z_T = \psi(\bx). \]

Further, each party $j$ can locally compute $z_T$ for each $T$ so that $j \not\in T$.  This is because $z_T$ requires knowledge only of $x_{i,S}$ for $i \in [d]$ and for $S \subseteq T$.  If $j \not\in T$, then $j \not\in S$, and therefore $x_{i,S}$ appears as part of the share $\by^{(j)}$. 
Thus, we define
\[ \varphi_j(\by^{(j)}) = (z_T \,:\, T \subseteq [k], |T| = dt, j \not\in T), \]
and the shares $\varphi_j(\by^{(j)})$ are legitimate shares of $\psi(\bx)$ under $\mathcal{L}'$. 
\end{proof}

Using Lemma~\ref{cnf-product} along with Theorem~\ref{thm:cnf-hss-d1} implies the following extension to $\POLY^\ell_d$.

\begin{theorem}[Download rate of HSS-for-polynomials follows from distance of linear codes]\label{thm:cnf-hss} 
Let $\ell, t,k, d, m$ be integers, and let $\FF$ be a finite field.  Suppose that for some integer $b > \ell/k$,
there is an $\FF$-linear code $C \in (\FF^b)^k$  with rate at least $\ell/(kb)$ and distance at least $dt+1$.  Then there is a $t$-private, $k$-server linear HSS for $\POLY_{d,m}(\FF)^\ell$ with download rate $\ell/(kb)$ and upload cost  
$k \ell m {k - 1\choose t} \log_2|\FF|$. 



\end{theorem}

\begin{proof}

By linearity, we may assume without loss of generality that the function $f \in \POLY_{d,m}(\FF)^\ell$ is a vector of $\ell$ {\em monomials} given by $f_r:\FF^m \to \FF$ for $r \in [\ell]$, so that that $f_r$ is a monomial of degree at most $d$.
We define a $t$-private, $k$-server linear HSS $\Pi' = (\Share', \Eval', \Rec')$ for $\POLY_{d,m}(\FF)^\ell$ with download rate $\ell/(kb)$ as follows.

Let $\Pi=(\Share,\Eval,\Rec)$ be the $dt$-private, $k$-server linear HSS for $\CONCAT_\ell(\FF)$ with output shares in $\FF^b$ and download rate $\ell/(kb)$ that is guaranteed by Theorem~\ref{thm:cnf-hss-d1}.  Consider $f \in \POLY_{d,m}(\FF)^\ell$ of the form described above.  For $r \in [\ell]$, let $\varphi_{1,r}, \ldots, \varphi_{k,r}$ be the local share conversion functions guaranteed by Lemma~\ref{cnf-product} for the monomial $f_r$.
\begin{itemize}
    \item \textbf{Sharing.} The $\Share'$ function is $t$-CNF sharing.  Suppose that we share the inputs $x_{i,r}$ for $i \in [m]$ and $r \in [\ell]$ as shares $\by_r^{(j)} \in \FF^{m{k-1 \choose t}}$ for each party $j \in [k]$ and each $r \in [\ell]$.  Let $\by^{(j)} = (\by_1^{(j)}, \ldots, \by_\ell^{(j)})$.
    \item \textbf{Evaluation.} Define $\bz_j$ and $\Eval'$ by:
    \[ \bz^{(j)} = \Eval'(f, j, \by^{(j)}) = \Eval(g,j,(\varphi_{j,1}(\by^{(j)}_1), \ldots, \varphi_{j,\ell}(\by^{(j)}_\ell)) ) \]
    for all $j \in [k]$ and for $\by^{(j)}$ as above, where $g$ is the identity function.
    \item \textbf{Reconstruction.} 
    Define
    \[ \Rec'(\bz_1, \ldots, \bz_k) = \Rec(\bz_1, \ldots, \bz_k). \]
\end{itemize} To see why this is correct, notice that
    by construction, $\{\varphi_{j,r}(\by^{(j)}_r)\,:\, j \in [k]\}$ are $dt$-CNF shares of the secret $f_r(x_{1,r}, \ldots, x_{m,r})$.  Thus, again by construction, the shares $\bz^{(j)}$ are shares under some $\ell$-LMSSS with reconstruction algorithm $\Rec$ for the concatenated secrets $(f_1(\bx_1), \ldots, f_\ell(\bx_\ell))$.  Therefore, 
    $\Rec'( \bz_1, \ldots, \bz_k)$ indeed returns $(f_1(\bx_1), \ldots, f_\ell(\bx_\ell))$, the desired outcome.
\end{proof}

Below are some instantiations of Theorem \ref{thm:cnf-hss} that yield an HSS with high rate. Example \ref{ex:cnf-bigfield} gives an HSS with optimal rate but requires sufficiently large $\ell$, while Example \ref{ex:hss-ham} has worse rate (and only works for $d=2$), but uses a significantly smaller $\ell$.

\begin{example}\label{ex:cnf-bigfield}
	Let $\ell, t,k, d, m$ be integers, and let $\FF$ be a finite field. There is a $t$-private, $k$-server linear HSS for $\POLY_{d,m}(\FF)^\ell$, where $\Share$ is $t$-private CNF over $\FF$, with download rate at least $1 - dt/k$, for any $\ell$ of the form $\ell = b(k-dt)$ where $b \geq \log_{|\FF|}(k)$ is an integer.  
\end{example}

\begin{proof}
	We begin with an observation about the existence of linear codes with good rate/distance trade-offs (in particular, that meet the \emph{Singleton bound}).
	
	\begin{fact}\label{fact:rs}
		Let $k,b$ be positive integers, and suppose that $\FF$ is a finite field with $|\FF| = q$.  Suppose that $q^b \geq k$.  For any $R \in (0,1)$ so that $kR \in \mathbb{Z}$, there is an $\FF$-linear code $C \subset (\FF^b)^k$ with rate $R$ and distance $k(1 - R) + 1$.
	\end{fact}
	Indeed, we may simply take the standard Reed-Solomon code over $\FF_{q^b}$, and then concatenate with the identity code over $\FF$, replacing each field element in $\FF_{q^b}$ with an appropriate vector in $\FF_q^b$.

	To obtain the result, combine Theorem \ref{thm:cnf-hss} with the concatenated Reed-Solomon code of Fact~\ref{fact:rs}.  We choose a code of rate $R = 1 - dt/k$ and distance $dt + 1$, and we choose $b \geq \log_{|\F|}(k)$ to be any parameter.  As in the theorem statement, we suppose that $\ell = bkR = b(k-dt)$, so that the download rate of the resulting HSS is $\ell/(bk) = R.$
\end{proof}

\begin{example}\label{ex:hss-ham}
		Let $\ell, r, m$ be integers, and let $\FF$ be a finite field with $|\FF|=q$. There is a $1$-private, $(k=(q^r-1)/(q-1))$-server linear HSS for $\POLY_{2,m}(\FF)^\ell$, where $\Share$ is 1-private CNF over $\FF$, with download rate at least $1 - r/k=1-O(\log_q(k)/k)$, for $\ell=k-r=k-O(\log_q(k))$.
\end{example}

\begin{proof}
	We begin with the following fact, which is achieved by Hamming codes.
	\begin{fact}\label{fact:ham}
		Let $r$ be a positive integer, suppose that $\FF$ is a finite field with $|\FF| = q$, and set $k=(q^r-1)/(q-1)$.  There is an $\FF$-linear code $C \subset (\FF)^k$ with dimension $k-r$, rate $1-r/k$, and distance $3$.
	\end{fact}
	Combining Theorem \ref{thm:cnf-hss} with Fact \ref{fact:ham} yields the desired result, with $\ell=k-r$ and download rate $\ell/k=1-r/k$.
\end{proof}

\subsubsection{Linear HSS from Shamir Sharing}\label{sec:HSS-shamir}

Our next construction of linear HSS is based on Shamir sharing (Definition~\ref{shamir}).  Later in Section~\ref{sec:comparison}, we compare this to the CNF-based construction from the previous section.

As noted in Section~\ref{sec:related}, HSS from Shamir sharing is related to work on the use of Reed-Solomon codes as regenerating codes.  In particular, existing results in the regenerating codes literature immediately imply the following corollary.


\begin{corollary}[Follows from~\cite{TYB17}]\label{cor:TYB}
Let $k,d,m,t$ be integers so that $k > dt$, and let $p$ be prime.  There is a field $\FF$ of characteristic $p$ and of size $|\FF| = \exp_p(\exp(k \log k))$ so that the following holds.
There is a $t$-private, $k$-server HSS for $\POLY_{d,m}(\FF)$ where $\Share$ is given by a Shamir scheme (Definition~\ref{shamir}, where both the secrets and the shares lie in $\FF$) and where the download rate is $1 - dt/k$.
\end{corollary}


\begin{remark}[Limitations of regenerating-code based HSS]\label{rem:regenCompare}
While Corollary~\ref{cor:TYB} has optimal download rate, the fact that the field $\FF$ must be so large makes it impractical.  On the one hand, if we take the corollary as written, so that $\ell=1$, each secret must lie in an extremely large extension field, with extension degree exponential in $k$.  On the other hand, it is possible to ``pack'' $\ell = \exp(k \log k)$ secrets over $\FF_p$ into a single element of the large field $\FF$; we do this in our Theorem~\ref{thm:shamir-hss}, for example.  However, the problem with this when starting with a regenerating code is that each secret must be shared independently over the large field $\FF$.  This means that the upload cost will scale like $\ell^2$, while the number of secrets is $\ell$.  Since $\ell = \exp(k \log k)$ is so large, this results in a huge overhead in upload cost.

Moreover, \cite{TYB17} shows that a field size nearly this large is necessary for any optimal-download linear repair scheme for Reed-Solomon codes.  Thus, we cannot hope to obtain a practical Shamir-based HSS with optimal download rate in a black-box way from regenerating codes.
\end{remark}
In light of Remark~\ref{rem:regenCompare}, rather than using an off-the-shelf regenerating code, we adapt ideas from the regenerating codes literature to our purpose.  In particular, our approach is inspired by \cite{GW16,TYB17} (and is in fact much simpler).
Mirroring our approach for CNF-based sharing, our main tool is an HSS-for-concatenation scheme that begins with Shamir sharing.  


\begin{theorem}[HSS-for-concatenation from Shamir]\label{reg-share-conversion}
Let $t < k$ be any integers and
let $\ell = k-t$.  Let $\FF$ be any finite field with $|\FF| \geq k$ 
and let $\EE$ be an extension field of $\FF$ of size $|\FF|^\ell$.
There is an LSSS $\mathcal{L}$ given by $t$-private $k$-party Shamir secret sharing with secrets in $\FF$ and shares in $\EE$ (as in Definition~\ref{shamir}) so that the following holds.

There is a $t$-private, $k$-party, linear (over $\FF$) HSS $\Pi = (\Share, \Eval, \Rec)$ for $\CONCAT_\ell(\FF)$ so that $\Share$ is given by $\mathcal{L}$, 
and the outputs of $\Eval$ are elements of $\FF$.  In particular, the download rate is $1 - t/k$.

\end{theorem}

\begin{proof}

Let $\EE$ be an extension field of $\FF$ so that $|\EE| = |\FF|^\ell$.  Let $\gamma$ be a primitive element of $\EE$ over $\FF$.  We choose $\mathcal{L}$ to be 
the Shamir scheme with secrets in $\FF$, shares in $\EE$, and evaluation points $\alpha_0 = \gamma$ and any $\alpha_1, \ldots, \alpha_k \in \FF$, using the notation of Definition~\ref{shamir}.

Given $\mathcal{L}$, we define $\Share$, $\Eval$, and $\Rec$ below.

\begin{itemize}
    \item \textbf{Sharing}. The share function for $\Pi$ is the same as that for $\mathcal{L}$, and is thus given by
$\Share(x_i)_j = p_i(\alpha_j)$, where $p_i \in \EE[X]$ is a random polynomial of degree at most $t$, so that $p_i(\alpha_0) = x_i$.
Let $\bz^{(j)} := (p_1(\alpha_j), \ldots, p_\ell(\alpha_j))$.

\item \textbf{Evaluation}. 
%
First, we bundle together the $\ell$ shares to form one Shamir sharing over the larger field $\EE$.  In more detail,
let $y = \sum_{i=1}^\ell x_i \gamma^i \in \mathbb{E}$, and define $f(X) \in \EE$ by
\[ f(X) = \sum_{i=1}^\ell p_i(X) \gamma^i. \]
Note that $f(X)$ is a degree-$t$ polynomial in $\EE[X]$, and that
\[ f(\gamma) = \sum_{i=1}^\ell p_i(\gamma) \gamma^i = \sum_{i=1}^\ell x_i \gamma^i = y. \]
Thus, if each server $j$ locally computes $z^{(j)} = f(\alpha_j) \in \EE$ from its shares, they now have a valid Shamir sharing over $\EE$ of the secret $y \in \EE$.

Since $\ell = k-t$, for any polynomial $f \in \EE[X]$ of degree at most $t$, and any $r\in[\ell]$, we have $\deg(X^{r-1} \cdot f(X)) \leq (k-t-1) + t = k-1$.  Thus, for each $r \in [\ell]$, there is a linear relationship between the $k+1$ evaluations $\alpha_1^{r-1}f(\alpha_1),  \ldots, \alpha_k^{r-1}f(\alpha_k)$ and $\gamma^{r-1}f(\gamma)$, in the sense that there are some coefficients $\lambda_i \in \EE$ (independent of $f$ and $r$) so that for any $f \in \EE[X]$ of degree at most $t$, and any $r \in [\ell]$,
\[ \gamma^{r-1}f(\gamma) = \sum_{i=1}^k \lambda_i\alpha_i^{r-1} f(\alpha_i).  \]
Define 
\[ w^{(j)} = \tr( \lambda_j f(\alpha_j) ), \]
where $\tr(\cdot)$ denotes the \emph{field trace} of $\EE$ over $\FF$, defined as
$ \tr(X) = \sum_{i=0}^{\ell-1} X^{|\FF|^i}. $
We define
\[ \Eval(f,j,\bz^{(j)}) = w^{(j)}.\]
Notice that $w^{(j)} \in \FF$, since the image of $\tr$ is $\FF$.

\item \textbf{Reconstruction}. Suppose we are given $w^{(1)}, \ldots, w^{(k)}$ produced by $\Eval$.  Our goal is to recover
 $(x_1, \ldots, x_\ell) \in \FF^\ell$ in an $\FF$-linear way.

Since $\tr$ is $\FF$-linear and $\alpha_j \in \FF$ for all $j \in [k]$, we have for all $r \in [\ell]$ that
\begin{align*}
    \tr( \gamma^{r-1}f(\gamma) ) &= \tr\left( \sum_{j=1}^k \lambda_j\alpha_j^{r-1} f(\alpha_j)\right) \\
    &= \sum_{j=1}^k \alpha_j^{r-1} \tr( \lambda_j f(\alpha_j) ) \\
    &= \sum_{j=1}^k \alpha_j^{r-1} w^{(j)}
\end{align*}
Therefore, there is an $\FF$-linear map to recover each $\tr(\gamma^{r-1} f(\gamma))$ for $r \in [\ell]$, given the shares $w^{(j)}$.
Since $\gamma^0, \ldots, \gamma^{\ell-1}$ form a basis for $\EE$ over $\FF$, 
there is also an $\FF$-linear function that recovers $f(\gamma) = y$ from $\tr(\gamma^{r-1}f(\gamma))$ for $r \in [\ell]$.
Finally, since $y$ is defined as $y = \sum_{j=1}^\ell x_i \gamma^i$ and $\gamma, \gamma^2, \ldots, \gamma^\ell$ form a basis for $\EE$ over $\FF$, there is an $\FF$-linear map to recover the coefficients $x_i$ from $y$.  Composing all of these, we see that there is an $\FF$-linear map that can recover $(x_1, \ldots, x_\ell)$ given the shares $w^{(1)}, \ldots, w^{(k)}$.  This map will be $\Rec$.
\end{itemize}

Finally, we observe that the information rate of $\mathcal{L}$ is $\frac{\ell}{k} = 1 - t/k$, as there are $\ell = t-k$ elements of $\FF$ that are shared as $k$ elements of $\FF$.
\end{proof}

In order to extend Theorem~\ref{reg-share-conversion} to an HSS for $\POLY_d^\ell$, we will use the 
well-known local multiplication property of Shamir sharing (e.g., \cite{STOC:BenGolWig88,STOC:ChaCreDam88,EC:CraDamMau00}).

\begin{fact}\label{shamir-product}
Let $\FF \leq \EE$ be finite fields and let $\alpha_0, \ldots, \alpha_k \in \mathbb{E}$ be distinct.
Let $\mathcal{L}$ be the $d$-LMSSS that is given by $d$ instances of the $t$-private, $k$-party Shamir sharing scheme with secrets in $\FF$, shares in $\EE$, and evaluation points $(\alpha_0, \ldots, \alpha_k)$ (Definition~\ref{shamir}).   Let $\mathcal{L}'$ be the LSSS given by the $dt$-private, $k$-party Shamir sharing scheme with the same specifications.  Let  $\psi(x_1, \ldots, x_d) = \prod_{i=1}^d x_i$.  Then $\mathcal{L}$ is locally convertible to $\mathcal{L}'$ with respect to $\psi$.

\end{fact}


\begin{remark}[Comparison of Theorem~\ref{reg-share-conversion} to the straightforward scheme]\label{rem:triv_concat_wont_work}
As noted in Remark~\ref{rem:triv_concat}, there is a straightforward HSS-for-concatenation with a different version of Shamir sharing, which arises from the $\ell$-LMSSS of Franklin-Yung~\cite{FY92} (Remark~\ref{rem:ell-shamir}(b)).  This scheme matches the parameters of Theorem~\ref{reg-share-conversion} (and in particular is optimal for the parameters to which it applies).  However, we cannot use that straightforward scheme for our next step, which is to generalize to $\POLY_d^\ell$.  The reason is that the generalization uses Fact~\ref{shamir-product}.  This fact is not true for the sharing that arises from Remark~\ref{rem:triv_concat}, since to ensure $t$-privacy, the degree of the polynomials used to share must be $t + \ell - 1$, rather than $t$.
\end{remark}

By combining Fact \ref{shamir-product} with Theorem \ref{reg-share-conversion} we conclude the following.

\begin{theorem}[HSS-for-polynomials from Shamir]
\label{thm:shamir-hss}
Let $\FF$ be a finite field.  Let $m$ be a positive integer.  Let $b \geq \log_{|\FF|}(k)$ be a positive integer and let $\ell = b(k - dt)$.  
There is a $t$-private $k$-server linear HSS $\Pi = (\Share, \Eval, \Rec)$ for $\POLY_{d,m}(\FF)^\ell$ so that $\Share$ is Shamir sharing (Definition~\ref{shamir}), where the upload cost is $kmb^2 (k-dt)^2 \log_2|\FF|$, and the download cost is $kb \log_2|\FF|$.  Consequently, the download rate is $1 - dt/k$.

\end{theorem}

\begin{proof}
We first observe that a proof nearly identical to the proof of the first part of Theorem~\ref{thm:cnf-hss} establishes the theorem for $b = 1$, if it holds that $|\FF| \geq k$.  (Indeed, we use Theorem~\ref{reg-share-conversion} rather than Theorem~\ref{folded-codes} for the HSS-for-concatenation, and Fact~\ref{shamir-product} rather than Lemma~\ref{cnf-product} for the share conversion.)  To see the upload and download costs for the $b=1$ case, notice that each server $j \in [k]$ holds as an input share an element of $\EE$ (as in the statement of Theorem~\ref{reg-share-conversion}), which is $\log_2|\EE| = \ell \log_2|\FF|$ bits.  Thus the upload cost is $k m \ell^2 \log_2|\FF| = km(k-dt)^2 \log_2|\FF|$ bits.  Each output share is a single element of $\FF$, so the download cost is $k\log_2|\FF|$.  Thus, the download rate is \[ \frac{ \ell \log_2|\FF| }{k \log_2 |\FF| } = 1 - dt/k, \] as desired.

In order to prove the theorem for $b \geq 1$ (and hence possibly for $|\FF| < k$), let $\tilde{\FF}$ be an extension field of $\FF$ of degree $b$, so that $|\tilde{\FF}| = |\FF|^b \geq k$.  Let $\gamma$ be a primitive element of $\tilde{\FF}$ over $\FF$.  Then let $\Pi' = (\Share', \Eval', \Rec')$ be the HSS scheme over $\tilde{\FF}$ guaranteed by the $b=1$ case by the above, for $\tilde{\ell} = k-dt$.  We will define $\Pi = (\Share, \Eval, \Rec)$ as follows.
\begin{itemize}
    \item \textbf{Sharing.} Suppose the inputs to $\Pi$ are $\bx_1, \ldots, \bx_\ell \in \FF^m$. Organize these inputs as $\bx_{i,r}$ for $i \in [b]$ and $r \in [\tilde{\ell}]$. The $\Share$ function shares each $x_{i,r}$ independently by applying $\Share'$ to $(x_{i,r})_{r \in [\tilde{\ell}]}$ independently for each $i \in [b]$.  Suppose that these input shares are given by $y_{i,r}^{(j)} \in \tilde{\FF}$ for each $i \in [b], r \in [\tilde{\ell}]$ and for each server $j \in [k]$.  
    \item \textbf{Evaluation.} For $f \in \POLY_{d,m}(\FF)^\ell$ given by $(f_1, \ldots, f_{\ell})$ and re-organized as $(f_{i,r})_{i \in [b], r \in [\tilde{\ell}]}$, define
\[ z^{(j)} := \Eval(f,j, (y_{i,r}^{(j)})_{i \in [b], r \in [\tilde{\ell}]}) = \sum_{i \in [b]} \Eval'( f, j, (y_{i,r}^{(j)})_{r \in [\tilde{\ell}]} )\gamma^i, \]
so that $z^{(j)} \in \tilde{\FF}$.  
\item \textbf{Reconstruction.} Suppose the output shares are $z_1, \ldots, z_k$ as above.  Then by linearity and the correctness of $\Rec'$,
\begin{align*}
\Rec'(z^{(1)}, \ldots, z^{(k)}) &= \sum_{i \in [b]} \gamma^i \Rec'( \Eval'( f, j, ((y_{i,r}^{(j)})_{r \in [\tilde{\ell}]})_{j\in[k]}))  \\
&= \sum_{i \in [b]} \gamma^i (f_{i,r}(x_{i,r}))_{r \in [\tilde{\ell}] }.
\end{align*}
Since $\gamma, \gamma^2, \ldots, \gamma^b$ form a basis of $\tilde{\FF}$ over $\FF$, and since $(f_{i,r}(x_{i,r})_{r \in [\tilde{\ell}]} \in \FF^{\tilde{\ell}}$, given this, we can recover all of the vectors $(f(x_{i,r})_{r \in [\tilde{\ell}]} \in \FF^{\tilde{\ell}}$ for all $i \in [b]$.  Thus we, after re-arranging, we have assembled $(f_s(x_s))_{s \in [\ell]}$, as desired.  This defines $\Rec$.

\end{itemize}

It remains to verify the upload and download costs. The upload cost is $b$ times the upload cost for $\Pi'$, for a total of 
\[ b \cdot km (k-dt)^2 \log_2|\tilde{\FF}| = b^2 km (k-dt)^2 \log_2|\FF|.\]
The download cost is the same as the download cost for $\Pi'$, which is
\[ k \log_2|\tilde{\FF}| = kb \log_2|\FF|. \]
Consequently, the download rate is
\[ \frac{ \ell \log_2{|\FF|}}{kb \log_2|\FF|} = 1 - dt/k. \]

%
%
\end{proof}

Finally, note that if $b$ is a multiple of $\lceil \log_{|\FF|}(k) \rceil$ in Theorem~\ref{thm:shamir-hss}, then the upload cost can be improved to $kmb\lceil \log_{|\FF|}(k) \rceil (k-dt)^2 \log_2|\FF|$. This is achieved by initializing the scheme from Theorem~\ref{thm:shamir-hss} with $b=\lceil \log_{|\FF|}(k) \rceil$, and then repeating it $b/\lceil \log_{|\FF|}(k) \rceil$ times.

\subsubsection{Comparison between CNF-based HSS and Shamir-based HSS}\label{sec:comparison}

In this section we comment briefly on the differences between Theorem~\ref{thm:cnf-hss} and Theorem~\ref{thm:shamir-hss} in particular and on the differences between CNF and Shamir sharing more generally.

\paragraph{Comparison of our HSS schemes.}
For simplicity, we focus on our schemes for HSS-for-concatenation (Theorems~\ref{thm:cnf-hss-d1} and \ref{reg-share-conversion}), and the case where $m=1$ and where the $\Eval$ functions of both schemes output a single element of the same field (that is, $b=1$).  Thus, there are only two parameters to adjust: the number of parties $k$ and the privacy $t$.  The goal is to maximize the download rate (or equivalently to minimize the number of repetitions $\ell$).  

Note that for our Shamir-based HSS, the download rate is $1 - t/k$ whenever the scheme is defined, and in particular when $b=1$ we always take $\ell = k-t$ and enforce $|\FF| \geq k$.  Note that our CNF-based HSS scheme is defined for any field and any $\ell$, since linear codes exist for all such choices.


Both our CNF-based HSS and our Shamir-based HSS achieve the optimal download rate $1 - t/k$ when they are both defined, but there are three important differences.
\begin{itemize}
\item \textbf{Upload cost.} The main advantage of the Shamir-based HSS (Theorem~\ref{reg-share-conversion}) is that the upload cost is smaller: when $\ell=k-t$ (so that our Shamir-based HSS exists), the upload cost for our Shamir-based HSS is $k(k-t)^2 \log_2|\FF|$ bits, while the upload cost for our CNF-based HSS is $k(k-t){k-1\choose t}\log_2|\FF|$ bits.  Thus, the Shamir-based HSS has a smaller upload cost whenever $k-t < {k-1 \choose t}$.
\item \textbf{Flexibility.} The main advantage of our CNF-based HSS is that the parameter regime in which it is defined is much less restrictive. 
More concretely, the Shamir-based HSS requires that the field $\F$ have size at least $k$ (when $b=1$ as we consider here), and it requires $\ell \geq k-t$.  In contrast, the CNF-based HSS applies for any $\ell$ and over any field, as linear codes of rate $\ell/(kb)$ exist for any $\ell$ over any field.

\item \textbf{Download rate.}  As mentioned above, whenever the Shamir-based HSS is defined, both schemes have the same (optimal) download rate.  However, the CNF-based HSS applies more generally, and in the case of HSS-for-concatenation has optimal rate for any linear HSS whenever it is defined, as per Theorem~\ref{thm:cnf-hss-d1}.  

In order to get a meaningful comparison between our two HSS constructions, one can try to apply the Shamir-based HSS over a field of size less than $k$ by embedding the secrets into a field $\mathbb{K} \geq \mathbb{F}$, where $|\mathbb{K}| \geq k$.  When we do this, the CNF-based HSS attains a strictly better download rate. 

For example, suppose that $\FF = \FF_2$, and $t=2$, and suppose for simplicity that $k = 2^r - 1$ for some $r$.  Furthermore, suppose that we only allow $\ell\leq k$. In this case, the best CNF-based HSS is given by the best binary linear code with distance $3$ and length $2^r - 1$, as per Theorem~\ref{thm:cnf-hss-d1}.  This is the Hamming code of length $2^r - 1$ (see Fact \ref{fact:ham}).  Translating the parameters of the Hamming code to HSS, the CNF-based HSS has $\ell = k - r$ and hence download rate 
\[1 - \frac{r}{k} \approx 1 - \frac{\log k}{k}.\]
In contrast, the Shamir-based HSS has $\ell = k - 2$ (recalling that $t=2$ for this comparison), but is working over a bigger field of size $2^r$ and so has download rate 
\[\frac{1 - t/k}{r} \approx \frac{1 - 2/k}{\log k}. \]
\end{itemize}

\paragraph{Comparison of CNF and Shamir sharing, with any HSS scheme.}

One may wonder whether there is an inherent limitation to Shamir sharing relative to CNF sharing.  That is, could it be possible that \em some \em HSS based on Shamir sharing can match the performance of our CNF-based sharing scheme?



More precisely, we ask if there is a Shamir-based HSS that can match the performance of our CNF-based HSS, \emph{where the field $\EE$ that the Shamir shares lie in is the smallest such that $|\EE| > k$} (aka, the smallest so that Shamir sharing is defined). The reason to phrase the question this way is that if we allowed $\EE$ to be arbitrarily large, it might be the case that the Shamir-based HSS had a larger share size than the CNF-based HSS, undoing the advantage of the Shamir-based HSS.

It turns out that the answer to this question can be ``no,'' as the following example shows.


\begin{proposition}\label{prop:impossibleShamir}
Consider the problem of constructing $1$-private $5$-server HSS for $\CONCAT_4(\FF_2)$.

There is a $\F_2$-linear HSS that solves this problem with CNF sharing and download cost $5$ bits (that is, one bit per server).
However, there is \emph{no} $\F_2$-linear HSS that solves this problem with Shamir sharing and with download cost $5$, where each of the $5$ secrets in $\FF = \FF_2$ are shared independently over $\EE = \FF_8$.
\end{proposition}

We defer the proof of Proposition~\ref{prop:impossibleShamir} to Appendix~\ref{app:impossibleShamir}.
Briefly, the possibility result (that is, that there is such a CNF-based HSS) follows from an example of such a scheme.
The impossibility result (that is, that there is no such Shamir-based HSS) follows from a computer search.  Naively such a search (for example, over all sets of $k=5$ linear functions from $\FF_8^4 \to \FF_2$) is not computationally tractable.  Instead we first analytically reduce the problem to one that is tractable, and then run the search.   


Notice that the share size per server is $16$ bits with CNF, and $12$ bits with Shamir (when the shares are in $\FF_8$).  In particular, if $\EE$ were any larger extension field of $\FF_2$ than $\FF_8$, the Shamir scheme an upload cost at least as large as the CNF scheme.  Thus, the example when $\EE = \FF_8$ is the most generous for Shamir-based HSS if we demand that the Shamir-based HSS still has an upload cost advantage.

\subsubsection{Application to Private Information Retrieval} \label{sec:linearPIR}

As a simple application of HSS for low-degree polynomials, we extend the information-theoretic PIR protocol from \cite{Chor:1998:PIR:293347.293350}
to allow better download rate by employing more servers, while maintaining the same (sublinear) upload cost.


\begin{theorem}[PIR with sublinear upload cost and high download rate]\label{thm:PIR}
For all integers $d,t,k,w>0$, such that $dt+1\leq k$, there is a $t$-private $k$-server PIR protocol for $(w\cdot(k-dt)\cdot\lceil \log_2 k\rceil)$-bit record databases of size $N$ such that:
\begin{itemize}
    \item The upload cost is $O(k^3\log k\cdot N^{1/d})$ bits;
    \item The download cost is $wk\lceil \log_2 k\rceil$ bits. Consequently, the rate of the PIR is $1-dt/k$.
\end{itemize}
\end{theorem}
\begin{proof}
Without loss of generality, assume $w=1$, as the scheme can be repeated to apply for arbitrary $w$. To this end, we view the database of size $N$ as $k-dt$ vectors $(D_{i,1},\ldots, D_{i,N})\in(\mathbb{F}_{2^{\lceil \log_2 k\rceil}})^N$, $i=1,\ldots,k-dt$, such that each record with index $j\in [N]$ in the database consists of $(D_{1,j},D_{2,j},\ldots,D_{{k-dt,j}})$. Next, fix an injective mapping $\eta:[N]\rightarrow W(d,m)$, where $W(d,m)$ is the set of all vectors from $\{0,1\}^m$ with Hamming weight $d$, $m>0$ is the smallest integer such that $N\leq \binom{m}{d}$. In addition, for $j\in[N]$ and $\ell\in [d]$ let $\eta(j)_\ell$ denote the $\ell$'th nonzero coordinate of $\eta(j)$. Furthermore, for $i=1,\ldots,k-dt$ define the polynomial over $\mathbb{F}_{2^{\lceil \log_2 k\rceil}}[z_1,\ldots,z_m]$:
\[
p_i(z_1,\ldots,z_m)=\sum_{j=1}^N D_{i,j}\cdot z_{\eta(j)_1}\cdots z_{\eta(j)_d}
\]
Now, notice that the polynomials $p_1,\ldots, p_{k-dt}$ jointly encode the database, in the sense that for every $j\in [N]$, $i=1,\ldots,k-dt$ we have $p_i(\eta(j))=D_{i,j}$. 


Let $\Pi$ be the $t$-private $k$-server HSS from Theorem \ref{thm:shamir-hss} for $\POLY_{d,m}(\mathbb{F}_{2^{\lceil \log_2 k\rceil}})^{(k-dt)}$ (with the choice of $b=1$).
Suppose that the client wishes to recover the $j$'th record  $(D_{1,j},D_{2,j},\ldots,D_{{k-dt,j}})$.  By construction, the client needs to receive the evaluation of all $k-dt$ polynomials $p_i$ on point $\eta(j)$. 
Thus, the client uses $\Pi$ to share the $\eta(j)$ among the $k$ servers; $t$-privacy of the PIR protocol follows from $t$-privacy of $\Pi$.  Next, the servers each compute their output shares according to $\Pi$ and return them to the client.  The client runs $\Pi$'s reconstruction algorithm to recover $p_i(\eta(j))$ for each $i$.  This determines the $j$'th record as noted above.


By the guarantees of Theorem \ref{thm:shamir-hss}, the upload cost is $k\cdot m\cdot(k-dt)^2\cdot\lceil \log_2 k\rceil=O(k^3\log k\cdot N^{1/d})$, as $m=O(N^{1/d})$. In addition, the download cost is $wk\lceil \log_2 k\rceil$ bits.
\end{proof}

\subsection{Negative Results for Linear HSS}\label{sec:linearNegative}


In Section~\ref{sec:HSS-cnf}, we obtained linear HSS schemes for $\POLY_{d,m}(\FF)^\ell$ with rate $1 - dt/k$.  Our main result in this section, Theorem~\ref{thm:barrier}, shows that this bound is tight for {\em linear} HSS.  In particular, this implies that for linear HSS, there is an inherent overhead for polynomials of larger degree in terms of download cost or rate. 

Moreover, we also provide some simple negative results that (a) strengthen the bound of Theorem~\ref{thm:barrier} for linear HSS-for-concatenation for specific parameters (Proposition~\ref{cor:negDist}), and (b) provide a negative result for general (not necessarily linear) HSS, which is tight for $d=1$ (Theorem~\ref{lem:negNotLinear}).

\begin{theorem}\label{thm:barrier}
Let $t,k,d,m, \ell$ be positive integers so that $m \geq d$.
Let $\FF$ be any finite field.
Let $\Pi = (\Share, \Eval, \Rec)$ be a $t$-private $k$-server linear HSS for $\POLY_{d,m}(\FF)^\ell$.  
Then $dt < k$, and the download cost of $\Pi$ is at least $k\ell\log_2|\mathbb{F}|/(k-dt)$. Consequently, the download rate of $\Pi$ is at most $1-dt/k$. 
\end{theorem}

\begin{remark}
An inspection of the proof shows that Theorem~\ref{thm:barrier} still holds if $\Rec$ is allowed to be linear over some subfield $\B$ of $\FF$,  rather than over $\FF$ itself.
\end{remark}

\begin{proof}[Proof of Theorem~\ref{thm:barrier}]
Suppose that $\Pi = (\Share, \Eval, \Rec)$ is a $t$-private $k$-server linear HSS scheme for $\POLY_d^\ell$. By the definition of linearity (Definition~\ref{def:linear}) and of $\POLY^\ell_d$ (Definition~\ref{def:POLY}), we have $\mathcal{X} = \mathcal{Y} = \mathbb{F}$ for some finite field $\FF$; $\Share(\cdot)$ is linear over $\FF$; each output share $\by^{(j)}$ (output of $\Eval$) is an element of $\F^{b_j}$ for some $b_j \geq 0$; and $\Rec:\F^{\sum_j b_j} \to \F^{\ell}$ is linear over $\F$.
Suppose without loss of generality that $b_1 \leq b_2 \leq \cdots \leq b_k.$

Without loss of generality, we may assume that $\Share$ is $t$-CNF sharing; indeed, by Theorem~\ref{cdi}, $t$-CNF shares can be locally converted to shares of any $t$-private secret-sharing scheme with a linear $\Share$ function.

Also without loss of generality, it suffices to show an impossibility result for the special case where there are $d\ell$ inputs, and the HSS is for the class $\mathcal{F}^\ell$, where $\mathcal{F}$ consists only of the monomial $f(x_1, \ldots, x_d) = \prod_{i=1}^d x_i$.  Indeed, this is because $\mathcal{F} \subset \POLY_d$.

Suppose that the secrets are $x_{i,r}$ for $i \in [d]$ and for $r \in [\ell]$, so that we want to compute $\prod_{i=1}^d x_{i,r}$ for each $r \in [\ell]$.  In order to share them according to $t$-CNF (Definition~\ref{cnf}), we choose random $X_{i,r}^S \in \F$ for all $i \in [d], r \in [\ell],$ and all $S\subset[k]$ of size $t$, so that for all $i,r$,
\begin{equation}\label{eq:sharing}
 x_{i,r} = \sum_S X_{i,r}^S, 
\end{equation}
where the sum is over all sets $S \subset [k]$ of size $t$.
Then server $j$ receives the input share $\bx^{(j)} = (X_{i,r}^S\,:\, j \not\in S)$ and generates an output share $\by^{(j)} \in \FF^{b_j}$ via the $\Eval$ map.

We may treat each output share $\by^{(j)}$ as a vector of polynomials $\vec{p}^{(j)}(\bX) = (p^{(j)}_{1}(\bX), \ldots, p^{(j)}_{b_j}(\bX))$ in the variables $\bX = (X_{i,r}^S)_{i \in [d],r\in[\ell],S\subset [k] \text{ with } |S| =t }$.
Since $\Rec$ is $\FF$-linear, there are some vectors $\vec{v}^{(j)}_{r} \in \FF^{b_j}$ for each $j \in [k], r \in [\ell]$ so that the $r$'th output of $\Rec(\by^{(1)}, \ldots, \by^{(k)})$ is given by
\begin{equation}\label{eq:prdef}
 p_r(\bX) := \sum_{j=1}^k \langle \vec{v}^{(j)}_{r}, \vec{p}^{(j)}(\bX) \rangle.
\end{equation}
By the correctness of $\Pi$, this output must be equal to $\prod_{i=1}^d x_{i,r}$.  Plugging in \eqref{eq:sharing}, this reads
\begin{equation}\label{eq:pr}
 p_r(\bX) = \sum_{S_1, \ldots, S_d} \left( \prod_{i=1}^d X_{i,r}^{S_i} \right), 
\end{equation}
where the sum is over all choices of $S_1, \ldots, S_d \subset [k]$ of size $t$.
Define $\ell$ polynomials $w_r(\bX)$ by
\[ w_r(\bX) = \sum_{j=1}^{k-dt} \langle \vec{v}_r^{(j)}, \vec{p}^{(j)}(\bX) \rangle \]
for $r \in [\ell]$.  We will show below in Claim~\ref{cl:linind} that, on the one hand, the $w_r$ are linearly independent over $\FF$; but on the other hand, they are all clearly contained in a $\left(\sum_{j=1}^{k-dt} b_j\right)$-dimensional subspace, spanned by the polynomials $p^{(j)}_{h}(\bX)$ for $j \in [k-dt]$ and $h \in [b_j]$ (recalling that $\vec{p}^{(j)} = (p^{(j)}_{h})_{h \in [b_j]}$).  This implies that
\[ \ell \leq \sum_{j=1}^{k-dt} b_j \leq \max\left\{ 0, \frac{k-dt}{k} \sum_{j=1}^k b_j\right\}, \]
where we have used the assumption that $b_1 \leq b_2 \leq \cdots \leq b_k$.  Rearranging (and using the fact that $\ell > 0$), this implies that $dt < k$ and that the scheme downloads at least
\[ \sum_{j=1}^k b_j \geq \ell \left( \frac{k}{k-dt} \right) \]
symbols of $\FF$, or at least $k \ell \log_2|\FF| / (k-dt)$ bits, which will prove the theorem.

Thus, it suffices to prove the following claim.
\begin{claim}\label{cl:linind}
The polynomials $w_r(\bX)$ are linearly independent over $\FF$.
\end{claim}
\begin{proof}
Choose sets $S_1^*, S_2^*, \ldots, S_d^* \subseteq [k]$ with $|S_i^*|=t$ so that 
$$\bigcup_{i=1}^{d} S_i^* = \{k - dt + 1, k - dt + 2, \ldots, k\}.$$
If $dt < k$, then the $S_i^*$ can be any sets that cover the last $dt$ elements of $[k]$. Then consider the monomial
\[ M_r(\bX) := \prod_{i=1}^d X^{S_i^*}_{i,r} \]
for $r \in [\ell]$.  From \eqref{eq:pr}, we see that $M_r(\bX)$ appears in $p_r(\bX)$.  However, since for all $j > k-dt$ there exists $i \in [d]$ such that  $j \in S_i^*$,
$M_r(\bX)$ cannot appear in $\vec{p}^{(j)}(\bX)$ for any $j > k - dt$.  Then \eqref{eq:prdef} implies that $M_r(\bX)$ must appear in $w_r(\bX)$.  

Further, for any $r' \neq r$, $M_{r}(\bX)$ cannot appear in $w_{r'}(\bX)$.  Indeed, since $M_{r}(\bX)$ does not appear in $p_{r'}(\bX)$ (from \eqref{eq:pr}), if $M_{r}(\bX)$ appeared in $w_{r'}(\bX)$, then \eqref{eq:prdef} implies that it must be canceled by a contribution by some $\vec{p}^{(j)}(\bX)$ for $j > k - dt$. But, as above, $M_r(\bX)$ cannot appear in $\vec{p}^{(j)}(\bX)$.

Thus, each $M_r(\bX)$ appears in $w_{r'}(\bX)$ if and only if $r' = r$.  Since the monomials $M_{r}(\bX)$ are linearly independent over $\FF$, this implies that the polynomials $w_{r}(\bX)$ are also linearly independent over $\FF$.  This proves the claim.
\end{proof}
With the claim proved, the theorem follows.
\end{proof}


Theorem~\ref{thm:barrier} implies that any linear HSS scheme for concatenation must have download rate at most $1 - t/k$.  As we have seen, this is achievable in some parameter regimes, but it turns out that it is not achievable for all parameter regimes.
In particular, in the corollary below, we observe that Theorem~\ref{thm:cnf-hss-d1} immediately rules this out, based on known bounds on the distance of linear codes.  

\begin{corollary}\label{cor:negDist}
Let $k,b,t, \ell$ be integers so that $\ell < kb-1$.
Suppose that $\Pi$ is a linear $t$-private $k$-server HSS for $\CONCAT_\ell(\FF)$ with output shares in $\FF^b$, and suppose that $|\FF|^b < k/2$.  Then the download rate of $\Pi$ is strictly less than $1 - t/k$.
\end{corollary}

\begin{proof}
It is known (see, e.g., \cite{Ball12}) that any linear code $C$ of length $k$ over an alphabet of size $q < k/2$ with $\dim(C) < k-1$ \em cannot \em achieve the Singleton bound; that is, we must have $\mathrm{dim}(C) < k - \mathrm{Dist}(C) + 1$.  Thus, Theorem~\ref{thm:cnf-hss-d1} implies the corollary.
\end{proof}

Finally, we note that for general (not necessarily linear) information-theoretic HSS, one can also show a negative result that implies that some degredation in rate is inevitable for larger security thresholds $t$.  However, this result does not depend on the complexity of $\mathcal{F}$, and in fact applies to families $\cF$ with only one function in them.  The proof below is similar to one that appears in \cite{HLKB16} for low-bandwidth secret sharing.

\begin{lemma}[Similar to Prop. 1 in
\cite{HLKB16}]\label{lem:negNotLinear}
Let $\Pi$ be any $t$-private $k$-server HSS for a function class $\mathcal{F}$ with output alphabet $\mathcal{Y}$, so that there is some $f \in \mathcal{F}$ with $\mathrm{Im}(f) = \mathcal{Y}$.
Then 
$\Pi$ must have download cost at least $\frac{k\log_2|\mathcal{Y}|}{k-t}$. Consequently, its download rate is at most $1-t/k$.
\end{lemma}

\begin{proof}
Fix a function $f \in \mathcal{F}$ so that $\mathrm{Im}(f) = \mathcal{Y}$.
Choose a secret $\bx \in \cX^m$ randomly, according to some distribution so that $f(\bx)$ is uniformly distributed in $\mathcal{Y}$.  (Such a distribution exists since the image of $f$ is $\cY$).
Suppose that $y^{(1)}, \ldots, y^{(k)}$ are the output shares (outputs of $\Eval$), so that $y^{(j)} \in \{0,1\}^{b_j}$ is $b_j$ bits long.  Assume without loss of generality that $b_1 \geq b_2 \geq \cdots \geq b_k$.  Then we have
\begin{align}
    H(\bx) &= H(\bx | y^{(1)}, \ldots, y^{(t)}) \label{eq:sec}\\
    &\leq H(\bx, y^{({t+1})}, \ldots, y^{(k)} | y^{(1)}, \ldots, y^{(t)} ) \notag\\
    &= H(\bx | y^{(1)}, \ldots, y^{(k)} ) + H(y^{({t+1})}, \ldots, y^{(k)} | y^{(1)}, \ldots, y^{(t)} ) \label{eq:chrl} \\
    & \leq H(\bx | f(\bx) ) + H(y^{({t+1})}, \ldots, y^{(k)} ), \label{eq:last}
\end{align}
where \eqref{eq:sec} follows from $t$-privacy, \eqref{eq:chrl} follows from the chain rule for conditional entropy, and \eqref{eq:last} follows from the facts that $y^{(1)}, \ldots, y^{(k)}$ determine $f(\bx)$ and that $H(A|B) \leq H(A|g(B))$ for any function $g$ and any random variables $A,B$.
Rearranging and using the fact that $H(y^{({t+1})}, \ldots, y^{(k)}) \leq b^{({t+1})} + \ldots + b^{(k)}$, we conclude that
\[H(\bx) - H(\bx | f(\bx)) \leq H(y^{({t+1})}, \ldots, y^{(k)}) \leq b_{t+1} + \cdots + b_k \leq \frac{k-t}{k}(b_1 + \cdots + b_k),\]
using the assumption that $b_1 \geq b_2 \geq \cdots \geq b_k$.
Finally, we have
\[ H(\bx) - H(\bx | f(\bx) ) = I(\bx; f(\bx)) = H(f(\bx)) - H(f(\bx) | \bx) = H(f(\bx)), \]
because since $f$ is fixed, $H(f(\bx) | \bx) = 0$.
We conclude that
\[ H(f(\bx)) \leq \frac{k-t}{k}(b_1 + \cdots + b_k). \]
Since $f(\bx)$ is uniform in $\mathcal{Y}$, we have $H(f(\bx)) = \log_2|\cY|$, and re-arranging we see that
\[ b_1 + \cdots + b_k \geq \frac{k \log_2|\cY|}{k-t}, \]
as desired.
\end{proof}

\section{Black-Box Rate Amplification}\label{sec:blackbox}

In this section we present a {\em generic} rate amplification technique for HSS.  More concretely, we show how to make a black-box use of any $t_0$-private $k_0$-server HSS  scheme $\Pi_0 = (\Share_0, \Eval_0, \Rec_0)$ for $\mathcal{F}$ with {\em additive reconstruction} (see Definition~\ref{def:linear}) 
to obtain a $t$-private $k$-server HSS scheme $\Pi=(\Share,\Eval,\Rec)$ (where $t/k<t_0/k_0$) for $\mathcal{F}^\ell$ with better download rate and similar upload cost. This construction can be applied, with useful corollaries, to both information-theoretic and computational HSS. 

After laying out the general approach (Lemma \ref{black-box}), we present the three useful instances of this technique discussed in the introduction. First, in Section \ref{sec:bblarge} we obtain a general-purpose (computationally) $t$-private $k$-server HSS $\Pi$ with rate $1-t/k$.  Second, in Section~\ref{sec:bb2}, we show how to transform any $1$-private $2$-server HSS scheme $\Pi_0$ into a $1$-private $k$-server HSS scheme $\Pi$ with rate $1-1/k$, and obtain high-rate computationally secure PIR and HSS.  Finally, in Section~\ref{sec:mvpir}, we obtain (information theoretic) PIR with sub-polynomial ($N^{o(1)}$) upload cost and download rate approaching $1$.

\paragraph{The general approach.} Let $\Pi_0 = (\Share_0, \Eval_0, \Rec_0)$ be a $t_0$-private, $k_0$-server HSS for $\mathcal{F}$ \emph{with additive reconstruction over $\FF$}.  We will construct a $t$-private, $k$-server HSS $\Pi$ for $\mathcal{F}^\ell$, using $\Pi_0$ as a black box.  Suppose that we have inputs $\bx_1, \ldots, \bx_\ell$ and we wish to compute $(f_1(\bx_1), \ldots, f_\ell(\bx_\ell))$, for $f_i \in \mathcal{F}$.  The general paradigm proceeds as follows:
\begin{enumerate}
    \item Using $\Share_0$, we share the inputs as input shares $\bx_i^{(v)}$, for $i \in [\ell], v \in [k_0]$.  In order to define $\Share$ in $\Pi$, these input shares $\bx_i^{(v)}$ are replicated, \emph{somehow}, between the $k$ servers.
    \item Each server $j \in [k]$ applies $\Eval_0$ to the input shares $\bx_i^{(v)}$ that they hold, to obtain output shares $y_i^{(v)}$.  Thus, $(y_i^{(v)})_{v \in [k_0]}$ are additive shares of $f_i(\bx_i)$: $\sum_{v \in [k_0]} y_i^{(v)} = f_i(\bx_i)$.  
    Now each server $j$ holds some collection $\bY(j)$ of the $y_i^{(v)}$. These $\bY(j)$ form shares of $(f_1(\bx_1), \ldots, f_\ell(\bx_\ell))$ under some $\ell$-LMSSS $\mathcal{L}_0$.
    \newline\newline
    In order to define $\Eval$ for $\Pi$, we use \emph{some} share conversion from $\mathcal{L}_0$ to a new $\ell$-LMSSS $\mathcal{L}$.  The download rate of $\Pi$ is thus given by the information rate of $\mathcal{L}$.
    \item The output client uses the reconstruction algorithm for $\mathcal{L}$ in order to define $\Rec$ for $\Pi$.
\end{enumerate}

This approach allows essentially two degrees of freedom: First, in Step 1 we may choose how to replicate the shares; and second, in Step 2 we may choose the share conversion and the
new LMSSS $\mathcal{L}$.
Hence, since the shares $y_i^{(v)}$ are $k_0$-additive shares, the above can be seen as a \emph{generalized CNF share conversion question} \`a la Theorem \ref{cdi}: instead of using the usual CNF sharing from Definition~\ref{cnf}, each secret is additively shared into only $k_0$ shares, and there is the freedom to choose which server receives which share. 
In general, for all integers $t,k,t_0,k_0,\ell$, we will be interested in whether a $t$-private $k$-server HSS $\Pi$ for $\mathcal{F}^\ell$ can be constructed from a $t_0$-private $k_0$-server HSS $\Pi_0$ for $\mathcal{F}$ using the black-box approach above. As we will see later, for our purposes it will be sufficient to restrict ourselves to $t_0=(k_0-1)$-private $k_0$-server schemes, as even for $(t_0<k_0-1)$-private $\Pi_0$, we could still apply black-box transformations which preserve $k_0-1$ privacy (see Remark \ref{remark:black-box-smaller-t}). 
We can capture this approach with the following notion:

\begin{definition}[Black-box transformation]\label{def:blackbox}
Let $t,k,k_0,\ell$ be integers, and $\FF$ a finite field. We say there is a \emph{black-box transformation for $(t,k,k_0,\ell,\mathbb{F})$}, if there is $k$-server $\ell$-LMSSS over $\F$ $\mathcal{L} = (\Share_\mathcal{L}, \Rec_\mathcal{L})$ with parameters $(e,b_1,\ldots,b_k)$, \emph{replication functions} $\psi_i:[k_0]\rightarrow 2^{[k]}$, $i\in [\ell]$, and \emph{conversion functions} $\varphi_j:\FF^{c_j}\rightarrow\FF^{b_j}$, $j\in[k]$, where $c_j=\sum_{i=1}^\ell |\{v \in [k_0]\,:\, j \in \psi_i(v)\}|$, 
such that 
\begin{itemize}
    \item \textbf{Correctness:} For every $(y_i^{(j)})_{i\in [\ell],j\in[k_0]}$ it holds that
\[
\exists \br\in\FF^e: (\varphi_j(\vec{Y}(j)))_{j=1}^k=\Share_\mathcal{L}\left(\left(\sum_{v=1}^{k_0} y_1^{(v)},\ldots,\sum_{v=1}^{k_0} y_\ell^{(v)} \right),\br\right).
\]
where $\vec{Y}(j)=(y_i^{(v)})_{i\in[\ell],v\in [k_0]:j\in\psi_i(v)}$ denotes the view of party $j \in [k]$.
\item \textbf{Security:} For every $T\subseteq [k]$ such that $|T|\leq t$ and every $i\in[\ell]$ we have that $$\left|\bigcup_{j\in T}\{v\in [k_0]:j\in\psi_i(v)\}\right|\leq k_0-1.$$
\end{itemize}
We further say that the black-box transformation $(\mathcal{L},(\psi_i)_i,(\varphi_j)_j)$ has rate $R$ if $\mathcal{L}$ has information rate $R$.
\end{definition}

Given a black-box transformation and a compatible HSS, we may readily deduce an HSS with potentially better rate. A formal description of our 3-step black-box approach is given in the following lemma:

\begin{lemma}[Black-box transformations applied to HSS]\label{black-box}
Suppose there is a (\emph{potentially computationally secure}) $(k_0-1)$-private $k_0$-server HSS $\Pi_0$ for $\mathcal{F}=\{\mathcal{X}^m\rightarrow\FF\}$ with additive reconstruction over $\FF$. In addition, let $(\mathcal{L}=(\Share_\mathcal{L},\Rec_\mathcal{L}),(\psi_i)_{i\in[\ell]},(\varphi_j)_{j\in[k]})$ be a black-box transformation with parameters $(t,k,k_0,\ell,\mathbb{F})$ and rate $R$. Then, there is a (computationally secure if $\Pi$ is computationally secure) $t$-private $k$-server HSS $\Pi$ for $\mathcal{F}^\ell$ with rate $R$. Furthermore, in $\Pi$ server $j\in[k]$ has individual upload cost\footnote{Here, individual upload cost stands for the total number of bits an individual server $j\in[k]$ receives. That is, if $x_i^{(j)}$ are the shares of all $\ell m$ inputs, then the individual upload cost is $\sum_{i=1}^{\ell m}|x_i^{(j)}|$.} (for all $\ell m$ inputs) of $\sum_{i\in[\ell],v\in[k_0]:j\in\psi_i(v)}L_v$, where $L_v$ is the individual upload cost of the $v$'th server in $\Pi_0$.
\end{lemma}

\begin{proof}
As discussed above when we went over the general approach, the HSS $\Pi$ is defined as follows:
\begin{enumerate}
    \item \textbf{Sharing:} Every input $\vec{x}_i\in\mathcal{X}^m$, $i=1,\ldots,\ell$, is secret shared according to the source $(k_0-1)$-private $k_0$-server HSS $\Pi_0$ to obtain $(x_i^{({v})})_{v=1}^{k_0}$. For $i=1,\ldots,\ell$, the shares $(x_i^{({v})})_{v=1}^{k_0}$ are replicated according to $\psi_i$: $x_i^{(v)}$ is sent to all servers in the set $\psi_i(v)\subseteq[k]$.
    
    \item \textbf{Evaluation:} Each server $j\in[k]$ holding share $x_i^{({v})}$ computes $y_i^{({v})}:=\textsf{Eval}({f_i},v,x_i^{({v})})$, where $f_i\in\mathcal{F}$ is the function to be evaluated at $\vec{x}_i$. Then, each server $j\in [k]$ locally computes $z_j=\varphi_j(\vec{Y}(j))$ as its output share.
    \item \textbf{Reconstruction:} The client computes $\Rec_\mathcal{L}(z_1,\ldots,z_k)$.
\end{enumerate}
We first argue that $\Rec_\mathcal{L}(z_1,\ldots,z_k)=(f_1(\vec{x}_1),\ldots,f_\ell(\vec{x}_\ell))$. Indeed, by the correctness of $\Pi_0$, we have that $y_i^{(1)} +\cdots+y_i^{({k_0})}=f_i(\vec{x}_i)$, and thus, by the correctness of the black-box transformation there is a vector $\br$ such that 
\[
(z_j)_{j=1}^k=(\varphi_j(\vec{Y}(j)))_{j=1}^k=\Share_\mathcal{L}((f_1(\vec{x}_1),\ldots,f_\ell(\vec{x}_\ell)),\br).
\]
Hence, applying $\Rec_\mathcal{L}$ to $(z_j)_{j=1}^k$ should recover $(f_1(\vec{x}_1),\ldots,f_\ell(\vec{x}_\ell))$ by definition of the LMSSS. Furthermore, because the rate of the black-box transformation is the same as the information rate of $\mathcal{L}$, and $(z_j)_{j=1}^k$ are shares of $\mathcal{L}$, we conclude the rate of this HSS is also $R$. In addition, we want to argue that the HSS satisfies (computational or information theoretic) $t$-privacy. By the $(k_0-1)$ privacy assumption of $\Pi_0$ we only need to show that every set of $t$ servers receives at most $k_0-1$ shares of each $\bx_i$. The fact each server $j\in [k]$ receives a share $x_i^{(v)}$ if $j\in\psi_i(v)$, together with the security of the black-box transformation, guarantee this. Finally, the claim about individual upload cost follows because server $j$ receives shares $(x_i^{(v)})_{i\in[\ell],v\in[k_0]:j\in\psi_i(v)}$ and $|x_i^{(v)}|=L_v$.
\end{proof}

\begin{remark}[Using Lemma \ref{black-box} with $t_0<k_0-1$]\label{remark:black-box-smaller-t}
If $t=k_0-1$ then Lemma \ref{black-box} can also allow to combine a $t_0$-private (possibly $t_0<k_0-1$) $k_0$-server HSS $\Pi_0$ for $\mathcal{F}$ with a black-box transformation to obtain a new $t_0$-private $k$-server HSS $\Pi$ for $\mathcal{F}^\ell$. This follows because a black-box transformation with $t=k_0-1$ \emph{necessarily} replicates at most a single share $y_i^{(v)}$ of $y_i$ to every party.
\end{remark}

\subsection{Black-Box Transformations with Large \texorpdfstring{$k_0$}{k0}}\label{sec:bblarge}

A simple case to consider is one where $k_0$ can be large compared to $t,k$. More concretely, fixing $t,k$, we would like to find $k_0=k_0(k,t)$ large enough that there is a black-box transformation with the best possible rate $R=R(k,t)$. Note that if such a transformation exists for some $k_0(k,t)$, then it also exists for all $k_0'>k_0(k,t)$, because we may simply ``bunch servers together'' in $\Pi_0$. As it turns out, picking $k_0=\binom{k}{t}$ is sufficient, since we can replicate the shares in such a manner as to obtain the usual $t$-private $k$-server CNF sharing of $y_1,\ldots,y_\ell$. Then, to characterize the best rate achievable, it only remains to apply Corollary \ref{cdiplusplus} and Lemma \ref{folded-codes}. Formally, we have the following theorem.

\begin{theorem}
\label{th:bb1}
Let $t,k,\ell$ be integers, and let $\FF$ be a finite field.  Suppose that $b > \ell/k$.

There is a black-box transformation with parameters $(t,k,k_0=\binom{k}{t},\ell,\FF)$, such that the associated LMSSS $\mathcal{L}$ has output shares in $\FF^b$ (and hence information rate $\ell/(kb)$) if and only if
there is an $\FF$-linear code $C \in (\FF^b)^k$  with rate at least $\ell/(kb)$ and distance at least $t+1$.

Further, for the ``if'' direction, for every $j\in[k]$, $c_j = \ell {k-1 \choose t}$, where $c_j=\sum_{i=1}^\ell |\{v \in [k_0]\,:\, j \in \psi_i(v)\}|$ is as in Definition~\ref{def:blackbox}.

\end{theorem}

\begin{proof}
Since the black-box transformation implies a $t$-private $k$-server $\ell$-LMSSS $\mathcal{L}$, the ``only if'' direction follows by Lemma \ref{folded-codes}.

For the ``if'' direction, we need to construct a black-box transformation. To this end, let $\eta:[k_0]\rightarrow\binom{[k]}{t}$ be a bijection which assigns a $t$-sized subset of $[k]$ to each index of $k_0$. For every $i\in [\ell]$ we let $\psi_i(v)=[k]\setminus \eta(v)$ (this already implies $c_j=\ell\binom{k-1}{t}$). Next, let $\mathcal{L}$ be the $t$-private $k$-server $\ell$-LMSSS $\mathcal{L}$ constructed from $C$ according to Lemma \ref{folded-codes} (which has output shares in $\FF^b$). For every $i\in[\ell]$, by definition of $\psi_i$, each server $j$ holds $\vec{Y}(j)=(y_i^{(v)})_{i\in[\ell],v\in [K]:j\notin\eta(v)}$. Hence, the servers jointly hold a $t$-CNF sharing of $y_i^{(1)} +\ldots+y_i^{({k_0})}$. By Corollary \ref{cdiplusplus}, this sharing is convertible to $\mathcal{L}$, and we choose $\varphi_j$, $j\in[k]$ to be the functions of this conversion.

Correctness follows by the aforementioned observation that $\vec{Y}(j)$ are $t$-CNF shares of $y_i^{(1)} +\ldots+y_i^{({k_0})}$, combined with the fact that the functions $\varphi_j$ are taken from Corollary \ref{cdiplusplus} to yield a joint sharing of $(y_i^{(1)} +\ldots+y_i^{({k_0})})_{i\in[\ell]}$ according to $\mathcal{L}$.

For security, note that indeed, by definition of $\eta$, for every $T\subseteq [k]$ such that $|T|\leq t$ we have that $|\bigcup_{j\in T}\{v\in [K]:j\notin\eta(v)\}|=|\{A\in \binom{[k]}{t}:A\cap T\neq\emptyset\}|=|\{B\in \binom{[k]}{k-t}:B\cup T\neq[k]\}|\leq \binom{k}{t}-1$.

\end{proof}

As a corollary of Theorem \ref{th:bb1}, we can amplify the rate of a $t$-private computationally secure HSS scheme for circuits that relies on a circular-secure\footnote{Alternatively, circular security can be eliminated if the share size is allowed to grow with the size of the circuit given as input to $\Eval$.
} variant of the Learning With Errors assumption~\cite{DHRW16,BGILT18}. %
See Appendix~\ref{app:computational} for the relevant definitions.

\begin{corollary}[High-rate $t$-private HSS for circuits]\label{lwe-hss} 
Let $t,k$ be integers.
Suppose there exists a computationally $t_0$-private $k_0$-server HSS for circuits, for $k_0=\binom{k}{t}$ and $t_0=k_0-1$, with additive reconstruction over $\FF_2$ and individual upload cost $L$. Then, there exists a computationally $t$-private $k$-server HSS for circuits with $\ell$-bit outputs, $\ell=(k-t)\lceil \log_2 k\rceil$, with download rate $1-t/k$ and individual upload cost $\ell\binom{k-1}{t}L$. 
\end{corollary}

\begin{remark}\label{remark-break-barrier} Note that the conclusion of Corollary \ref{lwe-hss} circumvents the $1-dt/k$ barrier in Theorem~\ref{thm:barrier} by allowing nonlinear (and computational) HSS.
\end{remark}


\subsection{Black-Box Transformations with \texorpdfstring{$k_0=2$}{k0=2}}\label{sec:bb2}

In this section, we consider black-box transformations with $k_0=2$. 
We show that, for the special case of $t=1$, we may take $k_0=2$ and obtain a black-box transformation with \emph{optimal} rate $1 -1/k$ (for any $k$).  Note that this is better than the guarantee of Theorem~\ref{th:bb1}, which says that if we choose $k_0 = {k \choose 1} = k$ then we can achieve optimal rate.  The fact that we may obtain optimal rate already with $k_0 = 2$ may be surprising since when $k_0=2$, essentially only a $1$-private $2$-server HSS scheme $\Pi_0$ can be combined with the black-box transformation.

\begin{theorem}[Optimal black-box transformations with $k_0=2$]\label{hss2}
Let $k\geq 2$ be an integer and $\FF$ be a finite field. There is a black-box transformation with parameters $(1,k,k_0=2,\ell=k-1,\FF)$ and rate $1-1/k$. Furthermore, for every $j\in[k]$, 
$c_j=k-1$, where
$c_j=\sum_{i=1}^\ell |\{v \in [k_0]\,:\, j \in \psi_i(v)\}|$ is as in Definition~\ref{def:blackbox}.
\end{theorem}

\begin{proof}
For every $i=1,\ldots,k-1$ we define
\[
\psi_i(v)=\begin{cases}
[k]\setminus\{i\},&v=1\\
\{i\},&v=2.
\end{cases}
\]
In addition, note that for every server $j=1,\ldots,k-1$ we have
\[
\vec{Y}(j)=((y_i^{(1)})_{i\neq j},(y_j^{(2)}))
\]
and for server $j=k$ we have
\[
\vec{Y}(k)=(y_i^{(1)})_{i\in[k-1]}
\]
This implies that $c_j=k-1$ for every $j\in[k]$. In addition, for every server $j=1,\ldots,k-1$ we choose
\[
\varphi_j(\vec{Y}(j))=y_j^{(2)} -\sum_{i\neq j}y_i^{(1)}
\]
and for server $j=k$ we choose
\[
\varphi_k(\vec{Y}(k))=\sum_{i=1}^{k-1}y_i^{(1)}.
\]
To prove the correctness of the black-box transformation, consider the LMSSS $\mathcal{L}=(\Share_\mathcal{L},\Rec_\mathcal{L})$ over $\mathbb{F}$ defined by $\Share_\mathcal{L}((x_1,\ldots,x_{k-1}),r)=(x_1+r,\ldots,x_{k-1}+r,-r)$ and $\Rec_\mathcal{L}(z_1,\ldots,z_k)=(z_1-z_k,z_2-z_k,\ldots,z_{k-1}-z_k)$. It is not difficult to see that this is a $k$-server $(k-1)$-LMSSS. To see that the black-box transformation is correct with respect to $\mathcal{L}$, note that the output shares take the form 
\[
\left(y_1^{(2)} -\sum_{i\neq 1}y_i^{(1)},y_2^{(2)} -\sum_{i\neq 2}y_i^{(1)},\ldots,y_{k-1}^{(2)} -\sum_{i\neq k-1}y_i^{(1)},\sum_{i=1}^{k-1}y_i^{(1)} \right).
\]
Letting $r=-\sum_{i=1}^{k-1}y_i^{(1)}$, we see that for any $j \in [k-1]$, the $j$'th output share is $$y_j^{(2)} - \sum_{i \neq j} y_i^{(2)} = y_j^{(2)} + y_j^{(1)} + r = x_j + r,$$
as desired.
By construction, the information rate of the black box scheme is $1-1/k$, as we transmit $k$ symbols to reconstruct $k-1$.  Security follows because every server receives exactly one share of each secret.
\end{proof}

Before proceeding with implications of Theorem~\ref{hss2}, we observe that any $1$-private $k$-server LMSSS has rate at most $1 - 1/k$ 
(see Theorem \ref{thm:barrier}). In this sense, Theorem~\ref{hss2} obtains the optimal rate.



Next, we discuss the implications of Theorem~\ref{hss2}.
Combining Theorem \ref{hss2} with Proposition \ref{black-box}, we obtain several corollaries. In particular, this can be immediately applied to the constructions in \cite{BoyleGI15,BGI16}, to obtain  computationally secure PIR with rate approaching 1 and upload cost $O(\lambda\log N)$, where $\lambda$ is a security parameter (concretely, the seed length of a length-doubling pseudorandom generator), assuming one way functions exist.

\begin{corollary}[High-rate CPIR with logarithmic upload cost]
\label{cor-cpir}
Suppose one-way functions exist. Then, for any $k\ge 2$, $w\geq 1$, and $\ell=w(k-1)$, there is a computationally $1$-private $k$-server PIR protocol for databases with $N$ records of length $\ell$, with upload cost $O(k\lambda\log N)$ and download rate $1-1/k$.

\end{corollary}

The above can be extended to high-rate HSS schemes for other simple classes of functions, considered in~\cite{BoyleGI15,BGI16}, based on one-way functions.
Similarly, we can get high-rate variants of 2-server HSS schemes for {\em branching programs} based on DDH~\cite{BGI16a}\footnote{Known DDH-based constructions of additive HSS for branching programs have an inverse-polynomial failure probability. When amortizing over $\ell$ instances, one can use an erasure code to make the failure probability negligible in $\ell$ while maintaining the same asymptotic download rate.
} or DCR~\cite{OrlandiSY21,RoyS21}. 
Analogously to Corollary~\ref{cor-cpir}, we can get similar HSS schemes for branching programs with rate $1-1/k$.

\subsection{High-Rate Information-Theoretic PIR with Sub-Polynomial Upload Cost}\label{sec:mvpir}

While Theorem \ref{thm:PIR} provides PIR protocols with download rate approaching $1$ as $k\rightarrow\infty$, these schemes have polynomial ($O(N^{1/d})$) upload cost in the database size $N$. In this section we will construct PIR protocols with download rate approaching $1$ as $k\rightarrow\infty$ and that have sub-polynomial ($N^{o(1)}$) upload cost, albeit at a worse trade-off between the number of servers and the download rate.

The black-box approaches of Sections~\ref{sec:bblarge} and \ref{sec:bb2} are not sufficient because they either require an HSS with a high ratio of security-to-servers ($1-1/k_0$), or a very small number of servers ($k_0=2$), while in this section we eventually apply it to a $1$-private $6$-server HSS. We don't obtain a tight result in this section, but, nevertheless, it allows us to convert a $t_0$-private $k_0$-server HSS $\Pi_0$ to a $t_0$-private $k$-server HSS $\Pi$ with higher rate, where $k>k_0$, and hence it is sufficient for our purposes of improving the rate of PIR protocols with sub-polynomial upload cost. 

Note that we preserve the $t_0$ privacy of the original HSS (in similar fashion to the $k_0=2$ case where $t_0=1$). Thus, in view of Remark \ref{remark:black-box-smaller-t}, without loss of generality we can assume that $t_0=k_0-1$. To demonstrate the idea of this black-box transformation, suppose we are given a $3$-additive secret sharing for every secret $y_i=y_i^{(1)} +y_i^{(2)} +y_i^{(3)}$, $i\in[2]$. We will show how to obtain a $2$-private $5$-server $2$-LMSSS sharing with information rate $2/5>1/3$. Let $q=3$ and suppose we have servers indexed by the elements of $[q]=\{1,2,3\}$. In addition, suppose we have two additional servers indexed by the \emph{sets} $\{1,2\}$ and $\{1,3\}$, respectively. In this construction we give the third share of the first secret $y_1^{(3)}$ to server $\{1,2\}$ and distribute the other shares $y_1^{(1)},y_1^{(2)}$ arbitrary among the servers indexed with $1,2$. Next, suppose we do the same with the servers $1,3,\{1,3\}$ and the shares $y_2^{(1)},y_2^{(2)},y_2^{(3)}$. Assuming each server simply adds all the shares it gets, we would like that $\varphi_{1}(\vec{Y}(1))+\varphi_{2}(\vec{Y}(2))+\varphi_{\{1,2\}}(\vec{Y}(\{1,2\}))=y_1$ (and the same with $y_2$ and servers $1,3,\{1,3\}$). Unfortunately, as it currently stands this does not hold because the servers $1,2,3$ hold terms related to $y_1$ and $y_2$ that do not cancel. To remedy this, suppose we also give server $\{1,2\}$ all the shares of $y_2$ servers $1,2$ received, so it may subtract them from the sum (and the same with server $\{1,3\}$  and the shares of $y_1$). Formally:
\begin{eqnarray*}
    \vec{Y}(1)=\{y_1^{(1)},y_2^{(1)}\}&\qquad&\varphi_1(\vec{Y}(1))=y_1^{(1)} +y_2^{(1)} \\
    \vec{Y}(2)=\{y_1^{(2)} \}&\qquad&\varphi_2(\vec{Y}(2))=y_1^{(2)} \\
    \vec{Y}(3)=\{y_2^{(2)} \}&\qquad&\varphi_3(\vec{Y}(3))=y_2^{(2)} \\
    \vec{Y}(\{1,2\})=\{y_1^{(3)},y_2^{(1)} \}&\qquad&\varphi_{\{1,2\}}(\vec{Y}(\{1,2\}))=y_1^{(3)} -y_2^{(1)} \\
     \vec{Y}(\{1,3\})=\{y_2^{(3)},y_1^{(1)} \}&\qquad&\varphi_{\{1,3\}}(\vec{Y}(\{1,3\}))=y_2^{(3)} -y_1^{(1)}
\end{eqnarray*}
This satisfies $\varphi_{1}(\vec{Y}(1))+\varphi_{2}(\vec{Y}(2))+\varphi_{\{1,2\}}(\vec{Y}(\{1,2\}))=y_1$ (and the same with $y_2$ and servers $1,3,\{1,3\}$). Furthermore, since $|\{1,2\}\cap\{1,3\}|\leq 1$, this scheme is $2$-private, as each server receives at most a single share from each secret. Generalizing this idea leads to the following lemma:

\begin{lemma}
\label{th:bb3}
Let $k_0<q$ be an integer and $\FF$ a finite field. Suppose we have a collection of $(k_0-1)$-sized subsets from $[q]$, $\mathcal{S}\subseteq\binom{[q]}{k_0-1}$, with the additional requirement that any two distinct sets $S_1,S_2\in\mathcal{S}$ satisfy $|S_1\cap S_2|\leq 1$. Then there is a black-box transformation with parameters $(k_0-1,q+|\mathcal{S}|,k_0,|\mathcal{S}|,\mathbb{F})$ and rate $1-q/(q+|\mathcal{S}|)$. 

Furthermore, for every $j\in[q+|\mathcal{S}|]$, $c_j \leq |\mathcal{S}|$, where e $c_j=\sum_{i=1}^\ell |\{v \in [k_0]\,:\, j \in \psi_i(v)\}|$ is as in Definition~\ref{def:blackbox}.
\end{lemma}
\begin{proof}
Associate indices $i=1,\ldots,\ell$, $\ell=|\mathcal{S}|$, with the sets $S_i\in\mathcal{S}$, and the indices $\ell+1,\ldots,\ell+q$ with the elements of $[q]$. We will describe the construction of $(\psi_i)_{i\in[\ell]}$ in algorithmically as follows.
\begin{enumerate}
    \item For every $i\in[q+\ell]$ and $v\in[k_0]$ initialize $\psi_i(v)=\emptyset$.
    \item For every $i\in[\ell]$, update $\psi_i(k_0)=\psi_i(k_0)\cup \{i\}$.
    \item For every $i\in [\ell]$, let $S_i=\{j_1,\ldots,j_{k_0-1}\}\in\mathcal{S}$ be the associated set. Update $\psi_{i}(r)=\psi_i(r)\cup\{\ell+j_r\}$ for $r=1,\ldots,k_0-1$ (note that the indices in $S_i$ may be ordered arbitrarily).
    \item For every $i\in[\ell]$ let again $S_i=\{j_1,\ldots,j_{k_0-1}\}\in\mathcal{S}$ be the associated set. Now, for $i'\in[\ell]\setminus\{i\}$ and $r=1,\ldots,k_0-1$ update $\psi_{i'}(r)=\psi_{i'}(r)\cup\{i\}$.
\end{enumerate}
Next, we define $(\varphi_j)_{j\in[q+\ell]}$ as follows:
\begin{itemize}
    \item For every server $j=1,\ldots,\ell$ let $I_j\subseteq [\ell]\times[k_0]$ denote the set of indices $(i,v)\in I_j$ such that $j\in\psi_i(v)$, \emph{excluding} $(j,k_0)$. Then, we set $\varphi_j(\vec{Y}(j))=y_j^{({k_0})}-\sum_{(i,v)\in I_j}y_i^{(v)}$.
    \item Every server $j=\ell+1,\ldots,\ell+q$ simply computes the sum of the shares it got $\varphi_j(\vec{Y}(j))=\sum_{i\in[\ell],v\in [k_0]:j\in\psi_i(v)}y_i^{(v)}$.
\end{itemize}
To show that the above describes a black-box transformation over $\FF$, we need to establish security and correctness.  We first argue security, which holds because, for every $i\in[\ell]$, every server with index $j\in[q+\ell]$ receives at most a single element $y_i^{(v)}$ (which also implies $c_j\leq \ell$). To see this, note that $y_i^{({k_0})}$ is sent exclusively to server with index $i$ (and no other share $y_i^{(v)}$ is sent to it). In addition, every server with index $j=\ell+1,\ldots,\ell+q$ receives at most a single element $y_i^{(v)}$ for every $i\in[\ell]$. It is potentially only possible for a server with index $j\in[\ell]$ to receive elements $y_i^{(v)},y_i^{({v'})}$ with $v\neq v'$. However, if $j=i$ this does not happen because we exclude it in step 4. If $j\neq i$ it can only happen if there are two elements $j_v,j_{v'}\in S_i\cap S_{i'}$ for some $i'\neq i$. However, because we require that $|S_i\cap S_{i'}|\leq 1$ this does not happen. Hence, every subset $T$ of size $|T|\leq k_0-1$ satisfies that there is no $i\in[\ell]$ for which the servers in $T$ hold all the shares $y_i^{(1)},\ldots, y_i^{({k_0})}$.

To argue correctness we first define a $(q+\ell)$-server $\ell$-LMSSS over $\mathbb{F}$, $\mathcal{L}=(\Share_\mathcal{L},\Rec_\mathcal{L})$ as follows:
\begin{itemize}
    \item $\Share_\mathcal{L}$ first additively shares $y_i=y_i^{(1)} +\ldots+y_i^{({k_0})}$ for every $i\in[\ell]$. Then, each server $j$ receives share $y_i^{(v)}$ if $j\in\psi_i(v)$.
    \item For $\Rec_\mathcal{L}$ note that $\varphi_i(Y(i))+\sum_{j\in S_i}\varphi_{\ell+j}(Y(\ell+j))=y_i^{(1)} +\ldots+y_i^{({k_0})}=y_i$ because if $y_{i'}^{({v'})}\in Y(\ell+j)$ for $i'\neq i$ then also $y_{i'}^{({v'})}\in Y(i)$ (and so these terms cancel as they appear with a minus sign in $\varphi_i(Y(i))$). Therefore, only terms from $\{y_i^{(v)} \}_{v\in[k_0]}$ don't cancel. Because we can reconstruct every $y_i$, this fully specifies $\Rec_\mathcal{L}$.
\end{itemize}
Now, since $\mathcal{L}$ is an LMSSS and the shares $(\varphi_j(Y(j))_{j\in[q+\ell]}$ are shares of $\mathcal{L}$, correctness follows. Finally, for the rate of the black-box transformation, note that the information rate of $\mathcal{L}$ is $1-q/(q+\ell)$.
\end{proof}

Next, we need the PIR protocol from \cite{BeimelIKO12} (although the PIR from \cite{STOC:Efremenko09}, which has larger upload cost, can also be used).

\begin{theorem}\cite{BeimelIKO12}\label{thm:biko}
There exists a $1$-private $3$-server HSS $\Pi=(\Share,\Eval,\Rec)$ for $\ALL_{\FF_4}$ with upload cost $O\left(2^{6\sqrt{\log N\log\log N}}\right)$. Furthermore, the output shares satisfy $z^{(1)} +z^{(2)} +z^{(3)}=\eta\cdot f(x)$, where $f\in\ALL_{\FF_4}$ is the function on which $\Eval$ is applied, $x$ is the input, and $\eta\neq 0$ depends only on $x$ and the randomness of $\Share$.
\end{theorem}

The PIR from Theorem \ref{thm:biko} does not have additive reconstruction, but its reconstruction algorithm is simple enough so that it can be modified into an additive PIR, at the cost of doubling the number of servers.

\begin{proposition}\label{prop:mvpir}
There exists a $1$-private $6$-server HSS $\Pi$ for $\ALL_{\FF_4}$ with additive reconstruction and upload cost $O\left(2^{6\sqrt{\log N\log\log N}}\right)$.
\end{proposition}
\begin{proof}
Let $\Pi_0=(\Share_0,\Eval_0,\Rec_0)$ be the PIR from Theorem \ref{thm:biko}. The HSS $\Pi$ proceeds as follows:
\begin{itemize}
    \item \textbf{Sharing:} Let $(x_0^{(1)},x_0^{(2)},x_0^{(3)})=\Share_0(x)$ and $\eta\in\FF_4$ be as in Theorem \ref{thm:biko}. Let $\eta^{-1}=\eta_1+\eta_2$ be an additive sharing. We use $6$ servers with input shares: $(x^{(1)},x^{(2)},x^{(3)},x^{(4)},x^{(5)},x^{(6)})=((x_0^{(1)},\eta_1),(x_0^{(2)},\eta_1),(x_0^{(3)},\eta_1),(x_0^{(1)},\eta_2),(x_0^{(2)},\eta_2),(x_0^{(3)},\eta_2))$.
    \item \textbf{Evaluation:} Each server $r$ holding a share $x_0^{(i)}$ and $\eta_j$, outputs $z_r=\eta_j\cdot \Eval(f,i,x_0^{(i)})$.
    \item \textbf{Reconstruction:} The reconstruction algorithm computes $z_1+\ldots+z_6$.
\end{itemize}
Privacy holds because each server receives a single share $x_0^{(i)}$ and a single share $\eta_j$. Correctness holds because $\sum_{r=1}^6z_r=(\eta_1+\eta_2)\sum_{i=1}^3\Eval(f,i,x_0^{(i)})=\eta^{-1}\eta f(x)=f(x)$.
\end{proof}

\begin{theorem}\label{th:mv-pir}
There exists a $1$-private $k$-server PIR for $2w\cdot (k-\Theta(\sqrt{k}))$-bit, $w\in\mathbb{N}$, record databases of size $N$, with upload cost $O\left(k^2 \cdot 2^{6\sqrt{\log N\log\log N}}\right)$ and download cost $2w\cdot k$. Consequently, the rate of the PIR is $1-1/\Theta(\sqrt{k})$.
\end{theorem}

\begin{proof}

We will first need the following fact.

\begin{fact}[\cite{EH63}]
Let $k_0>0$ be an integer.
There exists a function $\delta_{k_0}\colon\mathbb{N}\rightarrow [0,\infty)$ such that $\lim_{q\rightarrow \infty}\delta_{k_0}(q)=0$ and such that for every integer $q>0$ there is a family of $(k_0-1)$-sized subsets from $[q]$, $\mathcal{S}_q\subseteq\binom{[q]}{k_0-1}$, with the additional requirement that any two distinct sets $S_1,S_2\in\mathcal{S}_q$ satisfy $|S_1\cap S_2|\leq 1$, and such that $|\mathcal{S}_q|=(1-\delta_{k_0}(q))\frac{q(q-1)}{k_0(k_0-1)}$.
\end{fact}

Now, we combine the above fact, Lemma \ref{th:bb3}, Proposition \ref{prop:mvpir}, and the modification to Lemma \ref{black-box} given in Remark \ref{remark:black-box-smaller-t}, which yields our claim with $w=1$. To obtain it for $w>1$ we can just repeat the scheme $w$ times. 
\end{proof}

Moreover, similar results for sub-polynomial upload cost PIR protocols with ($t\geq 2$)-privacy could also be obtained. These protocols exist, for example, by combining the results of \cite{BIW07} with the protocols from \cite{STOC:Yekhanin07,STOC:Efremenko09,BeimelIKO12}.

\section{Improving Rate via Nonlinear Reconstruction and a Small Failure Probability}
\label{sec:nonlinear}

In Section~\ref{sec:linearNegative}, we saw that that the rate of $1 - dt/k$ is the best possible for linear HSS.  In this section, we will see that if we allow non-linear reconstruction with an exponentially small failure probability, then we can in fact beat this bound.

\begin{remark}[Efficiency of $\Rec$]\label{rem:efficient}
For this section, we do not attempt to address the computational efficiency of the $\Rec$ algorithm in our main theorem statement.  We note however that there are efficient algorithms for a simpler scheme that still can beat the $1 - dt/k$ barrier; see Remark~\ref{rem:naiveapproach}.

Our main theorem, Theorem~\ref{thm:slepianWolf}, offers stronger guarantees on the download rate than the simple scheme in Remark~\ref{rem:naiveapproach}.  The tools that we use to prove it are information-theoretic---the approach is based on classical Slepian-Wolf coding---and do not immediately yield efficient algorithms for $\Rec$.  However, we are hopeful that progress in efficient implementations of Slepian-Wolf coding, e.g.~\cite{GC06, CHJ09, SMT15}, may be applicable in our setting as well.
\end{remark}

\paragraph{A warm-up.} To motivate how this relaxation to non-linear schemes might help, consider the following simple example. Suppose that $m=d$ and we would like to construct a $1$-private HSS for $\{f\}^\ell$, where $f(x_1, \ldots, x_d) = \prod_{i=1}^d x_i$, based on Shamir sharing.  

Let $k=d+1$. Let $q$ be a prime power so that $q>k$ and suppose that $|\F|=q\approx k$. 
Given inputs $x_{i,r} \in \F$ for $i \in [d]$ and $r \in [\ell]$, each server $j$ holds the input shares $(p_{i,r}(\alpha_j))_{i,r}$, where $p_{i,r}$ is a random degree $1$ polynomial over $\F$, 
 with $p_{i,r}(0) = x_{i,r}$.  Then each party computes the output shares
\[ \Eval(f, j, (p_{i,r}(\alpha_j))_{i,r}) = \left(\prod_{i=1}^d p_{i,r}(\alpha_j)\right)_{r \in [\ell]} =: \by^j \in \FF^\ell.\]
Then, since $k > d$, we may recover $\prod_{i \in [d]} x_{i,r}$ for each $r \in [\ell]$ from the output shares by polynomial interpolation.
The download cost of this scheme is $\ell k \log_2(q)$.

Now consider the following improvement on the same scheme.  We observe that for $r \in [\ell]$,
\[ \Pr[ y^j_r = 0 ] = (1 - 1/q)^d \approx 1/e. \]
Thus, with high probability, there are only about $\ell/e$ nonzero entries in each $\by^j$.  This allows each party to compress their output shares.  Rather than sending $\by^j$ itself, each party can send the location and value of each of the nonzeros of $\by^j$.  With high probability, the number of bits that each party needs to send is
\newvictor{\[ \log_2 {\ell \choose \ell/e } + \frac{\ell \log_2(q)}{e} = \ell \left( (1 + o(1))H(1/e) + \frac{  \log_2(q) }{e} \right). \]}
For large $d$ (with $k=d+1$ and $q\approx k$), this is a savings of about a factor of $e$ over the naive version.



\paragraph{A general methodology for HSS share compression.} The previous example  already shows that we can surpass the barrier of $1 - dt/k$ for linear schemes from Theorem~\ref{thm:barrier} by using a simple compression technique. However, it only applies to evaluating products over a big finite field, which (despite being natural) is not useful for any applications we are aware of. Moreover, this naive compression method is entirely {\em local}, not taking advantage of correlations between output shares.
In the following we develop a more general framework for compressing HSS shares by using Slepian-Wolf coding. We apply this methodology to a simple and well-motivated instance of HSS, where $f$ computes $\F_2$ multiplication (i.e., the AND of two input bits). This can be motivated a ``dense'' variant of the private set intersection problem, where the sets are represented by their characteristic vectors. We further show that, unlike in the above warm-up example, here it is possible to obtain improved rate while ensuring that the output shares reveal no additional information except the output.
\medskip

The rest of this section is organized as follows. In Section~\ref{sec:slepianWolf}, we show that it is possible to push the above line of reasoning as far as the classical \emph{Slepian-Wolf theorem} allows.  As described in Section~\ref{sec:SWtech}, the Slepian-Wolf theorem describes the extent to which we can compress (possibly correlated) sources that are held by different parties.  We cannot use the Slepian-Wolf theorem directly because the underlying joint distribution is not known to each server (as it depends on the value of the secret), but we show how to adapt the proof to deal with this. We apply this general methodology to optimize the download rate of 3-server HSS for AND ($\F_2$ multiplication).

Next, in Section~\ref{sec:symmetricHSS}, we consider a desirable {\em symmetric privacy} property, ensuring that the output shares reveal no information about the input beyond the output. While symmetric privacy is easy to achieve for HSS with linear reconstruction, this is not necessarily true for HSS with nonlinear reconstruction. Indeed, we show that Shamir-based HSS in the above warm-up example does not satisfy this property. However, somewhat unexpectedly, we show that the optimized HSS for AND from Section~\ref{sec:slepianWolf} does satisfy this requirement.

\subsection{Beating the \texorpdfstring{$1 - dt/k$}{1-dt/k} Barrier via Slepian-Wolf Coding}\label{sec:slepianWolf}


In this section we introduce a broad class of non-linear, $(1 - \exp(-\Omega(\ell)))$-correct HSS schemes for function classes of the form $\mathcal{F}^\ell$ when $\ell$ is large, which can circumvent the linear impossibility result of Theorem~\ref{thm:barrier} for linear, $1$-correct HSS schemes.

\subsubsection{Notation and theorem statements}\label{sec:SWnotation}
We begin by setting up some further notation for this section.  Suppose we have a function class $\mathcal{F}$, and suppose we have a $k$-server HSS scheme $\Pi = (\Share, \Eval, \Rec)$ for $\mathcal{F}$.  We will formally state how we use $\Pi$ below when we define our new HSS scheme $\Pi'$; for now we sketch enough of it to set up our notation.  Given $\bx_1, \ldots, \bx_\ell \in \mathcal{X}^m$, we run the $\Share$ and $\Eval$ functions of $\Pi$ independently to obtain $\ell$ output shares (outputs of $\Eval$) for each server $j$.  Let $z_{j,i} \in \mathcal{Z}_j$ be the $i$'th  output share held by server $j$, for $j \in [k]$ and for $i \in [\ell]$.  Let $\bZ$ be the $k \times \ell$ matrix whose $(j,i)$ entry is $z_{j,i}$.  We refer to the columns of $\bZ$ as $\bz_i$ and the rows as $\by_j$.  Letting \[\mathcal{Z} = \mathcal{Z}_1 \times \mathcal{Z}_2 \times \cdots \times \mathcal{Z}_k,\]
we also view $\bZ$ as an element of $\mathcal{Z}^\ell$, that is, as a sequence of the columns of the matrix $\bZ$.  Abusing notation, we will move back and forth between these two views.

Each $\bx \in \mathcal{X}^m$ induces a distribution on $\bz \in \mathcal{Z}$,  obtained by first using $\Share$ to share $\bx$ among the $k$ servers, and then using $\Eval$ to arrive at output shares $\bz$.  The randomness of this distribution comes from the randomness in the probabilistic algorithms $\Share$ and $\Eval$.  Let $\mathcal{D}$ be the set of all probability distributions on $\mathcal{Z}$ that arise this way.  Notice that $|\mathcal{D}| \leq |\mathcal{X}|^{m}$, the number of possible values for $\bx$.

We define $\Delta(\mathcal{D})$ to be the collection of probability distributions on $\mathcal{D}$.  We define $\Delta_\ell(\mathcal{D}) \subset \Delta(\mathcal{D})$ to be the collection of distributions with denominator $\ell$.  That is,
\[ \Delta_\ell(\mathcal{D}) = \{ \pi \in \Delta(\mathcal{D}) : \forall D \in \mathcal{D}, \exists i \in \mathbb{Z}, \pi(D) = i/\ell\}. \]
Notice that 
\begin{equation}\label{eq:sizeofDelta}
|\Delta_\ell(\mathcal{D})| \leq \ell^{|\mathcal{D}|} \leq \ell^{|\mathcal{X}|^m}.
\end{equation}

For $\pi \in \Delta(\mathcal{D})$, define a distribution $\sigma_\pi$ on $\mathcal{Z}$ by
\[ \sigma_\pi(\bz) := \sum_{D \in \mathcal{D}} \pi(D) \cdot D(\bz), \]
where for a distribution $D$, $D(\bz)$ denotes the probability of $\bz$ under $D$.  Thus, $\sigma_\pi$ is defined by first drawing $D \sim \pi$ and then drawing $\bz \sim D$.

With all of this notation out of the way, we may state the main result in this section, which states that given a particular HSS scheme $\Pi = (\Share, \Eval, \Rec)$ for a function class $\mathcal{F}$, we may ``compress'' the outputs of $\Eval$ to create a new HSS scheme for $\mathcal{F}^\ell$, for sufficiently large $\ell$, with potentially better download rate.  As noted in Section~\ref{sec:SWtech}, this theorem is closely related to the classical Slepian-Wolf theorem for compression of dependent sources with separate encoders and a joint decoder. 

\begin{theorem}[Slepian-Wolf coding for HSS shares]\label{thm:slepianWolf}
Suppose there is a $t$-private $k$-server HSS $\Pi = (\Share, \Eval, \Rec)$ for $\mathcal{F}$.  Let $\mathcal{D}$ be defined as above. There is a function $\delta:\mathbb{R} \to \mathbb{R}$ so that $\delta(\eps) \to 0$ as $\eps \to 0$, and constants $C$ and $\ell_0$ (all of which may depend on $\Pi$ and its parameters) so that the following holds for any $\eps > 0$ and any $\ell \geq \ell_0$.  

For a distribution $\pi' \in \Delta(\mathcal{D})$, let $\bz' \in \mathcal{Z}$ denote a random variable drawn from $\sigma_{\pi'}$.  
Suppose that $b_1, b_2, \ldots, b_k$ are non-negative integers so that for all $S \subseteq [k]$,
\[ \sum_{i \in S}b_i \geq \max_{\pi' \in \Delta(\mathcal{D})} \ell \cdot ( H(\bz'_S | \bz'_{S^c}) + \delta(\eps) ). \]

Then there is a $t$-private $k$-server $\alpha$-correct HSS scheme for $\mathcal{F}^\ell$ with $\alpha = 1 - \exp(-C \eps^2 \ell)$, and with download cost $\sum_{i=1}^k b_i$. 

\end{theorem}

\begin{remark}
In order to interpret the bound on the $b_i$'s in Theorem~\ref{thm:slepianWolf}, let us consider it in terms of the download cost of the naive scheme that repeats $\Pi$ $\ell$ times.  Suppose that, in the original scheme $\Pi$, server $i$ sends $b_i$ bits, so server $i$ sends $\ell b_i$ bits in the naive scheme.  Since $\bz_S' \in \{0,1\}^{\sum_{i\in S} b_i}$, we (naively) have
\[ H(\bz_S' | \bz_{S^c}') \leq \sum_{i \in S} b_i \]
for all $S \subseteq [k]$, and the total download cost is $\sum_{i=1}^k b_i$.
Thus (ignoring the slack $\delta(\eps)$), the naive bound on $H(\bz_S' | \bz_{S^c}')$ yields the download cost of the naive scheme.  As we will see below, by obtaining better-than-naive bounds on this conditional entropy term, we will be able to beat the naive scheme, resulting in a new scheme that has an improved download rate.
\end{remark}

\begin{remark}\label{rem:converse}
Theorem~\ref{thm:slepianWolf} is tight in the sense that if we are to follow the strategy of ``compress the output shares of a fixed HSS $\Pi$,'' $\max_{\pi'}\ell \cdot H(\bz_S' | \bz_{S^c}' )$ is the best download cost we can hope for.  This follows from the standard converse for Slepian-Wolf coding, by allowing an adversary to choose inputs that yield a sequence $D_1, \ldots, D_\ell$ of distributions in $\mathcal{D}$ with empirical distribution that is approximately the maximizer $\pi^*$.
\end{remark}

Before we prove the theorem, we apply it to a simple case where we would like to multiply two bits.  We will apply it to the following HSS. 


\begin{definition}[Greedy-Monomial CNF HSS]\label{def:greedy}
Let $t,k,d,m$ be positive integers with $k>dt$  and let $\FF$ be a finite field.
Define a $t$-private $k$-server HSS $\Pi = (\Share, \Eval, \Rec)$ for $\POLY_{d,m}(\FF)$ as follows.  
\begin{itemize}
    \item \textbf{Sharing.} The $\Share$ function is given by $t$-CNF sharing.  To set notation, suppose that server $j$ receives $\by^j = (X_{i,S} : j \not\in S)$ where $X_{i,S}$ for $i \in [m]$ and $S \subset [k]$ of size $t$ are random so that $\sum_S X_{i,S} = x_i$.
    \item \textbf{Evaluation.} Let $f \in \POLY_{d,m}(\FF)$.  We may view $f(x_1, \ldots, x_m)$ as a polynomial $F(\bX)$ in the variables $\bX = (X_{i,S})_{i \in [m],S \subset [k]}$.  Each server $j$ can form some subset of the monomials $\prod_{s=1}^r X_{i_s, S_s}$ that appear in $F(\bX)$.  Server $1$ greedily assembles all of the monomials in $F(\bX)$ that they can; the sum of these monomials is $\Eval(f,1,\by^1)$.  Inductively, Server $j$ greedily assembles all of the monomials in $F(\bX)$ that they can and that have not been taken by Servers $1, \ldots, j-1$, and the sum of these monomials is $\Eval(f,j, \by^j)$. 
    \item \textbf{Reconstruction.} By construction, $f(x_1, \ldots, x_m)$ is equal to $\sum_j \Eval(f,j,\by^j)$.  Thus, $\Rec$ is defined additively.
\end{itemize}
We refer to this $\Pi$ as the $t$-private, $k$-server \emph{greedy monomial CNF} HSS.
\end{definition}

\begin{corollary}\label{cor:beatlinear}
Let $\mathcal{F} = \{f(x_1,x_2) \mapsto x_1 \cdot x_2\}$.  Let $m=2$ and $\mathcal{X} = \mathcal{Y} = \mathbb{F}_2$, so $f:\mathcal{X}^m \to \mathcal{Y}.$  
For sufficiently large $\ell$, the Greedy Monomial CNF HSS is a
3-server, 1-private, $\left(1 - 2^{-\Omega(\ell)}\right)$-correct HSS for the function family $\mathcal{F}^\ell$, with download rate $R \geq 0.376$.

Moreover, for these parameters, the Greedy Monomial CNF HSS yields the best download rate when plugged into Theorem~\ref{thm:slepianWolf}, out of all $\F_2$-linear HSS schemes.
\end{corollary}

\begin{remark}[Greedy Monomial CNF admits optimal compression for $3$-server AND]\label{rem:tight}
Corollary~\ref{cor:beatlinear} shows that the simple Greedy Monomial CNF HSS can be modified to exceed the $1 - dt/k$ limit of Theorem~\ref{thm:barrier} for linear schemes, even in the extremely bare-bones case of multiplying two bits.  (Notice that the Shamir-based example at the beginning of Section~\ref{sec:nonlinear} requires large fields).  Moreover, the ``moreover'' part of the corollary, along with Remark~\ref{rem:converse}, suggests that the the rate $0.376$ is the best achievable by these methods for this problem. In comparison, a symmetric monomial assignement that assigns 3 monomials to each server achieves rate $\approx 0.350$.
\end{remark}

\begin{remark}[A simpler way to get the same qualitative result]\label{rem:naiveapproach}
In fact, we beat the $1 - dt/k$ bound without the full power of Theorem~\ref{thm:slepianWolf}, even over $\FF_2$.  Rather than compressing the joint distribution of the output shares, each server may compress their share individually.  Notice that by the privacy guarantee, the distribution of each server's share does not depend on the secret; thus, if this distribution is not uniform, it can be compressed.  As $\ell$ grows, the download rate achievable can approach
\begin{equation}\label{eq:naiverate}
\frac{\log_2|\mathcal{Y}|}{\sum_{j=1}^k H(y^j) }, 
\end{equation}
where $y^j$ is the output share for party $j$.  

This straightforward scheme, while it does not yield bounds as strong as those obtained by Theorem~\ref{thm:slepianWolf}, already allows us to beat the $1 - dt/k$ bound of Theorem~\ref{thm:barrier}.  For example, for the $3$-party Greedy Monomial CNF scheme for multiplying two bits, the rate in \eqref{eq:naiverate} is easily seen to be $R \approx 0.367 > 1/3$.

Further, this scheme immediately comes with efficient $\Eval$ and $\Rec$ algorithms, using known efficient algorithms for optimal compression (for example Lempel-Ziv coding).
\end{remark}

\begin{proof}[Proof of Corollary~\ref{cor:beatlinear}]
To prove the first part of the corollary, we apply Theorem~\ref{thm:slepianWolf}, which states that there is an HSS scheme $\Pi'$ for $\mathcal{F}^\ell$ of download cost $\sum_i b_i$ provided that, for each $S \subseteq[k]$,
\begin{equation}\label{eq:needToOpt}
 \max_{\pi' \in \Delta(\mathcal{D})} H(\bz'_S | \bz'_{S^c}) \leq \sum_{i\in S}b_i,
\end{equation}
where $\bz' \sim \sigma_\pi$.
To prove the theorem we describe a particular HSS scheme to use as $\Pi$.
Our scheme $\Pi$ is the $1$-private $3$-party greedy CNF HSS, as in Definition~\ref{def:greedy}.

We will compute the quantity on the left-hand-side of \eqref{eq:needToOpt} for $\Pi$.
We note that when $(X,Y)$ are random variables with joint distribution $\mu$, the function $H(X|Y)$ is concave in $\mu$.   This implies that if we view $H(\bz'_S | \bz'_{S^c})$ as a function of the distribution $\sigma_{\pi'}$ on $\bz \in \mathcal{Z}$ that is induced by $\pi'$, the function is concave in $\sigma_{\pi'}$.  Thus, we may solve a convex optimization problem to find the optimal $\sigma_{\pi'}$, and hence the optimal value of \eqref{eq:needToOpt}, for each $S \subseteq [k]$.\footnote{This approach works for any original HSS scheme $\Pi$, but for the particular $\Pi$ defined above, the situation is even simpler because $\mathcal{D}$ consists of only two distributions.
Thus, the problem becomes a univariate convex optimization problem that can be easily solved using calculus.}
Implementing this, we obtain the following requirements on $b_0,b_1,b_2$.\footnote{The code that we used to obtain this can be found at \url{https://web.stanford.edu/~marykw/files/HSS_SlepianWolf.sage}.  This code also obtains methods for finding the requirements on $b_i$ in Theorem~\ref{thm:slepianWolf}, starting from more general original HSS schemes $\Pi$.} 
\begin{align*}
\frac{b_0}{\ell} &\geq 0.75089\\
\frac{b_1}{\ell} &\geq 0.90690\\
\frac{b_2}{\ell} &\geq 0.90690\\
\frac{b_0 + b_1}{\ell} &\geq 1.70429\\
\frac{b_0 + b_2}{\ell} &\geq 1.70429\\
\frac{b_1 + b_2}{\ell} &\geq 1.84745\\
\frac{b_0 + b_1 + b_2}{\ell} &\geq 2.65873
\end{align*}
We see that if we take $b_1 = b_2 = \frac{\ell}{2} \cdot 1.84745 = \ell \cdot 0.92372$ and if we take $b_0 = \ell \cdot 2.65873 - b_1 - b_2 = \ell \cdot 0.81128$ then all of these requirements are satisfied, and we have
\[ b_0 + b_1 + b_2 = \ell \cdot 2.65873. \]
Thus, Theorem~\ref{thm:slepianWolf} implies that there is an HSS scheme $\Pi'$ that approaches this download cost, provided that $\ell$ is sufficiently large.  Thus, for any download rate $R$ satisfying
\[ R < \frac{1}{2.65873} \approx 0.37612, \]
there is some large enough $\ell$ so that 
there is an HSS scheme for $\mathcal{F}^\ell$ with download rate $R$.  This proves the first part of the corollary.

For the ``moreover'' part of the corollary, we first observe that any linear HSS scheme is some sort of monomial assignment.  Indeed, writing the shares of the inputs $x_1,x_2$ as $X_{i,j}$ for $i \in [2]$ and $j \in [3]$ (where $x_i = \sum_j X_{i,j}$), then each server's output share is some polynomial in the $X_{i,j}$.  
Next, we simply enumerate over all such assignments and check the values in Theorem~\ref{thm:slepianWolf}. 

In more detail,
as we assume that the HSS scheme is linear, we must include all monomials of the form $X_{1,j} X_{2,j}$ an odd number of times across the servers.  Up to symmetries, there are only two ways to distribute these degree-$2$ monomials (the way prescribed by the Greedy Monomial CNF assignment, and one other way that assigns three monomials to each server).  We must include each linear monomial $X_{i,j}$ an even number of times across the servers, as no degree-$1$ monomial occurs in $x_1 x_2 = \sum_{i,j} X_{1,i}X_{2,j}$.  There are 64 ways to do this (including ways that are equivalent by symmetry).  Finally, degree $3$ or higher monomials cannot occur in any server, as any such monomial $X_{i_1, j_1} X_{i_2,j_2}, X_{i_3,j_3}$ must contain two distinct values in $\{j_1, j_2, j_3\}$ and thus can only be held by one server and cannot cancel out.  Thus, there are only 128 things to check (again, with some redundancy by symmetry).  This is easily done with the code referenced earlier.  We see that for all of these settings, the condition on $b_1 + b_2 + b_3$ is one of the following three requirements:
\begin{align*} 
 \frac{b_0 + b_1 + b_2}{\ell} &\geq 2.65873\\
 \frac{b_0 + b_1 + b_2}{\ell} &\geq 2.90564\\
 \frac{b_0 + b_1 + b_2}{\ell} &\geq 2.85200
\end{align*}
The first of these is the binding constraint for the Greedy Monomial CNF scheme.  The other two are larger, meaning that there is no linear HSS that can outperform the Greedy Monomial CNF HSS.
\end{proof}

\subsubsection{Proof of Theorem~\ref{thm:slepianWolf}}
As discussed in Section~\ref{sec:SWtech}, the statement of Theorem~\ref{thm:slepianWolf} is quite similar to the classical Slepian-Wolf coding theory.  The difference between our setup and that of the Slepian-Wolf theorem is that we do not know the underlying distributions, which depend on the secrets.  However, our proof follows the basic outline of the classical Slepian-Wolf argument; the main difference is that we need to take a union bound over the distributions $\pi \in \Delta_\ell(\mathcal{D})$ that may arise, and we need to argue that the distribution of output shares arising from a sequence $(D_1, \ldots, D_\ell)$ with empirical distribution $\pi$ is sufficiently close to i.i.d. samples from $\sigma_\pi$.\footnote{The empirical distribution $\pi$ of $(D_1, \ldots, D_\ell) \in \mathcal{D}^\ell$ is sometimes called the \emph{type} of $(D_1, \ldots, D_\ell)$.}

We begin with some necessary notions from information theory.  We tailor these notions to our notation and setting for readability; in particular, the notation is non-standard from an information-theory point of view.

We use the notation from the beginning of Section~\ref{sec:SWnotation}.
Given $\bZ \in \mathcal{Z}^\ell$ and given $\bz \in \mathcal{Z}$ we write $f_{\bZ}(\bz)$ to denote the number of times that $\bz$ appears in $\bZ$.  (Viewing $\bZ$ as a matrix, this is the number of times that $\bz$ appears as a column of $\bZ$).

For a set $S \subset [k]$, let $\mathcal{Z}_S = \prod_{j \in S} \mathcal{Z}_j$ (recalling that $\mathcal{Z} = \prod_{j=1}^k \mathcal{Z}_j$). For $\bz \in \mathcal{Z}$, we write $\bz_S \in \mathcal{Z}_S$ to denote the restriction of $\bz$ to the coordinates indexed by $S$.    Given $\bz_S \in \mathcal{Z}_S$ and $\bz_{S^c} \in \mathcal{Z}_{S^c}$, we denote by $\bz^S \circ \bz^{S^c}$ the vector $\bz \in \mathcal{Z}$ obtained by combining the two of them in the natural way (that is, $z_j = z^S_j$ if $j \in S$ and $z^{S^c}_j$ otherwise).  Given $\bZ_S \in (\mathcal{Z}_S)^\ell$ and $\bZ_{S^c} \in (\mathcal{Z}_{S^c})^\ell$, we define $\bZ_S \circ \bZ_{S^c} \in \mathcal{Z}^\ell$ similarly.

\begin{definition}\label{def:typical}
Let $\eps > 0$.
Let $\sigma$ be a probability distribution on $\mathcal{Z}$.  
We say that $\bZ \in \mathcal{Z}^\ell$ is a \emph{$\eps$-strongly typical sequence} for $\sigma$ of length $\ell$ if for all $\bz \in \mathcal{Z}$ so that $\sigma(\bz) = 0$, we have $f_\bZ(\bz) = 0$; and for all $\bz \in \mathcal{Z}$ so that $\sigma(\bz) > 0$, we have
\[ |f_\bZ(\bz)/\ell - \sigma(\bz)| \leq \frac{\eps}{|\mathcal{Z}|}. \]
We define
\[ T_{\eps,\sigma}^{(\ell)} = \left\{ \bZ \in \mathcal{Z}^\ell \,:\,
\bZ \text{ is $\eps$-strongly typical for $\sigma$ } \right\}. \]

Let $S \subset [k]$ and fix $\bZ_{S^c} = (\bz^1_{S^c}, \ldots, \bz^\ell_{S^c}) \in (\mathcal{Z}_{S^c})^\ell$.  We define
\[ T_{\eps, \sigma, S}^{(\ell)}(\bZ_{S^c}) = \left\{ \bZ_S \in (\mathcal{Z}_S)^\ell \,:\, \bZ_S \circ \bZ_{S^c} \in T_{\eps, \sigma}^{(\ell)}  \right\}.
\]
\end{definition}
That is, $T_{\eps, \sigma}^{(\ell)}$ is the set of all $\bZ$ that have about the right number of copies of each $\bz \in \mathcal{Z}$, and $T_{\eps, \sigma,S}^{(\ell)}(\bZ_{S^c})$ is the set of all ways $\bZ_{S^c}$ to complete $\bZ_{S^c}$ in order to have that property.

We will use the following well-known fact (see, for example 
\cite[Eq. (10.175)]{coverThomas}).
\begin{theorem}\label{thm:infoTheory}
Let $0 < \eps' < \eps$.  Let $S \subseteq[k]$.  Suppose that $\ell$ is sufficiently large.
Let $\sigma$ be a distribution on $\mathcal{Z}$, and let $\sigma_S$ be the restriction to $\mathcal{Z}_S$.
Let $\bz' \sim \sigma$ be a random vector drawn from $\sigma$.
For all $\bZ_{S^c} \in \mathcal{Z}^\ell_{S^c}$, 
\[ |T_{\eps, \sigma, S}^{(\ell)}(\bZ_{S^c})| \leq (\ell + 1)^{|\mathcal{Z}|}\cdot 2^{\ell H( \bz'_S | \bz'_{S^c})(1 + \delta(\eps))}, \]
for some function $\delta$ so that $\delta(\eps) \to 0$ as $\eps \to 0$.
\end{theorem}

With this machinery in place, we can prove Theorem~\ref{thm:slepianWolf}.

\begin{proof}[Proof of Theorem~\ref{thm:slepianWolf}]
We use the notation listed from the beginning of Section~\ref{sec:SWnotation} as well as earlier in this section.
Define
\[ T^{(\ell)}_\eps := \bigcup_{\pi \in \Delta_\ell(\mathcal{D})} T_{\eps, \sigma_\pi}^{(\ell)}. \]

We now define our new HSS scheme, $\Pi' = (\Share', \Eval', \Rec')$.  Let $\Pi = (\Share, \Eval, \Rec)$ be the HSS scheme assumed in the theorem statement.

\begin{itemize}
    \item \textbf{Sharing.} Suppose the inputs are $\bx_1, \ldots, \bx_\ell \in \mathcal{X}^m$, where $\bx_i = (x_{i,1}, x_{i,2}, \ldots, x_{i,m}).$  We define the new $\Share'$ function to just use $\Share$ to independently share each input $x_{i,r}$.
    Organize server $j$'s shares into vectors
    \[ \bx_i^j = (x_{i,1}^j, x_{i,2}^j, \ldots, x_{i,m}^j ) \in \mathcal{X}^m \]
    for each $i \in [\ell]$.
    \item \textbf{Evaluation.} Let $f = (f_1, \ldots, f_\ell) \in \mathcal{F}^\ell$.
    For each $i \in [\ell]$ and for each server $j \in [k]$, we use the $\Eval$ function from the original HSS scheme to obtain
    \[ z_{j,i} = \Eval(f_i, j, \bx_i^j) \]
    Let $\bZ \in \mathcal{Z}^\ell$ be the organization of the $z_{j,i}$ as described in Section~\ref{sec:SWnotation}.  Note that we may also think of $\bZ$ as a matrix with $\bZ_{j,i} = z_{j,i}$.  In this notation, server $j$ holds $\by^j$, the $j$'th row of $\bZ$.

    For each server $j \in [k]$, choose a random map\footnote{\newvictor{The servers can have access to the same random map by assuming they share common randomness. Alternatively, we can use pairwise independent hash functions, whose description is short and can be given as part of the input shares.}} $h_j: \mathcal{Z}_j^\ell \to \{0,1\}^{b_j}$.  Define a new evaluation function $$\Eval'(f,j,(x_1^j, \ldots, x_m^j)) := h_j(\by^j).$$  For notational convenience later on, let $h(\bZ) = (h_1(\by^1), \ldots, h_k(\by^k))$.
    \item \textbf{Reconstruction.} Given output shares $(v^1, \ldots, v^k) \in \{0,1\}^{b_1 + \cdots + b_k}$, we wish to reconstruct $f_i(\bx_i)$ for $i \in [\ell]$.  We define the new reconstruction map, $\Rec'$, as follows.  Choose a parameter $\eps > 0$.
    \begin{itemize}
        \item If there is a unique $\tilde{\bZ} \in T^{(\ell)}_{\eps}$ with rows $\tilde{\by}^j$ so that $h_j(\tilde{\by}^j) = v^j$ for all $j \in [k]$, define 
        \[ \Rec'(v^1, \ldots, v^k) = \Rec(\tilde{\bz}^1), \Rec(\tilde{\bz}^2), \ldots, \Rec(\tilde{\bz}^\ell), \] where $\tilde{\bz}^i$ is the $i$'th column of $\tilde{\bZ}$.
        \item If there is not a unique such $\tilde{\bZ}$, define $\Rec'(v^1, \ldots, v^k) = \bot$.
    \end{itemize}
\end{itemize}
We first note that $\Pi'$ is $t$-private, since $\Pi$ is.  We also note that the download cost of $\Pi'$ is indeed $b_1 + \cdots + b_k$.  
It remains to show that this scheme is correct with high probability, provided that the $b_i$ satisfy the requirements in the theorem statement. 

Suppose that the secrets $\bx_1, \ldots, \bx_\ell$ are shared and evaluated using $\Share'$ and $\Eval'$ as above. 
Server $j$ holds $\by^j$, which is organized into $\bZ \in \mathcal{Z}^\ell$, as per the definition of $\Pi'$ above.
The $i$'th columns $\bz_1, \ldots, \bz_\ell$ of $\bZ$ is generated according to a distribution $D_i \in \mathcal{D}$ that depends on $\bx_i$.  Suppose that $(D_1, D_2, \ldots, D_\ell)$ is this vector of distributions, and let $\pi \in \Delta_\ell(\mathcal{D})$ be the empirical distribution of $(D_1, \ldots, D_\ell)$.  That is, 
\[ \pi(D) := \frac{|\{i \in [\ell] \,:\, D_i = D\}|}{\ell}. \]
Define
\[ \mathcal{S}_D = \{ i \in [\ell] \,:\, D_i = D\}.\]

We first observe that if the correct matrix $\bZ$ is identified in the definition of $\Rec'$ as the unique $\tilde{\bZ} \in T_\eps^{(\ell)}$, then the decoding algorithm $\Rec'$ is correct.  Indeed, this follows from the correctness of $\Rec$.  
Thus, to show that this algorithm is correct, it suffices to show that with high probability (over $\Share$, $\Eval'$), the following two things hold.
\begin{itemize}
    \item[(A)] $\bZ$ lies in $T_\eps^{(\ell)}$.
    \item[(B)] There is no other $\tilde{\bZ} \in T_\eps^{(\ell)}$ so that $h(\tilde{\bZ}) = h(\bZ)$.
\end{itemize}

We begin with (A).  
\begin{claim}\label{cl:A}
Let $\bZ \in \mathcal{Z}^\ell$ be as above.  Then with probability at least $1 - \exp(-\Omega(\eps^2 \ell))$, 
$\bZ$ is $\eps$-strongly typical for $\sigma_{\pi}$.  That is,
$\bZ \in T_{\eps, \sigma_{\pi}}^{(\ell)}. $

Above, the $\Omega(\cdot)$ hides constants that may depend on the original HSS scheme $\Pi$ and its parameters ($m,k,|\mathcal{X}|, \mathcal{Z}$).
\end{claim}
\begin{proof}
We must show that $\bZ$ is $\eps$-strongly typical for $\sigma_{\pi}$, or in other words that for all $\bz \in \mathcal{Z}$ with $\sigma_\pi(\bz) = 0$, we have $f_{\bZ}(\bz) = 0$; and for all $\bz \in \mathcal{Z}$ with $\sigma_\pi(\bz) > 0$, we have
\[ \left| \frac{ f_{\bZ}(\bz) }{\ell} - \sigma_{\pi}(\bz) \right| \leq \eps/|\mathcal{Z}|. \]

To show the first thing, suppose that $\sigma_\pi(\bz) = 0$.
Then
\[ 0 = \sigma_{\pi}(\bz) = \sum_{D \in \mathcal{D}} \pi(D) D(\bz), \]
and so for each $D$, either $\pi(D) = 0$, in which case $D_i$ is never equal to $D$, or else $D(\bz) = 0$, in which case there is no way to have $\bz^i_S = \bz$ if $D_i = D$.  In either case, if $\sigma_\pi(\bz) = 0$, then $\bz$ will never appear in $\bZ$ and $f_{\bz}(\bZ) = 0$.  

Next, fix any $\bz \in \mathcal{Z}$ so that $\sigma_\pi(\bz) > 0$.
The number of $\bz$ that appear in $\bZ$ is
\[ f_{\bZ}(\bz) = \sum_{i=1}^\ell \mathbf{1}[\bz^i = \bz] = \sum_{D \in \mathcal{D}} \sum_{i \in \mathcal{S}_D} \mathbf{1}[\bz^i = \bz]. \]
Recall that we have $\bz^i \sim D_i$ for $i \in [\ell]$, each independently (and the distributions $D_i \in \mathcal{D}$ are fixed by the secrets $\bx^i$).
The expectation of $f_{\bZ}(\bz)$ is thus
\begin{align*}
\mathbb{E}[ f_{\bZ}(\bz) ] &= \sum_{D \in \mathcal{D}} \ell \pi(D) \Pr_D[\bz^i = \bz] \\
&= \sum_{D \in \mathcal{D}} \ell \pi(D) D(\bz) \\
&= \ell \sigma_{\pi}(\bz).
\end{align*}
Thus, Hoeffding's inequality implies that, for all such $\bz$,
\begin{align*} 
\Pr\left[\left| \frac{ f_{\bZ}(\bz) }{\ell} - \sigma_{\pi}(\bz) \right| > \eps / |\mathcal{Z}| \right] 
&= \Pr\left[ \left|\frac{1}{\ell}\sum_{i=1}^\ell \left( \mathbf{1}[\bz^i = \bz] - \mathbb{E}\mathbf{1}[\bz^i = \bz] \right) \right| > \eps / |\mathcal{Z}| \right] \\
&\leq 2\exp\left( - 2\ell \eps^2 / |\mathcal{Z}|^2 \right)
\end{align*}

Since $|\mathcal{Z}|$ is some constant that depends only on the original HSS scheme $\Pi$ (and in particular, is independent of $\ell$ and $\eps$),
we conclude that for any $\bz \in \mathcal{Z}$, 
\[ \Pr\left[\left| \frac{ f_{\bZ}(\bz) }{\ell} - \sigma_{\pi}(\bz) \right| > \eps/|\mathcal{Z}| \right] \leq \exp\left( -\Omega( \ell \eps^2 )\right). \]
Taking a union bound over all possible $\bz \in \mathcal{Z}$, we see that
\[ \Pr[ \bZ \not\in T_{\eps, \sigma_{\pi}}^{(\ell)} ] \leq |\mathcal{Z}| \cdot \exp\left( -\Omega(\ell \eps^2)\right) \leq \exp(-\Omega(\ell \eps^2)), \]
using the fact again that $|\mathcal{Z}|$ is a constant. This proves the claim.
\end{proof}

As a corollary of Claim~\ref{cl:A}, we see that item (A) holds with probability at least $1 - \exp(-\Omega(\ell \eps^2))$.  Indeed, with at least that probability we have \[ \bZ \in T_{\eps, \sigma_\pi}^{(\ell)} \subseteq T_\eps^{(\ell)}, \]
using the definition of $T_\eps^{(\ell)}$.

Now we turn to item (B), that there is no other $\tilde{\bZ} \in T_\eps^{(\ell)}$ so that $h(\bZ) = h(\tilde{\bZ})$.  
We establish $2^k$ bad events that could prevent (B) from occurring, one for each $S \subset [k]$.

For $S \subseteq [k]$, define the event $\mathcal{E}_S$ to be the event that there is some $\tilde{\bZ}_S \in \mathcal{Z}_S^\ell$ so that
\begin{itemize}
    \item Each row of $\tilde{\bZ}_S$ is different than the corresponding row of $\bZ_S$.  (That is, if we let $\tilde{\by}_i$ denote the $i$'th row of $\tilde{\bZ}_S$, we have $\tilde{\by}_i \neq \by_i$ for all $i \in S$.)   
    \item $\tilde{\bZ}_S \circ \bZ_{S^c} \in T_\eps^{(\ell)}$
    \item $h(\tilde{\bZ}_S \circ \bZ_{S^c}) = h(\bZ)$
\end{itemize}
Above, we recall the notation that $\tilde{\bZ}_S \circ \bZ_{S^c}$ means that we should create a matrix in $\mathcal{Z}^\ell$ by taking the rows indexed by $S$ from $\tilde{\bZ}$ while taking the rows indexed by $S^c$ from $\bZ$. 

Observe that (B) will occur provided that none of the bad events $\mathcal{E}_S$ occur.  Indeed, suppose that there is some $\tilde{\bZ} \in T_\eps^{(\ell)}$ so that $\tilde{\bZ} \neq \bZ$ and so that  $h(\tilde{\bZ}) =h(\bZ)$.  Then $\mathcal{E}_S$ occurs, where $S \subseteq [k]$ is the set of rows on which $\bZ$ and $\tilde{\bZ}$ differ.

\begin{claim}\label{cl:B}
Fix $\bZ \in \mathcal{Z}^\ell$ and suppose that the favorable outcome of Claim~\ref{cl:A} holds for $\bZ$.  
Let $S \subseteq [k]$.
There is some function $\delta'(\eps)$ so that $\delta'(\eps) \to 0$ and $\eps \to 0$, (and which may depend on the HSS $\Pi$ and its parameters but which is independent of $\ell, \eps$), so that the following holds.

For a distribution $\pi' \in \Delta(\mathcal{D})$, let $\bz' \in \mathcal{Z}$ denote a random variable drawn from $\sigma_\pi'$.  
Suppose that
\[ \sum_{i \in S}b_i \geq \max_{\pi' \in \Delta(\mathcal{D})} \ell \cdot ( H(\bz'_S | \bz'_{S^c}) + \delta'(\eps) ), \]
where the $b_i$ are as in the definition of $\Share'$.
Then 
\[ \Pr[ \mathcal{E}_S ] \leq \exp(-\Omega(\eps \ell)), \]
where above the $\Omega(\cdot)$ notation suppresses constants that depend on the HSS $\Pi$ and its parameters.
\end{claim}
\begin{proof}
Note that for any $\tilde{\bZ}_S$ satisfying the first condition of $\mathcal{E}_S$, we have
\[ \Pr_h[ h(\tilde{\bZ}_S \circ \bZ_{S^c}) = h(\bZ) ] = 2^{-\sum_{i\in S} b_i}, \]
since each $h_j$ is uniformly random and since each row of $\bZ_S$ differs from the corresponding row of $\bZ$.  Thus, we may bound
\begin{align}
\Pr[ \mathcal{E}_S ] &\leq \sum_{\tilde{\bZ}_S \in \mathcal{Z}_S^\ell\setminus\{ \bZ_S \}} \mathbf{1}[\tilde{\bZ}_S \circ \bZ_{S^c} \in T_\eps^{(\ell)}] \cdot 2^{-\sum_{i\in S} b_i} \notag \\
&\leq \sum_{\pi \in \Delta_\ell(\mathcal{D})} \sum_{\tilde{\bZ}_S \in \mathcal{Z}_S^\ell} \mathbf{1}[\tilde{\bZ}_S \circ \bZ_{S^c} \in T_{\eps,\sigma_\pi}^{(\ell)}] \cdot 2^{-\sum_{i\in S} b_i} \label{eq1} \\
&\leq \sum_{\pi \in \Delta_\ell(\mathcal{D})} \left|T_{\eps, \sigma_\pi, S}^{(\ell)}(\bZ_{S^c}) \right| \cdot 2^{-\sum_{i\in S} b_i} \label{eq2} \\
&\leq \sum_{\pi \in \Delta_\ell(\mathcal{D})} (\ell + 1)^{|\mathcal{Z}|}\cdot 2^{\ell H(\bz'_S | \bz'_{S^c})(1 + \delta(\eps))-\sum_{i\in S}b_i} \label{eq3} \\
&\leq \ell^{|\mathcal{X}|^m + |\mathcal{Z}|} \max_{\pi \in \Delta(\mathcal{D})} 2^{ \ell H( \bz'_S | \bz'_{S^c}) (1 + \delta(\eps)) - \sum_{i\in S} b_i } \label{eq4}
\end{align}
where in \eqref{eq1} we have used the definition of $T_\eps^{(\ell)}$, in \eqref{eq2} we have used Definition~\ref{def:typical}, in \eqref{eq4} we have used Theorem~\ref{thm:infoTheory}, and in \eqref{eq4} we have used the fact \eqref{eq:sizeofDelta} that $|\Delta_\ell(\mathcal{D})| \leq \ell^{|\mathcal{X}|^m}$.

Thus, provided that
\[ \sum_{i \in S} b_i \geq \max_{\pi \in \Delta(\mathcal{D})} \ell\left( H(\bz'_S | \bz'_{S^c})(1 + \delta(\eps)) + \eps\right) + (|\mathcal{X}|^m + |\mathcal{Z}|) \log(\ell), \]
the probability that $\mathcal{E}_S$ occurs is at most
\[ \Pr[ \mathcal{E}_S ] \leq \exp(-\Omega(\ell \eps)). \]
Using the fact that $|\mathcal{X}|^m, |\mathcal{Z}|$ are constants that depend only on the original HSS scheme $\Pi$ (and not on $\ell$ or $\eps$), for sufficiently large $\ell$ it suffices to require that
\[ \sum_{i \in S}b_i \geq \max_{\pi \in \Delta(\mathcal{D})} \ell \cdot ( H(\bz'_S | \bz'_{S^c}) + \delta'(\eps) ) \]
for some function $\delta'$ that depends only on $\Pi$ so that $\delta'(\eps) \to 0$ as $\eps \to 0$.  This proves the claim.
\end{proof}
Finally, Claim~\ref{cl:B} implies that with high probability, (B) occurs.  Indeed, by a union bound over all $2^k$ choices for $S$, we conclude that
\begin{align*}
    \Pr[ \text{ (B) does not occur }] &= \Pr[ \text{ $\mathcal{E}_S$ occurs for some $S \subseteq [k]$ }] \\
    &\leq \sum_{S \subseteq [k]} \Pr[\mathcal{E}_S] \\
    &\leq 2^k \exp(-\Omega(\eps \ell)) \\
    &\leq \exp(-\Omega(\eps \ell)),
\end{align*}
using the fact that $k$ is a constant that depends only on the HSS $\Pi$.
Since Claim~\ref{cl:A} has established that (A) occurs with probability at least $1 - \exp(-\Omega(\eps^2 \ell))$, we conclude by another union bound that both (A) and (B) occur with probability at least $1 - \exp(-\Omega(\eps^2 \ell))$. As discussed above, if both (A) and (B) occur, then the algorithm $\Rec'$ is correct.  This completes the proof of the theorem.
\end{proof}
\subsection{Symmetrically Private HSS}\label{sec:symmetricHSS}

In HSS, the only security guarantee is that any set of $t$ colluding servers learn nothing about the input secrets.  We can also study the stronger notion of \emph{symmetric security}, where we additionally demand that the output client learn nothing beyond the desired output $f(x_1, \ldots, x_m)$. 

In this section, we begin an exploration of symmetric security that can beat the $1 - dt/k$ barrier.  First, we show that the Shamir-based non-linear HSS example at the beginning of Section~\ref{sec:nonlinear} is \emph{not} symmetrically secure; then we show that the Greedy Monomial CNF (Definition~\ref{def:greedy}) that is used in the proof of Corollary~\ref{cor:beatlinear} \emph{is} symmetrically secure.  We leave it as an intriguing open direction to further characterize which non-linear HSS schemes are symmetrically secure.

To begin, we first formally define symmetric security.

\begin{definition}[SHSS]
Let $\Pi = (\Share, \Eval, \Rec)$ be an HSS for $\mathcal{F}$ with inputs in $\mathcal{X}^m$.  We say that $\Pi$ is a \emph{symmetrically private} HSS (SHSS) if the following holds for all $f \in \mathcal{F}$ and all $\bx \in \mathcal{X}^m$.  Let  $y^{(1)}, \ldots, y^{(k)}$ denote the output shares of $\Pi$ (that is, the outputs of $\Eval$ given $f$).  Then the joint distribution of $y^{(1)}, \ldots, y^{(k)}$ depends only on $f(\bx)$.

\end{definition}

\begin{remark}[Relationship to SPIR]
	A related notion is that of \emph{symmetrically private information retrieval} (SPIR)~\cite{GIKM98}, where the (single) client only learns its requested record from the database and nothing else.  The notion of symmetric privacy in SPIR is stronger than the one we consider here in that it rules out additional information about the function $f$ in the joint distribution of {\em both input shares and output shares}. To meet this stronger requirement, the servers must inherently share a source of common randomness which is unknown to the client. Our weaker symmetric privacy notion considers the output shares alone. This is motivated by applications in which the output shares are delivered to an external output client who does not collude with any servers or input clients.
\end{remark}

We first note that schemes with linear reconstruction, such as our constructions from Sections \ref{sec:linear-constructions} and \ref{sec:blackbox}, can be trivially made into weak SHSS schemes by having some input client distribute among the $k$ servers a random sharing of $\vec{0}\in\mathbb{F}^\ell$ according to the $\ell$-LMSSS defined by the output shares of the HSS. Then, each servers adds its share of $\vec{0}$ to its output share. Unfortunately, this will result in a fresh sharing of the secret, which is not compressible, and thus we may not apply the Slepian-Wolf-like machinery from the previous section in order to beat the $1 - dt/k$ barrier.

Hence, it is interesting to ask whether we can construct HSS schemes which go beyond the $1-dt/k$ barrier \emph{and} are an SHSS. We begin by arguing that the usual Shamir-based HSS (instances of which can be compressed according to the warm-up in Section \ref{sec:nonlinear}) is not an SHSS. The fact the output shares don't encode a uniformly random polynomial (because it is necessarily reducible) is well known, see for example \cite{STOC:BenGolWig88,STOC:ChaCreDam88}. We prove the following proposition for completeness.

\begin{proposition}[Shamir-based HSS is not an SHSS]\label{shamir-not-shss}
Let $\FF$ be a finite field such that $|\FF|>3$. Let $\Pi$ be the $1$-private $3$-server HSS for $\mathcal{F}=\{f(x_1,x_2)\mapsto x_1x_2\}$ 
where $\Share$ is given by Shamir sharing (Definition~\ref{shamir}) and
where $\Eval$ simply computes the conversion in Lemma \ref{shamir-product}.  Then $\Pi$ is not an SHSS.
\end{proposition}
\begin{proof}
Let $\alpha_1,\alpha_2,\alpha_3$ be the evaluation points used in the HSS. Suppose first that we choose inputs $(x_1,x_2)=(0,0)$. For $i\in[2]$, let $p_i\in\FF[x]$ be the random polynomial of degree $1$ such that $p_i(0)=x_i=0$. After the servers apply $\Eval(f,\cdot,\cdot)$, the output shares constitute the polynomial $p=p_1\cdot p_2$, evaluated at points $\alpha_1,\alpha_2,\alpha_3$. Since the degree of the polynomial is at most $2$ we can fully recover it using $3$ evaluation points. Since $p_i(0)=0$ for every $i\in[2]$, the smallest (possibly) non-zero monomial in $p$ has degree $2$.

Now, suppose alternatively that we choose inputs $(x_1,x_2)=(0,1)$. Now the smallest (possibly) non-zero monomial in $p=p_1p_2$ has degree $1$. Since the random coefficient of this monomial is not always zero, the above distributions of $p$ are not identical, while $0\cdot 0=0\cdot 1$, and so this HSS is not an SHSS.
\end{proof}

On the flip side, as it turns out, the HSS that arises from the proof of Corollary~\ref{cor:beatlinear} is an SHSS.
Given Proposition~\ref{shamir-not-shss}, this is an additional benefit that the ``Greedy Monomial CNF'' scheme has over Shamir-sharing, aside from its better compressibility compared to Shamir (as in Remark \ref{rem:tight}). 


\begin{proposition}[Greedy-CNF-based HSS is an SHSS]\label{weakshss}
The HSS scheme $\Pi$ from Definition \ref{def:greedy} for $t=1,d=2,k=3$ is an SHSS. 
\end{proposition}

\begin{proof}
Suppose that we want to multiply $\alpha,\beta\in\FF$ using $\Pi$. To this end, suppose we additively share $\alpha=a_1+a_2+a_3$ and $\beta=b_1+b_2+b_3$ and, according to the $1$-CNF sharing, we give server with index $j\in[3]$ the shares $(a_i,b_i)_{i\neq j}$. Then, the servers apply their $\Eval$ algorithm and obtain output shares
\begin{align*}
        y_1 &= (a_2+a_3)(b_2+b_3) \\
        y_2 &= a_1b_1 + a_1b_3 + a_3b_1 \\
        y_3 &= a_1b_2 + a_2b_1
\end{align*}
where server with index $j$ outputs $y_j$. We need to show that the distribution of $(y_1,y_2,y_3)$ depends only on $a\cdot b$. To this end, let $\alpha,\alpha',\beta,\beta'\in\FF$ be such that $\alpha\beta=\alpha'\beta'$. In the same way $(y_1,y_2,y_3)$ depends on $(\alpha,\beta)$, suppose that $(y'_1,y'_2,y'_2)$ depends on $(\alpha',\beta')$. We need to prove that $(y_1,y_2,y_3)$ and $(y'_1,y'_2,y'_2)$ are identically distributed. Our strategy will be to find a bijective mapping $f_{\alpha,\alpha',\beta,\beta'}(a_1,a_2,b_1,b_2)=(a'_1,a'_2,b'_1,b'_2)$ such that if we replace $a_1,a_2,b_1,b_2$ by $a'_1,a'_2,b'_1,b'_2$ in the expression for $(y_1,y_2,y_2)$ we get $(y'_1,y'_2,y'_3)$. This can be shown to be equivalent to requiring that
\begin{align*}
    a_1b_2+a_2b_1&=a'_1b'_2+a'_2b'_1\\
    a_1\beta+b_1\alpha-a_1b_1&=a'_1\beta'+b'_1\alpha'-a'_1b'_1
\end{align*}
The rest of the proof proceeds by case analysis.
\paragraph{Case 1: $\alpha\beta=\alpha'\beta'\neq 0$.} We have that $\alpha,\beta,\alpha',\beta'\neq 0$ so we may choose
\begin{align*}
    a_1'&=a_1 (\beta')^{-1}\beta\\
    b_1'&=b_1 (\alpha')^{-1}\alpha\\
    a_2'&=a_2 \alpha'\alpha^{-1}\\
    b_2'&=b_2 \beta'\beta^{-1}
\end{align*}
which is a bijection.
\paragraph{Case 2: $\alpha\beta=\alpha'\beta'=0$.} If $\alpha=\alpha'=0$ and $\beta,\beta'\neq 0$ then we require
\begin{align*}
     a_1b_2+a_2b_1&=a'_1b'_2+a'_2b'_1\\
    a_1\beta-a_1b_1&=a'_1\beta'-a'_1b'_1
\end{align*}
for which we may use the mapping
\begin{align*}
   a_1'&=a_1 (\beta')^{-1}\beta\\
    b_1'&=b_1 (\beta')^{-1}\beta\\
    a_2'&=a_2 \beta'\beta^{-1}\\
    b_2'&=b_2 \beta'\beta^{-1}
\end{align*}
which is a bijection. The case $\beta=\beta'=0$ and $\alpha,\alpha'\neq 0$ is symmetric. If $\alpha=\alpha'=\beta'=0$ and $\beta\neq 0$ we require
\begin{align*}
    a_1b_2+a_2b_1&=a'_1b'_2+a'_2b'_1\\
    a_1\beta-a_1b_1&=-a'_1b'_1
\end{align*}
for which we may use the mapping
\begin{align*}
a_1' & =a_1\\
b_1' & =\begin{cases}
b_1+\beta, & a_1\neq0\\
b_1, & a_1=0
\end{cases}\\
a_2' & =a_2\\
b_2' & =\begin{cases}
b_2- \beta a_2(a_1)^{-1}, & a_1\neq0\\
b_2, & a_1=0
\end{cases}
\end{align*}
which is also a bijection. The case $\alpha'=\beta'=\beta=0$ and $\alpha\neq 0$ is symmetric. Also, if $\alpha=\beta=\alpha'=0$ and $\beta'\neq 0$ we may use the inverse
of the above mapping.
\end{proof}

Somewhat surprisingly, this pattern does not continue to $k=4$, as we show in the following example:

\begin{example} The HSS from Definition \ref{def:greedy} with $t=1,d=3,k=4$,$\FF=\FF_3$ is not an SHSS.
\end{example}
\begin{proof}
By direct calculation, when $(x_0,x_1,x_2)=(0,0,1)$ we get for the output shares $\Pr (0,0,0,0)=431/2187$, while when $(x_0,x_1,x_2)=(0,0,0)$ we get for the output shares $\Pr (0,0,0,0)=17/81$.

\end{proof}

\section*{Acknowledgements}
We thank Elette Boyle and Tsachy Weissman for helpful conversations and the ITCS reviewers for useful comments.

\bibliographystyle{alpha}
\bibliography{bibl}

\appendix

\section{Computationally Secure HSS}
\label{app:computational}

In this section we define the computational relaxation of HSS, adapting earlier definitions (see, e.g.,~\cite{BGILT18}) to our notation. 

Unlike the information-theoretic setting of Definition~\ref{def:HSS}, in the computational setting the input domain $\cal X$ and output domain $\cal Y$ are $\{0,1\}^*$ rather than finite sets. We further modify the syntax of Definition~\ref{def:HSS} in the following ways. 
\begin{itemize}
    \item The function $\Share$ takes a security parameter $\lambda$ as an additional input. 
    \item The function class $\cal F$ is replaced by a polynomial-time computable function $F({\hat f};x_1,\ldots,x_m)$, where ${\hat f}$ describes a function $f(x_1,\ldots,x_m)$ and is given as input to $\Eval$. For instance, private information retrieval can be captured by $F({\hat f};x_1)$ where ${\hat f}$ describes an $N$-symbol database and $x_1$ an index $i\in[N]$, and $F$ returns ${\hat f}[x_1]$. When referring to HSS for concrete computational models such as circuits or branching programs, the input ${\hat f}$ is a description of a circuit or a branching program with inputs $x_1,\ldots,x_m$. Finally, when considering {\em additive} HSS as in Definition~\ref{def:linear}, ${\hat f}$ also specifies the finite field over which the output is defined.
\end{itemize}

Security for computational HSS is defined in the following standard way.

\begin{definition}[Computational HSS: Security]\label{def:secure}
We say that $\Pi=(\Share,\Eval,\Rec)$ is  \emph{computationally $t$-private} if for every set of servers $T \subset [k]$ of size $t$ and polynomials $p_1,p_2$ the following holds. For all input sequences $x_\lambda,x'_\lambda$ such that $|x_\lambda|=|x'_\lambda|=p_1(\lambda)$,  circuit sequences $C_\lambda$  such that $|C_\lambda|=p_2(\lambda)$, and all sufficiently large $\lambda$, we have
\[ \Pr[C_\lambda(Y_T)=1]-\Pr[C_\lambda(Y'_T)=1] \le 1/p_2(\lambda),
\]
where $Y_T$ and $Y'_T$ are the $T$-entries of $\Share(1^\lambda,x_\lambda)$ and $\Share(1^\lambda,x'_\lambda)$, respectively.
\end{definition}

\section{Proof of Proposition~\ref{prop:impossibleShamir}}\label{app:impossibleShamir}
In this section we prove Proposition~\ref{prop:impossibleShamir}.


First, we show that there exists a CNF-based scheme.  By (the proof of) Theorem~\ref{thm:cnf-hss}, it suffices to exhibit a $5$-party $4$-LMSSS over $\FF_2$ with download cost $5$.  For secrets $x_1, x_2,x_3,x_4 \in \FF_2$ and a random bit $r \in \FF_2$, such a scheme is given by
\[ \Share(x_1, x_2, x_3, x_4, r) = ( r, x_1 + r, x_2 + r, x_3 + r, x_4 + r) \]
and
\[ \Rec(y_1, y_2, y_3, y_4, y_5) = (y_1 + y_2, y_1 + y_3, y_1 + y_4, y_1 + y_5).\]

Next, we show that there is no Shamir-based scheme that downloads only one bit from each party.  The proof proceeds by a computer search.  However, naively such a search (for example, over all sets of $k=5$ linear functions from $\FF_8^4 \to \FF_2$) is not computationally tractable.  Instead we first analytically reduce the problem to one that is tractable.


Suppose (towards a contradiction) that there were an $\F_2$-linear HSS-for-concatenation $\Pi$ based on Shamir-sharing.  
We note that since $\Rec$ is $\F_2$-linear, we may assume that $\Eval$ is also $\F_2$-linear.\footnote{Indeed, by the same argument in the proof of Theorem~\ref{thm:cnf-hss-d1}, our $\FF_2$-linear HSS-for-concatenation $\Pi = (\Share, \Eval, \Rec)$ gives rise to an $\FF_2$-linear $4$-LMSSS that has the output shares of $\Pi$ as its shares, with the same $\Rec$ function.  Since $\Rec$ is $\F_2$-linear, and since by \cite{beimelThesis} any LMSSS $\mathcal{L}$ with linear $\Rec$ may also have a $\F_2$-linear share function, we may assume that $\mathcal{L}$ has a $\F_2$-linear share function, $\Share'$.  But since $\Share$ (which is Shamir sharing) is also $\F_2$-linear, this implies that we may take $\Eval$ to be the $\F_2$-linear function $\Share' \circ \Share^{-1}$, where $\Share^{-1}$ denotes an arbitrary linear function that returns the inputs given all of the Shamir shares.}
We will write down $\Pi$ linear-algebraically over $\FF_2$.  Fix some basis for $\FF_8$ over $\FF_2$.  For any element $\alpha \in \F_8$, let $\nvec(\alpha) \in \F_2^3$ denote the vector representation of $\alpha$ in that basis, and let $\nmat(\alpha) \in \FF_2^{3 \times 3}$ denote the matrix representation of $\alpha$.  Thus, $\nmat(\alpha)\nvec(\beta) = \nvec(\alpha\beta)$ for all $\alpha, \beta \in \FF_8$.

Let $x^{(j)} \in \FF_2$ for $j \in [4]$ be the $\ell=4$ secrets.  Choose $\alpha_0, \alpha_1, \ldots, \alpha_5 \in \FF_8^*$, and suppose without loss of generality that $\alpha_0 = 0$.\footnote{Indeed, this is without loss of generality, as the Shamir scheme with general $\alpha_i$ has the same share function as the Shamir scheme with evaluation points $\alpha_i'$ where $\alpha_0' = 0$ and $\alpha_i' = \alpha_i - \alpha_0$. }
Under Shamir sharing, party $i$ receives
\[ (x^{(j)} + \rho^{(j)} \cdot \alpha_i )_{j \in [4]}, \]
where $\rho^{(j)} \in \FF_8$ is uniformly random.  Writing this over $\FF_2$, party $i$ receives
\[ [ \vec{v} | \nmat(\alpha_i) ] \cdot \begin{pmatrix} x^{(j)} \\ r_1^{(j)} \\ r_2^{(j)} \\ r_3^{(j)} \end{pmatrix} \qquad \text{for $j \in [4]$},\]
where $\vec{v} = \nvec(1) \in \FF_2^{3}$ is a column vector, and where $r_i^{(j)} \in \FF_2$ are uniformly random.  (That is, $\nvec(\rho^{(j)}) = \br^{(j)}$).  Let $W_i = [ \vec{v} | \nmat(\alpha_i) ]$ for all $i \in [5]$.

By assumption, each party $i$ sends a single bit, which must be an $\F_2$-linear combination of the bits that they hold.  Thus, each party $i$  sends
\[ z_i = \sum_{j=1}^4 \langle \bw_i^{(j)}, (x^{(j)}, r_1^{(j)}, r_2^{(j)}, r_3^{(j)} ) \rangle, \]
for some vectors $\bw_i^{(j)}$ in the rowspan of $W_i$.
Let 
\[ \bw_i = \bw_i^{(1)} \circ \bw_i^{(2)} \circ \bw_i^{(3)} \circ \bw_i^{(4)} \in \FF_2^{16}, \]
where $\circ$ denotes concatenation.
Since the recovery algorithm must also be linear, we have
\[ x^{(j)} = \sum_{i=1}^5 a_{ij} z_i \]
for some coefficients $a_{ij} \in \FF_2$ and for all $j \in [4]$.
This implies that for all $j \in [4]$,
\begin{equation}\label{eq:linalg}
 \sum_{i=1}^5 a_{ij} \bw_i = \vec{e}_{4j+1}, 
\end{equation}
where $\vec{e}_r \in \FF_2^{16}$ denotes the $r$'th standard basis vector.  (This is because if we concatenate the vectors $(x^{(j)}, r_1^{(j)}, r_2^{(j)}, r_3^{(j)} )$ to mirror the concatenation that created the $\bw_i$'s, the $x^{(j)}$ term appears in the $4j + 1$'st coordinate).  Consider the restriction $\by_i$ of $\bw_i$ to the coordinates indexed by $(2,3,4,6,7,8,10,11,12,14,15,16)$; that is, all of the elements of $[16]$ that are not equal to $4j + 1$.  Then, \eqref{eq:linalg} implies that for all $j \in [4]$, 
\[ \sum_{i=1}^5 a_{ij} \by_i = \mathbf{0}. \]
The $5 \times 4$ matrix formed by the $a_{ij}$ must be full rank, since the right hand sides of \eqref{eq:linalg} for $j \in [4]$ are linearly independent.  Thus, the matrix with the $\by_i$ as columns has a kernel of dimension at least $4$. We conclude that the $\by_i$ for $i \in [5]$ must have a span of dimension at most $1$.  In other words, without loss of generality we may assume that $\by_i = \by$ is independent of $i$.  (Notice that it's also possible that $\by_i = \vec{0}$, but in that case we could just make the corresponding $a_{ij}$ values zero and set $\by_i = \by$).  

Further, values for $\by_i$ determine all of the values in $\bw_i$.  This is because we have the restriction that each $\bw_i^{(j)}$ lies in the rowspan of $W_i$, and $\nmat(\alpha_i)$ is full rank.  This implies that the the first entry of $\bw_i^{(j)}$ is given by $\langle \by_i^{(j)} \nmat(\alpha_i)^{-1}, \vec{v} \rangle$ for all $i,j$.

Therefore, the values of $\bw_i$ for all $i \in [5]$ are determined by a single choice of $\by \in \F_2^{12}$. 
This suggests an algorithm to enumerate over all linear HSS with the claimed properties:
\begin{itemize}
    \item For each choice of distinct $\alpha_1, \ldots, \alpha_5 \in \FF_8^*$, and for each choice of $\by \in \FF_2^{4}$:
    \begin{itemize}
        \item Recover the vectors $\bw_i$ for $i \in [5]$ that are implied by $\by$, as described above.
        \item If it is the case that $\vec{e}_{3j + 1}$ lies in the span of $\{\bw_i \,:\, i \in [5]\}$ for all $j$, we have found a valid linear HSS; return it.
    \end{itemize}
    \item If we have not returned, return ``there is no such HSS.''
\end{itemize}
This search is tractable: there are only ${7 \choose 5} \cdot 2^{12}$ things to enumerate over, and for each of these we must do some linear algebra on a $16 \times 5$ matrix to determine if the choice of $\by$ results in a valid scheme for the choice of evaluation points.  We implemented this search, and did not find any valid schemes.\footnote{The code for our implementation of this search is available at \href{https://web.stanford.edu/~marykw/files/proof_of_no_Shamir_HSS.sage}{https://web.stanford.edu/\raisebox{0.5ex}{\texttildelow}marykw/files/proof\_of\_no\_Shamir\\\_HSS.sage}.}  Therefore we conclude that no such scheme exists.

\end{document}